\documentclass[12pt,a4paper]{article}
\usepackage{bbm}
\usepackage{amssymb,amsmath}
\usepackage{rotating}
\usepackage[centerlast]{caption}
\usepackage{graphicx}
\usepackage{verbatim}
\usepackage{enumerate}
\usepackage{mathrsfs}
\usepackage{amsfonts}
\usepackage{multirow}
\usepackage{amsthm,amscd,latexsym,color}
\usepackage{booktabs,tabularx}
\usepackage{caption}
\usepackage{subcaption}
\usepackage[title,titletoc]{appendix}
\usepackage{natbib}
\usepackage{amsmath}
\usepackage{amssymb}
\usepackage{mathtools}
\usepackage{bm}
\usepackage{nccmath}
\usepackage[title]{appendix}
\usepackage{tikz}
\usepackage[a4paper,text={17.0cm, 26.0cm},centering]{geometry}

\RequirePackage[colorlinks,citecolor=black,linkcolor=black,urlcolor=black,pagebackref]{hyperref}
\newtheorem{rem}{Remark}
\theoremstyle{definition}

\theoremstyle{plain}
\newtheorem{theo}{Theorem}
\theoremstyle{plain}

\theoremstyle{plain}
\newtheorem{cor}{Corollary}
\theoremstyle{plain}

\theoremstyle{plain}
\newtheorem{lemma}{Lemma}
\newtheorem{assumption}{Assumption}
\numberwithin{equation}{section} \numberwithin{theo}{section}
\numberwithin{defn}{section} \numberwithin{rem}{section}
\numberwithin{cor}{section} \numberwithin{lemma}{section}
\numberwithin{prop}{section} \numberwithin{assumption}{section}
\allowdisplaybreaks[4]
\newcommand{\D}{\displaystyle}
\newcommand{\DF}[2]{\frac{\D#1}{\D#2}}
\newcommand{\e}{{\mathbb E}}
\newcommand{\p}{{\mathbb P}}

\DeclareMathOperator*{\argmin}{\arg\!\min}

\DeclareMathOperator*{\sgn}{\normalfont\textrm{sgn}}
\DeclareMathOperator*{\diag}{\normalfont\textrm{diag}}

\DeclareMathOperator*{\VEC}{\normalfont\textrm{vec}}
\def\be{\begin{equation}}
\def\ee{\end{equation}}
\def\bestar{\begin{eqnarray*}}
\def\eestar{\end{eqnarray*}}
\def\bea{\begin{eqnarray}}
\def\eea{\end{eqnarray}}

\renewcommand {\arraystretch}{1.2}

\usepackage{xr-hyper}
\renewcommand{\arraystretch}{1.20}

\begin{document}

\begin{center}
{\Large \bf Identification, Estimation and Inference Based on \\Constructed Instruments}\footnote{An earlier version of this paper, titled ``Semiparametric Instrumental Variables Method,'' was presented at University of Cambridge, University of Amsterdam, University of Bologna, University College London, London School of Economics, University of Exeter, and University of Oxford during the second author's sabbatical leave in Europe and the UK between March and May 2026. More recently, an updated version, titled ``Identification, Estimation and Inference Based on Structural Error Projection,'' was presented at several international workshops held at Dongbei University of Finance and Economics, Peking University, Southwestern University of Finance and Economics, and Monash University between June and August 2026. We would like to thank many of the seminar and workshop participants for their constructive comments and suggestions.

Dong acknowledges support from the National Natural Science Foundation of China under Grant Numbers 72473156 and 72073143. Gao, Linton, and Peng acknowledge financial support from the Australian Research Council Discovery Grants Program under Grant Number DP250100063. We also thank Yuhuai Chen and Ruofan Xu for their computational assistance. An R--package is at \url{https://cran.r-project.org/package=SIVMethod}.}

\medskip

{\sc \small Chaohua Dong$^{\ast}$, Jiti Gao$^{\dag}$, Oliver Linton$^{\star}$ and Bin Peng$^{\dag}$}

{\small $^{\ast}$Zhongnan University of Economics and Law, $^{\dag}$Monash University and $^{\star}$University of Cambridge}
\end{center}

\begin{abstract}

We propose a structural projection (SP) method for identifying and estimating structural parameters in regression models associated with endogenous regressors.  The method projects the structural error onto the sigma-field generated by the regressors, decomposing it into a conditional mean component and an orthogonal residual. This decomposition yields an internally constructed instrument that is valid whenever the conditional mean lies in a known subspace, avoiding the need for external instruments. 

We establish identification conditions, derive consistency and asymptotic normality for the resulting estimator, and propose a LASSO--based selection procedure for choosing the subspace when no natural candidate is available. The SP method applies to a broad class of linear, nonlinear, and non-separable models. Simulations show that SP substantially outperforms least-squares-based benchmarks under endogeneity while remaining competitive when it is absent, and we illustrate the method using data on the returns to schooling.

{\em Keywords}: Error decomposition, Omitted variable, Orthonormal series expansion, Structural error projection

{\em JEL Classification}: C12, C13, C14; C26

\end{abstract}
\newpage

\section{Introduction}\label{Sec1}

Identifying structural parameters when regressors are correlated with
unobserved disturbances is a central problem in econometrics. Instrumental
variable (IV) methods provide a widely used solution, but they rely on the
availability of valid instruments that are correlated with the endogenous
regressors but uncorrelated with the structural errors. In many empirical
settings such instruments are difficult to obtain or justify. This paper
proposes to construct instruments directly from the observed regressors. The key observation is
that endogeneity arises from the component of the structural error that is
correlated with the regressors. By projecting the structural equation onto
the sigma--field generated by the regressors, we decompose the structural
error into a conditional expectation component and an orthogonal residual.
This decomposition yields an exogenous regression representation that allows
us to construct valid instruments as functions of the observed covariates.

Suppose that we have a linear structural model 
\begin{equation}
y=\mathbf{x}^{\top }\bm{\beta}_{0}+\varepsilon \ \ \
\mbox{with $\e[\varepsilon]
= 0$ \ but \ $\e[\mathbf{x} \, \varepsilon] \neq 0$},
\end{equation}%
where $\bm{\beta}_{0}$ is the parameter vector of interest and $\mathbf{x}$ is the
observed $d$--dimensional vector of variables related to the outcome $y$. The
main issue we face regarding identification and estimation is that $\e[\mathbf{x}\varepsilon ]\neq 0.$ We propose a structural error projection of the form: $\varepsilon ={\mathbb{E}}%
[\varepsilon |\mathbf{x}]+e,$ where $m(\mathbf{x)}\equiv: \e[\varepsilon |\mathbf{x}]$ is of unknown functional form and $\e[e|\mathbf{x}]=0$. We suppose that $m\in {\mathcal S},$ a subspace of the Hilbert space of square
integrable functions of $\mathbf{x}$. Letting $\mathcal{P}_{\mathcal S}(\mathbf{x})$ denote the projection
operator onto subspace $\mathcal{S}$, we construct an instrument: $\mathbf{z}=\mathbf{x}%
- \mathcal{P}_{\mathcal S}(\mathbf{x})$ and then apply the IV method to estimate the structural
parameters. Because $\mathbf{z}$ is by construction orthogonal to the space
containing $m(\mathbf{x})$, it follows that ${\mathbb{E}}[\mathbf{z} \, \varepsilon ]=0,$ so
that $\mathbf{z}$ serves as a valid instrument for $\mathbf{x}$.
Identification then follows from the condition that 
\begin{equation*}
\e[\mathbf{z} \, \mathbf{x}^{\top}] =\e[\left( \mathbf{x}- \mathcal{P}_{\mathcal S}(\mathbf{x})\right) \left( \mathbf{x}- \mathcal{P}_{\mathcal S}(\mathbf{x})\right)^{\top}] \equiv: \bm{\Sigma}_x>{\bf 0},
\end{equation*}%
under which the structural parameter satisfies $\bm{\beta} _{0}=\bm{\Sigma}_{{x}}^{-1} \bm{\Sigma} _{{x} \, y}$ with $\bm{\Sigma} _{{x} \, y}=\e[\left( \mathbf{x}- \mathcal{P}_{\mathcal S}(\mathbf{%
x})\right) y]$. 

The first condition $m(\cdot) \in {\mathcal S}$ specifies the admissible sources
of confounding. The second condition $\bm{\Sigma} _{{x}}>{\bf 0}$ requires sufficient
variation in the structural regressors after those sources have been
removed. We then develop estimation technology and inference methods for the
structural parameters. A key consideration is the choice of $\mathcal{S}$, which could
be justified by economic, institutional, or design information. Plausible
examples include: (a) a known collection of fixed effects or seasonal
components; (b) a symmetry or invariance restriction; (c) specified
interactions that exclude the structural linear index; (d) a
finite-dimensional nuisance space motivated by the empirical design. We
explore some examples in Section \ref{Sec2}.1 below. In the absence of such concrete choices, we
propose to use the LASSO method applied to a dictionary of basis functions
to parsimoniously represent this space.

Our approach is related to several strands of the literature that develop
identification strategies derived from instruments internally generated from
features of the data generating process. \cite{RR2003} shows that
heteroskedasticity across regimes can identify structural parameters in
simultaneous equations models even in the absence of conventional
instruments, see also \cite{sp2001}. In that framework,
changes in the variance of structural shocks generate additional moment
conditions that identify the structural coefficients. \cite{AL2012} proposes
another approach for simultaneous equations based on heteroskedasticity,
showing that covariance restrictions between regressors and structural
disturbances can generate valid instruments even when external instruments
are unavailable. These methods exploit variation in second moments of the data generating process to produce internal instruments. More recently, \cite{lewis2025} points to a related literature in macroeconometrics that exploits non-Gaussianity in the structural shocks to obtain identification.

The approach developed in this paper follows a similar philosophy but relies
on a different source of identification. Instead of exploiting
heteroskedasticity or regime variation, we use a projection-based
decomposition of the conditional mean of the structural error. This
decomposition allows us to construct instrumental variables directly as
functions of the regressors. Identification therefore arises from the
structure of the conditional expectation rather than from variance shifts or
heteroskedasticity--based moment conditions.

Meanwhile, our approach is also related to the control function literature discussed, for example, in \cite{npv1999}, \cite{np2003}, \cite{blundell2004}, \cite{bck2007}, \cite{IN2009}, \cite{newey2013}, \cite{JMW2015}, and \cite{ns2021}. These existing methods address endogeneity in several models by introducing auxiliary instruments. By contrast, the proposed SP method constructs instruments internally through an orthogonal decomposition of the structural error term. In each case, this yields a closed--form instrumental variable, which we call an internally constructed instrumental variable (CIV). It depends only on observed regressors and permits identification in a wide range of linear and nonlinear models. Our implementation of the CIV system is based on series expansion and estimation methods, as recently discussed in \cite{dg2025}, for example.  

In summary, the main contributions of this paper are given as follows.
\smallskip

\noindent{\bf a}: \, We propose a structural projection (SP) method to derive the so--called ``CIV'' system to identify and estimate the unknown parameters and functions in a class of identified models;

\noindent{\bf b}: \, We show that the SP method offers an alternative to existing methods for addressing endogeneity in certain classes of weakly identified models; 

\noindent{\bf c}: \, We show that the proposed SP method is invariant to the degree of endogeneity, including a wide range of weak endogeneity, in some classes of nonlinear models;

\noindent{\bf d}: \, We establish estimation theory showing that the proposed SP method produces consistent and asymptotically normal estimators;

\noindent{\bf e}: \, We develop a Hausman--type nonparametric statistic for testing full exogeneity against a wide range of weak endogeneity at a near--optimal rate; and

\noindent{\bf f}: \, We propose a simple LASSO selection method and implement the estimation and testing theory using simulated and real-data examples, associated with an R--package available at \url{https://cran.r-project.org/package=SIVMethod}.

The remainder of the paper is organized as follows. Section \ref{Sec2}.1 provides economic examples illustrating the applicability of the SP method. Section \ref{Sec2}.2 establishes parameter identification conditions for a class of linear models under internally constructed instruments. Section \ref{Sec3} develops the corresponding estimation and testing theory for such correctly identified linear models. Section \ref{Sec3} then shows that the proposed SP method is applicable to deal with certain weak endogeneity settings, which includes the case where $(\mathbf{x}, \varepsilon)$ follows joint Gaussianity. Section \ref{Sec4} proposes a LASSO method for selecting the orthonormal series spanning $\mathcal{S}$. Section \ref{Sec5} revisits the return to schooling case study. Section \ref{Sec6} reports simulation evidence for the estimation and testing methods. Section \ref{Sec7} concludes and discusses how the method applies to several classes of nonlinear and non--separable models.

In the supplementary document, Appendix A provides some additional interpretations and justifications for the main assumptions listed in Sections 2--4 and then the simulation settings considered in Section 6. Appendix B includes some additional simulation scenarios to support that the proposed SP method works well numerically without imposing orthogonality and employing the LASSO method. Appendix B also further discusses potential source of endogeneity for the return to schooling case study considered in Section \ref{Sec5}. In addition, Appendix B finally shows that the SP method is applicable time--series data by including a market--return analysis. The proofs of the main results established in Sections 3 and 4 are given in Appendix C.

For notational consistency, we use $\mathbf{x}$ to denote a vector, regardless of whether it is stochastic or deterministic, and $x$ to denote either a univariate variable or a given point. When there is no ambiguity, $x$ may also be used as a generic variable of integration. 

We use $\mathbf{A}$ to denote either a matrix or a vector, $\|\mathbf{A}\|^2$ for the squared Euclidean norm, $a$ for a real number, and $[a]$ for the largest integer not exceeding $a$. Finally, $\rightarrow_{\mathcal{D}}$ and $\rightarrow_P$ denote convergence in distribution and convergence in probability, respectively. 

\section{Linear Model Identification}\label{Sec2}

\renewcommand{\theequation}{2.\arabic{equation}} 
\setcounter{equation}{0}

\subsection{\normalsize Economic interpretation and ex ante specification of the nuisance space}

The identifying content of the model comes from the specification of the
nuisance space $\mathcal{S}$. The condition $m(\mathbf{x}) \in \mathcal{S}$ should be
interpreted as a restriction on the mechanisms through which the endogenous
component of the disturbance may depend on the observed state. An ex ante
specification of $\mathcal{S}$ should satisfy three principles. 

First, its defining features should be motivated by economic, institutional, or design
information available independently of the outcomes. Second, the restriction should have substantive content: it must exclude at least some functions of
the observed state that could otherwise be confused with the structural
index $\mathbf{x}^{\top} \bm{\beta }_{0}$.  Third, the restriction must leave enough residual variation in $\mathbf{x}$ for the covariance of $\mathbf{z}=\mathbf{x} -{\mathcal P}_{\mathcal S}(\mathbf{x})$ to be nonsingular. Ex ante
specification does not require every element of the projection operator to
be known before seeing the sample. 

The subspace $\mathcal{S}$ and its inner products are constructed from the observed regressors.
For example, the researcher may specify in advance that $\mathcal{S}$ is
generated by a set of interactions, calendar functions, or invariant
functions, and then use the empirical distribution to orthonormalize those
functions or estimate $\mathcal{P} _{\mathcal{S}}(\mathbf{x})$. Because these operations use only the regressors, they implement the maintained restriction rather than selecting
it according to its implications for the outcome. We give several illustrative examples.
\medskip

\noindent\textbf{Example 2.1.1}:\, Group and Time Effects. Suppose that observations are
indexed by unit $i$ and period $t$, and consider the structural model $%
y_{it}=\mathbf{x}_{it}^{\top}\bm{\beta} _{0}+\varepsilon _{it}.$  Institutional knowledge may suggest that the conditional mean of the disturbance consists
of a unit-specific component and a common time component $\e[\varepsilon
_{it}\mid x_{it}]=\alpha _{i}+\lambda _{t}.$ 

The nuisance space is then given by
\begin{equation}
\mathcal{S}={\rm span}\left\{ 1\{i=j\},\,1\{t=s\}:j=1,\ldots,N,\;s=1,\ldots ,T\right\} ,
\end{equation}%
subject to the usual normalizations. This space is specified by the panel
structure rather than selected from the outcome data. 

The residualized regressor is the familiar within-transformed variable, $\mathbf{z}_{it}=\mathbf{x}_{it}- \mathbf{P}_{%
\mathcal{S}}(\mathbf{x}_{it})$, and identification requires nondegenerate variation in $\mathbf{x}_{it}$ after removing the unit and time components: $\e[\mathbf{z}_{it} \, \mathbf{z}_{it}^{\top}]>{\bf 0}$. The economic restriction is not simply that unit and time
effects should be included because they improve fit. It is that systematic dependence between the disturbance and the observed state
operates via permanent unit heterogeneity and aggregate time shocks. If the conditional mean disturbance contains unit-specific
trends that are correlated with $\mathbf{x}_{it}$, the proposed space is
misspecified. Such departures can be accommodated by enlarging the nuisance space within the structural error projection framework developed in the paper.
\smallskip

\noindent\textbf{Example 2.1.2}: \, Symmetry and Invariance Restrictions. In some cases, economic reasoning implies that the confounding component is invariant under a known transformation, whereas the structural component is not. 

Suppose first that $x$ is scalar and that the selection mechanism depends on the magnitude of $x$, but not on its sign: $m(x)=m(-x).$ Then one may take 
\begin{equation}
\mathcal{S}_{\mathrm{even}}=\left\{ s\in L^{2}(F(x)): \, s(x)=s(-x)\right\} .
\end{equation}%
The structural function $x \, \beta _{0}$ changes sign under $x\mapsto -x$,
whereas every element of $\mathcal{S}_{\mathrm{even}}$ is invariant. Hence, ${\rm span}\{x\}\cap \mathcal{S}_{\mathrm{even}}=\{0\},$ provided the
support contains sufficiently rich positive and negative values.  The substantive interpretation could be that participation, selection, or
measurement error depends on the absolute magnitude of an exposure, while
the structural effect depends on its direction. 

For example, monitoring
intensity might respond to the absolute size of a position, whereas the
structural return depends on whether the position is long or short. Note
that the distribution of $x$ need not itself be symmetric. If it is
asymmetric, $\mathcal{P}_{\mathcal{S}_{\mathrm{even}}} \, x$ need not be zero, so the
residualization is nontrivial and $\e[x \, m(x)]$ may be nonzero. Identification nevertheless follows from the separation between the even
nuisance functions and the odd structural index. A multivariate version is
based on permutation invariance. 

Let $x=x_{1}-x_{2},$ and suppose that $m(x_{1},x_{2})=m(x_{2},x_{1}).$ The nuisance depends only on the unordered pair of characteristics, while the structural regressor is
directional and changes sign when the two components are exchanged. The relevant space is $\mathcal{S}_{\mathrm{sym}}=\left\{s:s(x_{1},x_{2})=s(x_{2},x_{1})\right\} .$ This structure may arise in
matched-pair, bilateral, ranking, or relative-performance settings in which common selection depends symmetrically on the characteristics of the pair, but the parameter of interest multiplies a directional difference.  

More generally, if a group of transformations $\mathcal{G}$ acts on the observed state, one may let $\mathcal{S}$ consist of functions invariant under that
group: $\mathcal{S}_{\mathcal{G}}=\left\{ s:s(g(x))=s(x)\text{ for every} \ g\in  \mathcal{G}\right\} .$ Identification is possible when the structural
regressors contain components that transform differently from the invariant nuisance functions.
\smallskip

\noindent\textbf{Example 2.1.3}:\, Specified Nonlinearities and Interactions. Economic
theory may imply that the omitted component operates through known nonlinear
features of the observed state but does not contain an unrestricted additive
linear component. For example, with two centred regressors, suppose that $%
m_{0}(x_{1},x_{2})=\gamma _{1}q_{1}(x)+\gamma _{2}q_{2}(x)+\gamma
_{3}q_{3}(x),$ where $q_{1}(x)=x_{1}x_{2},$ $q_{2}(x)=x_{1}^{2}-E[x_{1}^{2}],
$ and $q_{3}(x)=x_{2}^{2}-E[x_{2}^{2}].$ Then the nuisance space is $\mathcal{S}={\rm span}\{q_{1},q_{2},q_{3}\},$ whereas the structural index is $%
x_{1}\beta _{01}+x_{2}\beta _{02}.$ 

A possible economic interpretation is
that the omitted state reflects complementarities, congestion, dispersion,
or adjustment costs that depend on products and squared deviations, while
the coefficients of interest measure additive marginal effects. 

In a
production setting, for example, an omitted utilization component might be
known to respond to the interaction of capital and labour or to deviations
from normal operating levels, rather than to unrestricted linear
combinations of the inputs. 

The basis functions may be orthogonalized with respect to the distribution of $x$. Such orthogonalization can be estimated
using the regressors alone and does not change the substantive restriction.
What would not be legitimate is to search over a large dictionary using $y$,
retain the nonlinear terms that best fit the outcome, and then interpret the
resulting decomposition as population identification. 

This example also illustrates the strength of the assumption. The claim that confounding
contains interactions but no linear component must come from the economic
model. 
\smallskip

\noindent\textbf{Example 2.1.4}: \, Design--Based Strata and Known Assignment Rules. Suppose
treatment or exposure varies within strata determined by a known assignment rule. 

Let $A$ be a predetermined score or collection of baseline
characteristics, and suppose the conditional mean disturbance is constant,
or follows a specified low-dimensional form, within predefined strata: $%
m(A)=\sum_{j=1}^{J}\alpha _{j}\,1\{A\in \mathcal{A}_{j}\}.$ The sets $%
\mathcal{A}_{1},\ldots ,\mathcal{A}_{J}$ might be administrative eligibility
groups, sampling strata, geographic zones, cohorts, or categories fixed by
the treatment-assignment process. Then $\mathcal{S}={\rm span}\left\{{\rm I}\{A\in \mathcal{A}_{j}\}:j=1,\ldots ,J\right\} .$ The coefficient is
identified by variation in $x$ within these predetermined strata: $z=x-\e[x\mid A]$. The identifying assumption is that all systematic selection on unobservables is captured by
the known strata. The strata should therefore be determined by the design or
by institutional rules, rather than chosen after examining the relationship
between $x$ and \thinspace $y$. 

A related case arises when an administrative
assignment algorithm is known to depend on a finite set of basis functions $%
b_{1}(A),\ldots ,b_{K}(A)$. If economic reasoning implies that the
conditional mean of the disturbance depends on the assignment variables
through the same features, one may specify $\mathcal{S}={\rm span}\{b_{1}(A),\ldots ,b_{K}(A)\}.$
\smallskip

\noindent\textbf{Example 2.1.5}: \, Low--Rank Institutional, Network, or Spatial Components.
In some applications, the institutional structure suggests that confounding
is generated by a small number of common components. Suppose, for example,
that observations belong to known markets or network communities and that $%
m(x_{i})=\sum_{k=1}^{K}\gamma _{k}f_{k}(x_{i}),$ where the functions $%
f_{k}$ are fixed before the outcome analysis.  They might consist of market
indicators, predetermined network eigenvectors, geographic basis functions,
or exposure to known aggregate shocks. Then {\small $\mathcal{S}={\rm span}\{f_{1},\ldots ,f_{K}\}$}. Identification comes from the component of the
structural regressor that is not explained by these common institutional
factors. 

\subsection{Parameter Identification}

We return to the general structural linear model:
\begin{equation}\label{a1}
y=\mathbf{x}^{\top} \bm{\beta}_0 +\varepsilon \ \ \mbox{with $\e[\varepsilon] = 0$ \ but \ $\e[\mathbf{x} \, \varepsilon] \neq 0$}.
\end{equation}

Note that $\e\left[\mathbf{x} \, \e\left(\varepsilon|\mathbf{x}\right)\right]= \e[\mathbf{x}\, \varepsilon] \ne 0$ implies $\e[\varepsilon|\mathbf{x}]\ne 0$, and recall $e=\varepsilon- \e[\varepsilon|\mathbf{x}] = \varepsilon - m(\mathbf{x})$ with $\e[e|\mathbf{x}]=0$. To explain the main ideas, we focus on the following model:
\be
y = \mathbf{x}^{\top} \bm{\beta}_0 + \varepsilon \ \ \ \mbox{with} \ \ \ \varepsilon = \e[\varepsilon|\mathbf{x}] + e =m(\mathbf{x})+e.
\label{a4}
\ee

In the rest of Section \ref{Sec2} and in Sections \ref{Sec3}--\ref{Sec6}, we consider the case where $d$, the dimension of $\mathbf{x}$, is finite and fixed. Appendix B.1.1 of the supplementary document discusses the choice of series functions when $d$ is large but remains finite and fixed.

Before we discuss how to identify $\bm{\beta}_0$, we make the following comments.

\begin{rem}

(i) \, Model (\ref{a4}) is quite different from the partially linear model: 
\be
y= \mathbf{x}_u^{\top} \bm{\alpha}_0 + q(\mathbf{x}_v) + \varepsilon
\label{plm1}
\ee 
as discussed in the relevant literature by \cite{pmr1988}, \cite{linton1995}, \cite{hlg2000}, and many other more recent developments, where $\mathbf{x}_u$ and $\mathbf{x}_v$ are two different vectors of observed variables, and the following  condition:
\be
\e\left[(\mathbf{x}_u - \e[\mathbf{x}_u|\mathbf{x}_v]) (\mathbf{x}_u - \e[\mathbf{x}_u|\mathbf{x}_v])^{\top}\right]>{\bf 0}
\label{iden1}
\ee
is required for $\bm{\alpha}_0$ to be identifiable. However, Condition (\ref{iden1}) is violated for model (\ref{a4}) in the case of $\mathbf{x}_u=\mathbf{x}_v = \mathbf{x}$.

(ii) In the rest of Section \ref{Sec2}.2, we show how to construct an IV regression of the form:
\bea
y & = &  \mathbf{x}^{\top} \bm{\beta}_0 + \varepsilon \ \ \ \mbox{and} \ \ 
\label{a6}\\
\mathbf{z} & = & \mathbf{x} - \e\left[\mathbf{x} \, \bm{\Psi}^{\top}(\mathbf{x})\right] \, \bm{\Psi}(\mathbf{x}),
\label{a7}
\eea
where $\bm{\Psi}(\mathbf{x}) \equiv: \left(\psi_1(\mathbf{x}), \cdots, \psi_{\infty}(\mathbf{x})\right)^{\top}$, in which $\{\psi_j(\cdot): j\geq 1\}$ is an orthonormal sequence chosen to expand $m(\mathbf{x})$ by $m(\mathbf{x}) = \sum_{j=1}^{\infty} \psi_j(\mathbf{x}) \, \gamma_j$ with $\gamma_j = \e[m(\mathbf{x}) \, \psi_j(\mathbf{x})]$.

Equations (\ref{civ1}) and (\ref{civ2}) below show that $\mathbf{z}$ satisfies
\be
\e[\mathbf{z} \, \varepsilon] = {\bf 0} \ \ \ \mbox{and} \ \ \ \e\left[\mathbf{z} \, \mathbf{x}^{\top}\right] > {\bf 0},
\label{a8}
\ee
under $\e[\psi_j(\mathbf{x})]={\bf 0}$ and the orthonormality of $\{\psi_j(\mathbf{x}): j\geq 1\}$.
\smallskip

Equations (\ref{a7}) and (\ref{a8}) imply that $\mathbf{z}$ is an internally constructed instrumental variable (CIV), and hence a {\em valid} IV for model (\ref{a6}).

\end{rem}

To present the main ideas about model identification, we introduce a Hilbert space. Let ${\mathcal L}^2=\{g({\bf x}): \ \e[g^2({\bf x})]<\infty\}$, where the inner product is given by $\langle g_1({\bf x}), g_2({\bf x})\rangle=\e[g_1({\bf x})g_2({\bf x})]$ from which the induced norm is $\|g({\bf x})\|=\sqrt{\e[g^2({\bf x})]}$ for any given $g_1({\bf x}), g_2({\bf x}), g({\bf x})\in {\mathcal L}^2$. Equipped with this $\|\cdot\|$, ${\mathcal L}^2$ becomes a Hilbert space. 

\begin{assumption}\label{MTP1}
Suppose that (i) for any $\lambda\in \mathbb{R}^d$, $\lambda^\top {\bf x}\in {\mathcal L}^2$; (ii) $\e[ \mathbf{x} \, m(\mathbf{x})]\neq 0$ and $\e[m({\bf x})]=0$; (iii) there is a proper closed subspace ${\mathcal S}\subsetneqq {\mathcal L}^2$ such that $m(\mathbf{x})\in {\mathcal S}$ but $\lambda^\top {\bf x}\not\in {\mathcal S}$ for any $\lambda\ne 0$. 
\end{assumption}

Appendix A.1 of the supplementary document provides a geometric illustration of Assumption \ref{MTP1}. Appendix A.2 shows how to deal with the setting that the distributional structure of $\mathbf{x}$ is unknown in practice. 

\begin{rem}

(a) \, Condition (i) is a minimum requirement for ${\bf x}$ in the model. We also require $\e[(\lambda^\top {\bf x})^2]>0$ for any $\lambda\ne 0$, or equivalently, $\e[{\bf x}{\bf x}^\top]>{\bf 0}$. Condition (ii) ensures the presence of endogeneity in model \eqref{a1}, since $\e[\varepsilon \, \mathbf{x}] = \e[ \mathbf{x} \, m(\mathbf{x})]\neq 0$, while $\e[m({\bf x})]=\e[\varepsilon]=0$. Condition (iii) is crucial because the proper subspace ${\mathcal S}$ separates the linear form $\lambda^\top {\bf x}$ from the unknown function $m(\mathbf{x})$.

(b) \, In practice, the choice of $\mathcal{S}$ depends on the individual problem under study. Sections \ref{Sec5}, \ref{Sec6} and Appendix B show that a finite number of basis functions, such as a subset of trigonometric or Hermite series or a mixture of both, is only needed to construct a basis for $\mathcal{S}$.

(c) \, In general, Section \ref{Sec4} proposes a LASSO selection method for choosing ${\mathcal S}$, which has sound theory and is also readily applicable and implementable through using an R--package made publicly available, as demonstrated in Sections \ref{Sec5} and \ref{Sec6}.

\end{rem}

Given the closed subspace ${\mathcal S}$, ${\mathcal L}^2$ has the orthogonal decomposition ${\mathcal L}^2={\mathcal S}\oplus{\mathcal S}^\bot$, where ${\mathcal S}^\bot$ is the orthogonal complement of ${\mathcal S}$. Thus, every element $\xi \in {\mathcal L}^2$ has a unique decomposition $\xi=\xi_1+\xi_2$, where $\xi_1\in{\mathcal S}$, $\xi_2\in {\mathcal S}^\bot$, and $\e[\xi_1\xi_2]=0$. Accordingly, the decomposition defines two projection mappings, ${\mathcal P}_{\mathcal S}$ and ${\mathcal M}_{\mathcal S}$, from ${\mathcal L}^2$ into ${\mathcal S}$ and ${\mathcal S}^\bot$, respectively: ${\mathcal P}_{\mathcal S}(\xi)=\xi_1$ and ${\mathcal M}_{\mathcal S}(\xi)=\xi_2$.  We are now ready to establish the following important lemma.

\begin{lemma}\label{lemma2.1}
In addition to Assumption \ref{MTP1}, suppose that $\{\psi_j({\bf x}), j\ge 1\}$ is an orthonormal basis of the subspace ${\mathcal S}$ satisfying $\e[\psi_j({\bf x})\psi_{\ell}({\bf x})]=\delta_{j\ell}$ and $\e[\psi_j({\bf x})]=0$ for all $(j,\ell)$. Then: (1) ${\mathcal P}_{\mathcal S}(\xi)=\sum_{j=1}^\infty \e[\xi \psi_j({\bf x})] \psi_j({\bf x})$ for any $\xi$ with $\e[\xi^2]<\infty$; (2) ${\bm\Sigma}_x\equiv\e[{\bf x}{\bf x}^\top]-\sum_{j=1}^\infty \e[{\bf x} \psi_j({\bf x})] \e[{\bf x}^\top\psi_j({\bf x})]>{\bf 0}$; and (3) $\bm{\beta}_0$ is uniquely identified.
\end{lemma}

Lemma \ref{lemma2.1} imposes orthonormality with respect to a probability distribution. For notational simplicity in the rest of this section and in Appendices A and C, we maintain that $\{\psi_j(\cdot): j\geq 1\}$ is a sequence of orthonormal functions. The requirement $\e[\psi_j({\bf x})]=0$ ensures that ${\mathcal S}$ is consistent with $\e[m({\bf x})]=\e[\varepsilon]=0$.

Appendix A.2 of the supplementary document discusses how to construct such a probability distribution in theory and practice. As also discussed there, the orthonormality condition with respect to the distribution of $\mathbf{x}$ can be relaxed. This is relevant when there is no prior knowledge of the distributional structure of the data. Appendix B shows that the proposed SP method performs well numerically in finite--sample simulations without imposing orthogonality on $\{\psi_j(\cdot): j\geq 1\}$ and without employing the LASSO method.

As a closed subspace of ${\mathcal L}^2$, ${\mathcal S}$, equipped with the norm in ${\mathcal L}^2$, is also a Hilbert space. After specifying the orthonormal basis, we introduce ${\mathcal P}_{\mathcal S}$ and define ${\mathcal M}_{\mathcal S}=\mathcal{I}-{\mathcal P}_{\mathcal S}$, where $\mathcal{I}$ is the identity operator. Importantly, the positive definiteness of ${\bm\Sigma}_x=\e[{\mathcal M}_{\mathcal S}({\bf x}){\mathcal M}_{\mathcal S}({\bf x})^\top]$ is the cornerstone of the identification of $\bm\beta_0$. 

Because ${\mathcal M}_{\mathcal S}({\bf x})\in {\mathcal S}^\bot$, it is orthogonal to every element of ${\mathcal S}$, in particular $m({\bf x})$. Multiplying both sides of equation \eqref{a4} by ${\mathcal M}_{\mathcal S}({\bf x})$ and taking expectations gives
\bea
\e[{\mathcal M}_{\mathcal S}({\bf x})y]&=& \e[{\mathcal M}_{\mathcal S}({\bf x})\mathbf{x}^{\top}] \bm{\beta}_0=\e[{\mathcal M}_{\mathcal S}({\bf x}){\mathcal M}_{\mathcal S}(\mathbf{x}^{\top})] \, \bm{\beta}_0={\bm\Sigma}_x \, \bm{\beta}_0,
\label{a5cd}
\eea
which yields the following expression of $\bm{\beta}_0$ in population:
{\small
\be
\bm{\beta}_0 = {\bm\Sigma}_x^{-1} \e[{\mathcal M}_{\mathcal S}({\bf x}) \, y] ={\bm\Sigma}_x^{-1}\, \left(\e[{\bf x} \, y]-\sum_{j=1}^\infty \e[{\bf x} \, \psi_j({\bf x})]\,  \e[y \, \psi_j({\bf x})]\right) \equiv: {\bm\Sigma}_x^{-1} \, {\bm\Sigma}_{xy}.
 \label{a6cd}
 \ee}

Note that ${\mathcal M}_{\mathcal S}({\bf x})$ is exactly the same as $\mathbf{z}$ of (\ref{a7}) for the identification of $\bm{\beta}_0$ and satisfies the necessary conditions in the construction of an IV as follows:
\bea
\e[{\mathcal M}_{\mathcal S}({\bf x}){\bf x}^\top] & = & \e [{\mathcal M}_{\mathcal S}({\bf x}){\mathcal M}_{\mathcal S}^{\top}({\bf x})]>{\bf 0} \ \ \mbox{and} 
\label{civ1}\\ 
\e[{\mathcal M}_{\mathcal S}({\bf x})\varepsilon] & = & \e [{\mathcal M}_{\mathcal S}({\bf x})\, m({\bf x})]+\e [{\mathcal M}_{\mathcal S}({\bf x})\,e]={\bf 0}.
\label{civ2}
\eea

\begin{lemma}\label{lemma2.2}
The expression of $\bm{\beta}_0$ in \eqref{a6cd} is invariant to the choice of the basis in ${\mathcal S}$.
\end{lemma}

To conclude this section, we give the proofs of Lemmas \ref{lemma2.1} and \ref{lemma2.2} below.

\begin{proof}[Proof of Lemma \ref{lemma2.1}]
(1) Notice that as an element in ${\mathcal S}$, ${\mathcal P}_{\mathcal S}(\xi)=\sum_{j=1}^\infty c_j \psi_j({\bf x})$ is the projection of $\xi$ on ${\mathcal S}$ if and only if $\e[(\xi - \mathcal{P}_{\mathcal S}(\xi)) \, \eta]=0$ for any $\eta \in \mathcal S$, and furthermore if and only if $\e[(\xi - \mathcal{P}_{\mathcal S}(\xi)) \, \psi_j(\mathbf{x})]=0$ for all $j\geq 1$ (see, Theorem 11.1 of \cite{vdv1998}, for example). Therefore, $c_j =\e[\xi \, \psi_j(\mathbf{x})]$.

(2) By Assumption 2.1(iii), for any $\lambda\ne 0$, $\lambda^\top {\bf x}\not\in {\mathcal S}$. Therefore, ${\mathcal M}_{\mathcal S}(\lambda^\top {\bf x})=\lambda^\top {\bf x}-{\mathcal P}_{\mathcal S}(\lambda^\top {\bf x})\ne 0$. The proof of Lemma \ref{lemma2.1} then follows from {\small $${0}< \|{\mathcal M}_{\mathcal S}(\lambda^\top {\bf x})\|^2=\e[(\lambda^\top {\bf x}-{\mathcal P}_{\mathcal S}(\lambda^\top {\bf x}))^2] = \lambda^\top \left\{\e[{\bf x}{\bf x}^\top]-\sum_{j=1}^\infty \e[{\bf x} \psi_j({\bf x})] \e[{\bf x}^\top\psi_j({\bf x})]  \right\}\lambda.$$} 

(3) Let ${\bf b}$ be a generic vector with the same dimension as $\bm\beta_0$, and define a function $L({\bf b})=\e[{\mathcal M}_{\mathcal S}(y-{\bf x}^\top {\bf b})]^2$. Then, we have
\begin{align}
  L({\bf b})=&\e[{\mathcal M}_{\mathcal S}(\varepsilon+{\bf x}^\top (\bm\beta_0-{\bf b}))]^2=\e[e+{\mathcal M}_{\mathcal S}({\bf x})^\top (\bm\beta_0-{\bf b})]^2\notag\\
  =&\e[e^2]+(\bm\beta_0-{\bf b})^\top {\bm\Sigma}_x(\bm\beta_0-{\bf b})\ge \e[e^2]
\end{align}
where the equality holds if and only if ${\bf b}=\bm\beta_0$. Thus, $\bm\beta_0$ is identified as the minimum point of $L({\bf b})$, and the proof of Lemma \ref{lemma2.1} is completed.
\end{proof}

\begin{proof}[Proof of Lemma \ref{lemma2.2}]
We now show that $\bm{\beta}_0$ is invariant to the choice of the basis $\{\psi_j(\cdot): j\geq 1\}$ in ${\mathcal S}$. Towards this end, suppose that ${\mathcal S}$ can also be spanned by another orthonormal sequence $\{\widetilde{\psi}_j(\cdot): j\geq 1\}$. Denote two infinite-dimensional vectors $\Psi({\bf x})=(\psi_1({\bf x}), \psi_2({\bf x}), \cdots)^\top$ and $\widetilde{\Psi}({\bf x})=(\widetilde{\psi}_1({\bf x}), \widetilde{\psi}_2({\bf x}), \cdots)^\top$. 

\noindent Thus, $\widetilde{\Psi}({\bf x})=\mathbf{A}\Psi({\bf x}),$ where $\mathbf{A}$ is an infinite-dimensional matrix with element $a_{ij}=\e[\widetilde{\psi}_j(\mathbf{x})\psi_{i}(\mathbf{x})]$, the $i$--th coefficient in the orthogonal expansion of $\widetilde{\psi}_j(\mathbf{x})$ in terms of the basis $\{\psi_j(\cdot): j\geq 1\}$. 

On the other hand, each $a_{ji}$ is also the $j$--th coefficient in the orthogonal expansion of $\psi_i(\mathbf{x})$ in terms of the basis $\{\widetilde{\psi}_j(\mathbf{x}): j\geq 1\}$. Hence, $\Psi({\bf x})=\mathbf{A}^\top \widetilde{\Psi}({\bf x}).$ 

Since $\{\widetilde{\psi}_j(\cdot): j\geq 1\}$ and $\{\psi_j(\cdot): j\geq 1\}$ are orthonormal sequences, $\e[\widetilde{\Psi}({\bf x})\widetilde{\Psi}({\bf x})^\top]=\mathbf{I}$ and $\e[\Psi({\bf x})\Psi({\bf x})^\top]=\mathbf{I}$, where $\mathbf{I}$ is the identity matrix of infinite dimension; these imply $\mathbf{A}\mathbf{A}^\top=\mathbf{I}$ and $\mathbf{A}^\top \mathbf{A}=\mathbf{I}$. Therefore, $\mathbf{A}$ is an orthogonal matrix. 

It follows that
{\small
\bea
\sum_{j=1}^{\infty} \e[\mathbf{x} \, \widetilde{\psi}_j(\mathbf{x})] \,  \e[\mathbf{x}^{\top} \widetilde{\psi}_j(\mathbf{x})] & = &\e[\mathbf{x} \, \widetilde{\Psi}(\mathbf{x})^\top] \,  \e[ \widetilde{\Psi}(\mathbf{x})\mathbf{x}^{\top}]=\e[\mathbf{x} \, \Psi(\mathbf{x})^\top]\mathbf{A}^\top \mathbf{A}  \e[\Psi(\mathbf{x})\mathbf{x}^{\top}]
\nonumber\\
& =&\e[\mathbf{x} \, \Psi(\mathbf{x})^\top]\e[\Psi(\mathbf{x})\mathbf{x}^{\top}]=\sum_{i=1}^{\infty} \e[\mathbf{x} \, {\psi}_i(\mathbf{x})] \,  \e[\mathbf{x}^{\top} {\psi}_i(\mathbf{x})],
\nonumber
\eea
which} implies that the matrix $\bm{\Sigma}_x= \e[\mathbf{x}\, \mathbf{x}^{\top}] - \sum_{j=1}^{\infty} \e[\mathbf{x} \, {\psi}_j(\mathbf{x})] \,  \e[\mathbf{x}^{\top} {\psi}_j(\mathbf{x})]$ is invariant to the choice of the basis, so is $\bm{\Sigma}_{xy} = \e[\mathbf{x} \, y] - \sum_{j=1}^\infty \e[{\bf x} \, \psi_j({\bf x})] \, \e[y \, \psi_j({\bf x})]$. We then conclude that $\bm{\beta}_0=(\bm{\Sigma}_x)^{-1}\bm{\Sigma}_{xy}$ is invariant to the choice of the orthonormal basis $\{\psi_j(\cdot): j\geq 1\}$. 
\end{proof}

Lemmas \ref{lemma2.1} and \ref{lemma2.2} provide a rigorous treatment of identification by the proposed SP method. We are now ready to discuss estimation of $\bm{\beta}_0$ in Section \ref{Sec3}.

\section{Estimation and Inference}\label{Sec3}

\renewcommand{\theequation}{3.\arabic{equation}} 
\setcounter{equation}{0}

\subsection{Estimation theory}

We now consider estimation for model (\ref{a4}). Suppose that $\{\psi_j({\bf x}), j\ge 1\}$ is an orthonormal basis for ${\mathcal S}$, where $\e[\psi_j({\bf x})]=0$ and $\e[\psi_j({\bf x})\psi_{\ell}({\bf x})]=\delta_{j\ell}$ for all $(j,\ell)$. 

We then have an orthonormal series expansion for $m({\bf x})$ of the form:
\be
m({\bf x})=\sum_{j=1}^\infty \psi_j({\bf x}) \, \gamma_j = \sum_{j=1}^k \psi_j({\bf x}) \, \gamma_j  + \sum_{j=k+1}^\infty \psi_j({\bf x}) \, \gamma_j,
\label{es1}
\ee
which is a unified expression regardless of the dimensionality of $\mathbf{x}$ as discussed in the beginning of Appendix B.1 of the supplementary document. 

Given a truncation parameter $k$, we introduce {\small $\mathbf{V}_k(\mathbf{x})=(\psi_1(\mathbf{x}), \cdots, \psi_k(\mathbf{x}))^\top$ and $\bm{\gamma}=(\gamma_1, \cdots, \gamma_k)^\top$}, and then define $\delta_k(\mathbf{x})=\sum_{j=k+1}^\infty  \psi_j(\mathbf{x}) \, \gamma_j$. 

The proposed SP method then generates a ``new regressor,'' $\mathbf{V}_k({\bf x})$. In view of (\ref{es1}), model (\ref{a4}), $y= \mathbf{x}^{\top} \bm{\beta}_0 + m(\mathbf{x}) + e$, can be written as 
\be
y={\bf x}^\top {\bm \beta}_0 +\mathbf{V}_k({\bf x})^\top {\bm \gamma}+\delta_k({\bf x})+e. 
\label{grm}
\ee

Given a sample $\{(y_i,\mathbf{x}_i), 1\leq i \leq n\}$, we have $y_i=\mathbf{x}_i^\top{\bm{\beta}_0}+\mathbf{V}_k({\bf x}_i)^\top {\bm \gamma}+\delta_k({\bf x}_i)+e_i$ for $1\leq i\leq n$. Accordingly, the matrix form of the model is
\begin{equation}\label{2d}
\mathbf{y}=\mathbf{X} \, \bm{\beta}_0+\mathbf{V}\bm{\gamma}+\bm{\delta}+ \mathbf{e},
\end{equation}
where {\small $\mathbf{y}=(y_1, \cdots, y_n)^\top$, $\mathbf{X}=(\mathbf{x}_1, \cdots, \mathbf{x}_n)^\top$, $\mathbf{V}=(\mathbf{V}_k(\mathbf{x}_1), \cdots, \mathbf{V}_k(\mathbf{x}_n))^\top$, $\mathbf{e}=(e_1, \cdots, e_n)^\top$, and $\bm{\delta}\equiv \bm{\delta}(\mathbf{x}_1, \cdots, \mathbf{x}_n) = (\delta_k(\mathbf{x}_1), \cdots, \delta_k(\mathbf{x}_n))^\top$}.

Let $\mathbf{P}_v=\mathbf{V}\left(\mathbf{V}^{\top} \mathbf{V}\right)^{-1} \mathbf{V}^{\top}$ and $\mathbf{M}_v= \mathbf{I}_n - \mathbf{P}_v$. We then have $\mathbf{M}_v \, \mathbf{y}=\mathbf{M}_v\mathbf{X} \, \bm{\beta}_0+\mathbf{M}_v \, \bm{\delta} +\mathbf{M}_v \, \mathbf{e}$. \, The SP estimator is then given by
\be\label{2e}
\widehat{\bm{\beta}} \equiv: \widehat{\bm{\beta}}_{\rm SP} = ({\mathbf{X}^\top \mathbf{M}_v \mathbf{X}})^{-1}{\mathbf{X}^\top \mathbf{M}_v \mathbf{y}},
\ee
provided the inverse exists. 

If the cardinality of ${\mathcal S}$ is finite, then ${\bm \delta}=0$, and hence $\widehat{\bm{\beta}}$ is unbiased, as shown in the proof of Theorem 3.1 in Appendix C. For the general case, we allow ${\mathcal S}$ to be infinite dimensional in the remainder of the paper. 

In order to establish asymptotic normality results for $\widehat{\bm{\beta}}_{\rm SP}$, we need to introduce the following assumption.

\begin{assumption}\label{MTP2}
Suppose that (i) $\{{\bf x}_i, 1\leq i\leq n\}$ is a sequence of independent and identically distributed (i.i.d.) random variables and $\{e_i, 1\leq i\leq n\}$ is a sequence of independent random errors with $0<\sigma_i^2=\e\left[e_i^2|\mathbf{x}_i\right]<\infty$; (ii) $\bm{\Gamma}_{x} \equiv: \lim_{n\rightarrow \infty} \frac{1}{n} \e\left[{\bf X}^\top \, {\bf M}_v \, {\bm \Omega} \, {\bf M}_v \, {\bf X}\right]$ is positive definite with ${\bm \Omega}={\rm diag}(\sigma_1^2, \cdots, \sigma_n^2)$; (iii) $\{\psi_j(\cdot): j\geq 1\}$ is an orthonormal sequence with $\sup_j\e[\psi_j^4({\bf x})]<\infty$; and (iv) The truncation parameter $k$ is chosen such that $k^4=o(n)$ and $\sum_{j=k+1}^{\infty} \bm{\lambda}_j \, \gamma_j = o(n^{-1/2})$ as $(n, k)\rightarrow (\infty, \infty)$, where $\bm{\lambda}_j = \e[\mathbf{x}_1 \, \psi_j(\mathbf{x}_1)]$ and $\gamma_j=\e[m(\mathbf{x}_1)\, \psi_j(\mathbf{x}_1)]$.
\end{assumption}

\begin{rem}
Condition (i) allows for heteroskedasticity. Condition (ii) permits $\sigma_i^2$ to be stochastic or deterministic. In the special case of $\sigma_i^2\equiv \sigma_e^2 =\e[e_1^2]$, Condition (ii) reduces to $\bm{\Sigma}_x>{\bf 0}$ and $0<{\sigma}_e^2 <\infty$. Condition (iii) is verifiable for many orthonormal functions. 

As discussed in the supplementary Appendix A.4, Condition (iv) can be satisfied. In the case where the cardinality of $\mathcal{S}$ is finite, Condition (iv) follows trivially due to $\bm{\delta}={\bf 0}$. When there is no endogeneity, that is $\bm{\lambda}_j\equiv 0$, it is also satisfied automatically. 

In other cases, $\sum_{j=k+1}^{\infty} \bm{\lambda}_j \, \gamma_j = o(n^{-1/2})$ follows when $m(\mathbf{x})$ satisfies certain smoothness and/or periodicity conditions (see, for example, Theorem B.2.5 of Appendix B and Theorem C.4.8 of Appendix C in \cite{dg2025}). 

The choice of $k$ can be made using either a rule of thumb of the form {\small $k = \left[\ell_0 \, n^{\frac{d}{2(s+1) +d}}\right]$} or generalized cross--validation, where $s$ is the smoothness order of $m(\cdot)$ and $\ell_0>0$ is chosen by the user.
\end{rem}

We now establish the first theorem; its proof is given in the supplementary Appendix C.

\begin{theo}\label{MTPclt}
{Let Assumptions \ref{MTP1} and \ref{MTP2} hold.} (i) $\widehat{\bm{\beta}}_{\rm SP}$ given in \eqref{2e} is consistent. \, (ii) Let $\lim_{n\rightarrow \infty} \frac{1}{n} \sum_{i=1}^n \e[e_i^4|\mathbf{x}_i]<\infty$ (a.s.) and $\e\left[\left\|\mathbf{x}_i\right\|^4\right]<\infty$. We have as $(k, n)\to (\infty, \infty)$
\be
\left({\bf X}^\top \, {\bf M}_v \, \widehat{\bm \Omega} \, {\bf M}_v \, {\bf X} \right)^{-1/2} \left({\bf X}^\top\, {\bf M}_v \, {\bf X} \right)\, (\widehat{\bm{\beta}}_{\rm SP} -\bm{\beta}_0)\to_{\mathcal D} N\left(\mathbf{0}, {\bf I}_d\right),
\label{jiti4}
\ee
where {\small $\widehat{\bm \Omega}={\rm diag}(\widehat{e}_1^2, \cdots, \widehat{e}_n^2)$, $\widehat{e}_i=y_i-{\bf x}_i^\top\widehat{\bm{\beta}}_{\rm SP}- \mathbf{V}_k({\bf x}_i)^\top \widehat{\bm \gamma}$ and $\widehat{\bm \gamma}=({\bf V}^\top {\bf V})^{-1}{\bf V}^\top ({\bf y}-{\bf X}\widehat{\bm{\beta}}_{\rm SP})$}.

(iii)  We also have {\small $\DF{1}{n}{\bf X}^\top {\bf M}_v \, \widehat{\bm \Omega} \, {\bf M}_v{\bf X}=\DF{1}{n}{\bf X}^\top {\bf M}_v \, {\bm \Omega} \, {\bf M}_v{\bf X}+o_P(1) = \bm{\Gamma}_x  + o_P(1)$}, where the quantities involved are the same as in Assumption \ref{MTP2}.

(iv) Consequently, we have as $(k, n)\to (\infty, \infty)$
\be
\sqrt{n} \, (\widehat{\bm{\beta}}_{\rm SP} -\bm{\beta}_0)\to_{\mathcal D} N\left(\mathbf{0}, \bm{\Sigma}_x^{-1} \, \bm{\Gamma}_x \, \bm{\Sigma}_x^{-1}\right).
\label{jiti4new}
\ee

\end{theo} 
 
\begin{rem}

(i) As shown in Appendix A.2, when {\small ${\bf G}_k\equiv: \e\left[\mathbf{V}_k(\mathbf{x}) \, \mathbf{V}_k^{\top}(\mathbf{x})\right] \neq {\bf I}_k$}, Theorem \ref{MTPclt} remains true, basically because {\small $\frac{1}{n} {\bf X}^\top {\bf M}_v{\bf X} = \e[{\bf x} \, {\bf x}^\top]-\e[{\bf x} \, \widetilde{\mathbf{V}}_k^{\top}({\bf x})] \, \e[\widetilde{\mathbf{V}}_k({\bf x}) \, {\bf x}^\top]+o_P(1)$}, where {\small $\widetilde{\mathbf{V}}_k(\mathbf{x}) = {\bf G}_k^{-1/2} \, \mathbf{V}_k(\mathbf{x})$} satisfies {\small $\e\left[\widetilde{\mathbf{V}}_k(\mathbf{x}) \, \widetilde{\mathbf{V}}_k^{\top}(\mathbf{x})\right] = {\bf I}_k$}.

\noindent(ii) The estimator $\widehat{\bm{\beta}}_{\rm SP}$ is asymptotically normal and converges at the square-root--$n$ rate under endogeneity. Supplementary Appendix A.3 shows that $\widehat{\bm{\beta}}_{\rm SP}$ can be more efficient than $\widehat{\bm{\beta}}_{\rm LS}$ under endogeneity. 

\noindent(iii) Let {\small $\widehat{m}(\mathbf{x}) =\mathbf{V}_k(\mathbf{x})^{\top} \widehat{\bm{\gamma}}$, where $\widehat{\bm{\gamma}} = \left(\widehat{\gamma}_1, \cdots, \widehat{\gamma}_k\right)^{\top} =  \left(\mathbf{V}^{\top} \mathbf{V}\right)^{-1} \mathbf{V}^{\top} (\mathbf{y} - \mathbf{X} \, \widehat{\bm{\beta}}_{\rm SP})$}. The availability of $\widehat{m}(\mathbf{x})$ facilitates construction of a simple test of exogeneity against weak endogeneity, as discussed in Section \ref{Sec3}.2 below.

(iv) The SP method also yields {\small $\widehat{\varepsilon}_i = y_i - \mathbf{x}_i^{\top} \widehat{\bm{\beta}}_{\rm SP}$} without relying on external instruments. The availability of $\widehat{\varepsilon}_i$ helps reveal possible sources of endogeneity in empirical applications, as discussed in Supplementary Appendix B.4.

\end{rem}

It is noted that the identification and estimation of model (\ref{a4}) discussed above is applicable to an additive model of the form:
\be
y = \mathbf{x}_{\rm exo}^{\top} \bm{\beta}_{\rm exo} + \mathbf{x}_{\rm end}^{\top}\bm{\beta}_{\rm end} + \varepsilon, \ \ \mbox{with \ $\e[\varepsilon|\mathbf{x}_{\rm end}]\neq 0$ \ but \ $\e[\varepsilon|\mathbf{x}_{\rm exo}]=0$},
\label{con1}
\ee
where $\mathbf{x}_{\rm end}$ and $\mathbf{x}_{\rm exo}$ are the endogenous and exogenous components of $\mathbf{x}$, respectively. 

In Section \ref{Sec3}.2, we propose a simple statistic for testing exogeneity against a wide range of weak endogeneity. 

\subsection{Testing for exogeneity versus weak endogeneity}

Consider $y =\mathbf{x}^{\top}\bm{\beta}_0 + m(\mathbf{x}) + e$ with $m(\mathbf{x})=\e[\varepsilon|\mathbf{x}]$ under the following null hypothesis:
\begin{equation}\label{nonlintest1}
H_0:\ \ \p(m(\mathbf{x})=0)=1.
\end{equation}

\noindent We approximate $m(\mathbf{x})$ by $\sum_{j=1}^k \psi_j(\mathbf{x}) \gamma_{j} \equiv \mathbf{V}_k(\mathbf{x})^{\top} \bm{\gamma}$, where $\mathbf{V}_k(\cdot) = \left(\psi_1(\cdot), \cdots, \psi_k(\cdot)\right)^{\top}$, and ${\bm\gamma} =\left(\gamma_{1}, \cdots, \gamma_{k}\right)^{\top}$ can be estimated by $\widehat{\bm\gamma}$ in the same way as in Section \ref{Sec3}.1.

In order to test $H_0: \, \p(m({\bf x}) =0)=1$, it suffices to test ${\bm\gamma}=0$, or equivalently $({\bf V}^{\top} {\bf V}) \, {\bm\gamma}=0$. Since the ${\bm\gamma}$ parameter can be estimated by {\small $\widehat{\bm\gamma} = ({\bf V^{\top} V})^{-1} {\bf V}^{\top} ({\bf y} - \mathbf{X}\widehat{\bm\beta}_{\rm SP})$}, we introduce $\widetilde{\bm\gamma}_n = ({\bf V^{\top} V}) \, \widehat{\bm\gamma}$ and then define for each given point $x$:
\begin{equation}\label{4b}
\widetilde{m}_n(x) \equiv:  {\bf V}^{\top}_k(x) \widetilde{\bm\gamma}_n = \sum_{i=1}^n \mathbf{V}_k(x)^{\top} \mathbf{V}_k({\bf x}_i) \, \widetilde{e}_i \ \ \mbox{with \, $\widetilde{e}_i = y_i -  \mathbf{x}_i^{\top} \widehat{\bm\beta}_{\rm SP}$}.
\end{equation}

This motivates the following simple test statistic, a nonparametric version of the Hausman test proposed in \cite{hausman1978}:
\be
L_n(k) \equiv \int_{\mathcal X} \widetilde{m}_n^2({x}) \, f_{\mathbf{x}}(x) d {x} = \sum_{i=1}^n \sum_{j=1}^n \mathbf{V}_k^{\top}({\bf x}_i) \mathbf{V}_k({\bf x}_j) \, \widetilde{e}_i \, \widetilde{e}_j
\label{4c}
\ee
by orthonormality: $\int_{\mathcal X} \mathbf{V}_k({x}) \mathbf{V}_k^{\top}({x}) \, f_{\mathbf{x}}(x) d {x}={\bf I}_k$, where $f_{\mathbf{x}}(\cdot)$ denotes the density of ${\bf x}$.

In order to avoid dealing with robustness issues regarding the choice of individual $k$ values in practice, we propose using an additive version of $L_n(k)$ of the form:
{\small
\bea
T_n = \sum_{k=k_{\min}}^{k_{\max}} L_n(k) - \sum_{k=k_{\min}}^{k_{\max}} \sum_{i=1}^n \mathbf{V}_k^{\top}({\bf x}_i) \mathbf{V}_k({\bf x}_i) \, \widetilde{e}_i^2  
= \sum_{k=k_{\min}}^{k_{\max}} \sum_{i=1}^n \sum_{j=1, \neq i}^n \mathbf{V}_k^{\top}({\bf x}_i) \mathbf{V}_k({\bf x}_j) \, \widetilde{e}_i \, \widetilde{e}_j,
\nonumber
\eea
where $k_{\min}$ and $k_{\max}$} are the smallest and largest series positive integers, respectively.

Recall that we consider the following: $y = \mathbf{x}^{\top} \bm{\beta}_0 + m(\mathbf{x}) + e$. To discuss the main ideas about testing without involving too much complexity, we focus on the homoscedasticity case where $\e[e^2|\mathbf{x}] =\e[e^2]$ (a.s.) in the rest of this subsection.

\begin{theo}\label{th4.1}
\noindent Suppose that  (a) Assumption \ref{MTP1}(i)(iii) hold; (b) $\{({\bf x}_i, e_i): i\geq 1\}$ is a sequence of i.i.d. random variables with $\e[e_1^2]>0$, $\e[e_1^4]<\infty$ and $\e\left[\left\|\mathbf{x}_1\right\|^4\right]<\infty$; (c) $\{\psi_j(\cdot): j\geq 1\}$ is an orthonormal sequence with $\sup_j\e[\psi_j^4({\bf x})]<\infty$; and (d) $1\leq k_{\min} < k_{\max}<[\sqrt{n}]$, $k_{\max} - k_{\min}\rightarrow \infty$ and $\frac{k_{\max}^2}{n}\rightarrow 0$ as $n\rightarrow\infty$. We have under $H_0: \, \p(m(\mathbf{x})=0)=1$,
\begin{equation}\label{th4.1a}
\frac{T_n}{S_n}\rightarrow_{\mathcal D} N(0,1) \ \ \ \mbox{as $n\rightarrow \infty$},
\end{equation}
where  {\small $S_n^2\equiv 2 \widetilde{\sigma}_e^4 \, \sum_{i=1}^n \sum_{j=1}^n \left(\sum_{k=k_{\rm \min}}^{k_{\rm \max}} \mathbf{V}_k^{\top}({\bf x}_i) \mathbf{V}_k({\bf x}_j)\right)^2$}, in which {\small $\widetilde{\sigma}_e^2 = \frac{1}{n} \sum_{i=1}^n \widetilde{e}_i^2$}.
\end{theo}

Note that the proof is postponed to the supplementary Appendix C. Note also that $\frac{S_n^2}{\sigma_n^2}\rightarrow_P 1$ with $\sigma_n^2 = \frac{2 \,{\sigma}_e^4}{3} \cdot n^2 \, k_{\max}^3\, (1+ o(1))$ as shown in the proof of Theorem \ref{th4.1} in Appendix C.  

We now consider a sequence of local alternatives of the form:
\be
H_1: \, \p(m_n({\bf x})=a_n\, \Delta({\bf x}))=1, 
\label{4e}
\ee
where the positive sequence $a_n\to 0$ at a specified rate, while $\Delta(\mathbf{x})\in {\mathcal S}$ and $\e[\Delta^2({\bf x})]>0$.

It is pointed out that such a sequence of local alternatives naturally covers a wide range of weak endogeneity. The following theorem establishes the consistency of the proposed test under $H_1$; its proof is postponed ro Appendix C of the supplementary document.

\begin{theo}\label{th4.2}
Suppose that  (i) Assumption \ref{MTP1}(i)(iii) hold with $m(\cdot)$ being replaced by $\Delta(\cdot)$; (ii) Conditions (b)--(d) of Theorem \ref{th4.1} hold; and (iii) $a_n\to 0$ and $a_n^2 \, n \, k_{\max}^{-1/2}\to \infty$ as $n\rightarrow \infty$. Under $H_1$ of (\ref{4e}), we have 
\be
\DF{T_n}{S_n}\rightarrow_{P} \infty \ \ \  \mbox{as $n\rightarrow \infty$}.
\label{th4.1b}
\ee
\end{theo}

Theorem \ref{th4.2} shows that the proposed test detects certain types of weak endogeneity of order $a_{n}$, provided that $a_n$ satisfies $a_{n}^2 \, n \, k_{\max}^{-1/2} \rightarrow \infty$. The fastest possible rate, $a_n = n^{-\frac{1}{2}} \, k_{\max}^{c_1}$, is near--optimal when $k_{\max} \rightarrow \infty$ at the slowest possible rate, such as $k_{\max} = \left[ c_2 \, \log(\log(n))\right]$, where $c_1>\frac{1}{4}$ and $c_2>0$ are chosen by the user. The conventional optimal parametric rate $a_n = n^{-1/2}$ is not achievable in the nonparametric series setting under consideration.

\subsection{Linear models under weak endogeneity}\label{Ap.A2}

The existing literature pays particular attention on the case where $(\varepsilon, \mathbf{x})$ in model: $y=\mathbf{x}^{\top} \bm{\beta}_0 + \varepsilon$ follows a joint Gaussian distribution, which implies that $m(\mathbf{x}) = \left(\mathbf{x} - \e[\mathbf{x}]\right)^{\top} \bm{\gamma}_0$ and model: $y=\mathbf{x}^{\top} \bm{\beta}_0 + \varepsilon$ becomes $y =\mathbf{x}^{\top} \bm{\beta}_0 + \left(\mathbf{x} - \e[\mathbf{x}]\right)^{\top} \bm{\gamma}_0+ e$ with $\e[e|\mathbf{x}]=0$. 

We show in the rest of this section that while certain linearity cases are excluded by Assumption 2.1, the proposed SP method is still applicable to consistently estimate $\bm{\beta}_0$ in such other linearity settings.  

When $m(\mathbf{x}) = \left(\mathbf{x} - \e[\mathbf{x}]\right)^{\top} \bm{\gamma}_0$, we rewrite model: $y = \mathbf{x}^{\top} \bm{\beta}_0 + \varepsilon$ as
\be
y - \e[y]= \left(\mathbf{x} - \e[\mathbf{x}]\right)^{\top} \bm{\beta}_0 + m(\mathbf{x}) + e = \left(\mathbf{x} - \e[\mathbf{x}]\right)^{\top} \left(\bm{\beta}_0 + \bm{\gamma}_0\right) + e,
\label{a44}
\ee
which implies that $\left(\bm{\beta}_0 +  \bm{\gamma}_0\right) = \e^{-1} [\left(\mathbf{x} - \e[\mathbf{x}]\right) \left(\mathbf{x} - \e[\mathbf{x}]\right)^{\top} ] \, \e\left[\left(\mathbf{x} - \e[\mathbf{x}]\right)\, (y - \e[y])\right]$ can only be correctly identified collectively, rather than $\bm{\beta}_0$ individually.

When there are some additional restrictions imposed on $\bm{\gamma}_0$, such as those weak endogeneity scenarios discussed below, however, $\bm{\beta}_0$ may still be asymptotically identifiable. In such cases where $\bm{\gamma}_0 \equiv \bm{\gamma}_{n0} \rightarrow {\bf 0}$ as $n\rightarrow \infty$, $\bm{\beta}_0$ can then be estimated consistently. We now focus on several weak endogeneity settings. 

Let us start with model: $y=\mathbf{x}^{\top} \bm{\beta}_0 + \varepsilon$ associated with a type of weak endogeneity of the form: 
\be
y_i = \mathbf{x}_i^{\top} \bm{\beta}_0 + \varepsilon_i \ \ \ \mbox{with} \ \ \ \varepsilon_i =m_n(\mathbf{x}_i) + e_i,
\label{jitiendo1}
\ee
where $\e[e_i|\mathbf{x}_i]=0$ for $1\leq i\leq n$, $m_n(\mathbf{x})=\e[\varepsilon_i|\mathbf{x}_i = \mathbf{x}] = a_n \, m_{\ast}(\mathbf{x})$ with $\e[m_{\ast}(\mathbf{x})]=0$, $\e[\mathbf{x} \, m_{\ast}(\mathbf{x})] \neq 0$, $\e[m_{\ast}^2({\bf x})]>0$, $m_{\ast}(\mathbf{x}) \in {\mathcal S}$ and $a_n\rightarrow 0$ as $n\rightarrow \infty$. We have as $n\rightarrow \infty$
\be
\e[\varepsilon_i \, \mathbf{x}_i] = \e[\mathbf{x}_i \, m_n({\bf x}_i)] = a_n \, \e[\mathbf{x}_i \, m_{\ast}({\bf x}_i)] \to 0.
\label{jitiendo2}
\ee

Because of the projection in (\ref{jitiendo1}) and then the definition of weak endogeneity in (\ref{jitiendo2}), we are able to show that $\bm{\beta}_0$ can still be consistently estimated. 

Since we assume that $m_{\ast}(\mathbf{x}) \in {\mathcal S}$, we approximate $m_{\ast}(\cdot)$ in the same way as for $m(\cdot)$ in Section \ref{Sec3}.1. In view of the fact that the `endogenous component' represented by $\mathbf{V}_k^{\top} \bm{\gamma}_n$, with $\bm{\gamma}_n= a_n \, \bm{\gamma}$, has been eliminated. Consequently, Section 2.2 shows that $\bm{\beta}_0$ of (\ref{jitiendo1}) is correctly identified and therefore it can still be estimated by $\widehat{\bm{\beta}} =\widehat{\bm{\beta}}_{\rm SP}$ of Section \ref{Sec3}.1 consistently. 

The establishment and the proof of Theorem 3.1 imply Corollary \ref{cor4.1} below.

\begin{cor}\label{cor4.1}
Let the conditions of Theorem 3.1(ii) hold for model (\ref{jitiendo1}) with $m(\cdot)$ being replaced by $m_{\ast}(\cdot)$. We then have 
    \be
\ \left({\bf X}^\top \, {\bf M}_v \, \widehat{\bm \Omega} \, {\bf M}_v \, {\bf X} \right)^{-1/2} \left({\bf X}^\top \, {\bf M}_v \, {\bf X} \right)(\widehat{\bm{\beta}}_{\rm SP} -\bm{\beta}_0)\to_{\mathcal D} N\left(\mathbf{0}, {\bf I}_d\right)
\label{endow}
\ee
as $(k, n)\to (\infty, \infty)$, where the quantities involved are the same as those of Theorem 3.1.
\end{cor}

Let us then discuss the linear weak endogeneity case where $m_n(\mathbf{x}) =\frac{\bm{\mu}^{\top} \left(\mathbf{x} - \e[\mathbf{x}]\right)}{\sqrt{n}}$, in which $\bm{\mu}$, satisfying $\|\bm{\mu}\|^2<\infty$, may not be generally identifiable. There is a substantial literature as reviewed in a recent survey by \cite{ass2019} about treatments of weak endogeneity involving weak IVs. 
\smallskip

Letting $\bm{\beta}_{n, \bm{\mu}} \equiv: \bm{\beta}_0 + \frac{\bm{\mu}}{\sqrt{n}}$, model (\ref{jitiendo1}) becomes 
\be
y_i - \overline{y}_n = \left(\mathbf{x}_i - \overline{\mathbf{x}}_n\right)^{\top} \bm{\beta}_{n, \bm{\mu}}  + e_i,
\label{jitiendo3}
\ee
when $m_n(\mathbf{x}) =\frac{\bm{\mu}^{\top} \left(\mathbf{x} - \e[\mathbf{x}]\right)}{\sqrt{n}}$, where $\overline{\mathbf{x}}_n = \frac{1}{n} \sum_{i=1}^n \mathbf{x}_i$ and $\overline{y}_n = \frac{1}{n} \sum_{i=1}^n y_i$.

It then follows that the standard OLS estimator of $\bm{\beta}_{n, \bm{\mu}}$ becomes
\be
 \widehat{\bm{\beta}}_{\rm LS} \equiv: n^{-1} \,\widehat{\bm{\Sigma}}_{\mathbf{x}}^{-1} \, \sum_{i=1}^n \left(\mathbf{x}_i - \overline{\mathbf{x}}_n\right) \, (y_i - \overline{y}_n) =\bm{\beta}_{n, \bm{\mu}} + n^{-1} \,\widehat{\bm{\Sigma}}_{\mathbf{x}}^{-1} \, \sum_{i=1}^n \left(\mathbf{x}_i - \overline{\mathbf{x}}_n\right) \, (e_i - \overline{e}_n), 
\label{ols1}
\ee
where $\widehat{\bm{\Sigma}}_{\mathbf{x}} = \frac{1}{n} \sum_{i=1}^n \left(\mathbf{x}_i - \overline{\mathbf{x}}_n\right)\, \left(\mathbf{x}_i - \overline{\mathbf{x}}_n\right)^{\top}$ and $\overline{e}_n = \frac{1}{n} \sum_{i=1}^n e_i$.

We then establish the following theorem; its proof follows trivially from verifying Lindeberg conditions needed for a central limit theorem for an independent data setting.

\begin{cor}\label{cor4.2}
 
Let $\min_{i\geq 1}\e\left[\left(\mathbf{x}_i - \e[\mathbf{x}_i]\right) \, \left(\mathbf{x}_i - \e[\mathbf{x}_i]\right)^{\top}\right]>{\bf 0}$ and $\max_{i\geq 1} \e\left[\left\|\mathbf{x}_i\right\|^4\right]<\infty$. Then under Assumption 3.1(i), we have
\be
 \sqrt{n} \, \widehat{\bm{\Sigma}}_n(\mathbf{x}) \, \left(\widehat{\bm{\beta}}_{\rm LS} - \bm{\beta}_0 - \frac{\bm{\mu}}{\sqrt{n}}\right) \rightarrow_{\mathcal D} N\left(\mathbf{0}, \, {\bf I}_d \right)
\label{ols2}
\ee
as $n\rightarrow \infty$, where $\widehat{\bm{\Sigma}}_n(\mathbf{x}) = \widehat{\bm{\Sigma}}_{x, e}^{-1/2} \, \widehat{\bm{\Sigma}}_{\mathbf{x}}$, in which $\widehat{\bm{\Sigma}}_{x, e} = \frac{1}{n} \sum_{i=1}^n \left(\mathbf{x}_i - \overline{\mathbf{x}}_n\right) \, \left(\mathbf{x}_i - \overline{\mathbf{x}}_n\right)^{\top} \, \widetilde{e}_i^2$ with $\widetilde{e}_i = (y_i - \overline{y}_n) - \left(\mathbf{x}_i - \overline{\mathbf{x}}_n\right)^{\top} \widehat{\bm{\beta}}_{\rm LS}$.
\end{cor}

Equation (\ref{ols2}) implies that $\widehat{\bm{\beta}}_{\rm LS} = \bm{\beta}_0 + \frac{\bm{\mu}}{\sqrt{n}} + O_P\left(\frac{\widehat{\bm{\Sigma}}_n^{-1}(\mathbf{x})}{\sqrt{n}}\right) \rightarrow \bm{\beta}_0$, indicating that $\widehat{\bm{\beta}}_{\rm LS}$ is a consistent estimator of $\bm{\beta}_0$. 
\smallskip

Furthermore, we show that the SP method offers a simple way to deal with a mixture of strong and weak endogeneity of the form: 
\be
m_n(\mathbf{x}) =  \left(\frac{\bm{\mu}}{\sqrt{n}}\right)^{\top}  \left(\mathbf{x} - \e[\mathbf{x}]\right) + m^{\ast}(\mathbf{x}),
\label{jitiendo4}
\ee
which implies  $\e\left[\mathbf{x} \, m_n(\mathbf{x})\right] =  {\rm var}[\mathbf{x}] \, \frac{\bm{\mu}}{\sqrt{n}} + \e\left[\mathbf{x} \, m^{\ast}(\mathbf{x})\right]$, where $\e[m^{\ast}(\mathbf{x})]=0$, $\e|m^{\ast}({\bf x})|^2>0$ and $m^{\ast}(\mathbf{x}) \in {\mathcal S}$. Our approach allows either $\e[\mathbf{x} \, m^{\ast}(\mathbf{x})] = 0$ or $\e[\mathbf{x} \, m^{\ast}(\mathbf{x})] \neq 0$. 

We use the same $\{\psi_j(\mathbf{x}): j\geq 1\}$ as an orthonormal series to approximate $m^{\ast}(\mathbf{x})$ by $m_{k}^{\ast}(\mathbf{x}) \equiv:\sum_{j=1}^{k} \psi_j(\mathbf{x}) \, \gamma_{j}^{\ast}$, where $\bm{\gamma}^{\ast}_k = \left(\gamma_{1}^{\ast}, \cdots, \gamma_{k}^{\ast}\right)^{\top}$ is an array of coefficients. 

In the same spirit as in Section \ref{Sec3}.1, we now have the following model: 
\be
w_i - \overline{w}_n = \left(\mathbf{z}_i - \overline{\mathbf{z}}_n\right)^{\top} \bm{\beta}_{n, \bm{\mu}} + e_i, 
\label{ols3}
\ee
which implies that we can estimate $\bm{\beta}_{n, \bm{\mu}} = \bm{\beta}_0 + \frac{\bm{\mu}}{\sqrt{n}}$ by 
\be
\widehat{\overline{\bm{\beta}}}_{\rm SP} = \left(\sum_{i=1}^n \left(\mathbf{z}_i - \overline{\mathbf{z}}_n\right)  \, \left(\mathbf{z}_i - \overline{\mathbf{z}}_n\right)^{\top}\right)^{-1} \, \sum_{i=1}^n  \left(\mathbf{z}_i - \overline{\mathbf{z}}_n\right) \, \left(w_i - \overline{w}_n\right),
\label{ols4}
\ee
where $\overline{\mathbf{z}}_n = \frac{1}{n} \sum_{i=1}^n \mathbf{z}_i$ and $\overline{w}_n = \frac{1}{n} \sum_{i=1}^n w_i$, in which $w_i = y_i - \mathbf{v}_i^{\top} \left(\sum_{l=1}^n \mathbf{v}_l \, \mathbf{v}_l^{\top}\right)^{-1} \, \sum_{j=1}^n y_j \,\mathbf{v}_j^{\top}$ and $\mathbf{z}_i =  \mathbf{x}_i - \sum_{j=1}^n \mathbf{x}_j \, \mathbf{v}_j^{\top}\, \left(\sum_{l=1}^n \mathbf{v}_l \, \mathbf{v}_l^{\top}\right)^{-1} \mathbf{v}_i$ with $\mathbf{v}_i=\mathbf{V}_k(\mathbf{x}_i)$.
\medskip

We finally establish the following corollary; its proof follows from that of Theorem 3.1.
 
\begin{cor}\label{cor4.3}
Let Assumptions 2.1 and 3.1 hold with $m(\cdot)$ being replaced by $m^{\ast}(\cdot)$. We have
\be
\sqrt{n} \, \overline{\bm{\Sigma}}_n(\mathbf{z}, k) \, \left(\widehat{\overline{\bm{\beta}}}_{\rm SP} - \bm{\beta}_0 - \frac{\bm{\mu}}{\sqrt{n}}\right) \rightarrow_{\mathcal D} N\left(\mathbf{0}, \, {\bf I}_d\right)
\label{ols5}
\ee
 as $n\rightarrow \infty$, where $\overline{\bm{\Sigma}}_n(\mathbf{z}, k) =\overline{\bm{\Sigma}}_{z, e}^{-1/2}(k) \, \overline{\bm{\Sigma}}_z(k)$, in which 
\be
\overline{\bm{\Sigma}}_z(k) = \frac{1}{n} \sum_{i=1}^n \left(\mathbf{z}_i - \overline{\mathbf{z}}_n\right)  \, \left(\mathbf{z}_i - \overline{\mathbf{z}}_n\right)^{\top} \ \  \mbox{and} \ \ \overline{\bm{\Sigma}}_{z, e}(k) = \frac{1}{n} \sum_{i=1}^n \left(\mathbf{z}_i - \overline{\mathbf{z}}_n\right) \, \left(\mathbf{z}_i - \overline{\mathbf{z}}_n\right)^{\top} \, \widehat{\overline{e}}_i^2
\nonumber
\ee
with $\widehat{\overline{e}}_i = w_i - \overline{w}_n - \left(\mathbf{z}_i - \overline{\mathbf{z}}_n\right)^{\top} \, \widehat{\overline{\bm{\beta}}}_{\rm SP}$.
\end{cor}

\noindent This implies {\small $\widehat{\overline{\bm{\beta}}}_{\rm SP} = \bm{\beta}_0 + \frac{\bm{\mu}}{\sqrt{n}} + O_P\left(\frac{\overline{\bm{\Sigma}}_n^{-1}(\mathbf{z}, k)}{\sqrt{n}}\right) \rightarrow \bm{\beta}_0$}, which indicates that $\widehat{\overline{\bm{\beta}}}_{\rm SP}$ is still a consistent estimator of $\bm{\beta}_0$. 

It is noted that $\bm{\mu}$ is not identifiable and estimable, so a bias correction method will need to be developed to eliminate it. We wish to leave this for further study.
\medskip

Section \ref{Sec4} proposes a simple LASSO selection method for implementing the SP method in practice. Section \ref{Sec5} evaluates the estimator using a real-data application, and Section \ref{Sec6} evaluates the estimation and testing methods using simulated data under orthonormality. 

The finite--sample results reported in Sections \ref{Sec5} and \ref{Sec6} as well as in the supplementary Appendix B indicate that the methods established in Section \ref{Sec3} work well numerically.

\section{LASSO Selection}\label{Sec4}

\renewcommand{\theequation}{4.\arabic{equation}} 
\setcounter{equation}{0}

Constructing the subspace ${\mathcal S}$ in Assumption 2.1 is crucial for the identification analysis in Section \ref{Sec2}. The limiting covariance matrix of $\widehat{\bm \beta}_{\rm SP}$ is proportional to ${\bm \Sigma}_x^{-1}$, where ${\bm \Sigma}_x=\e[{\bf x}{\bf x}^\top]-\sum_{j=1}^\infty \e[{\bf x}\psi_j({\bf x})]\e[\psi_j({\bf x}){\bf x}^\top]$. Hence, using fewer basis functions can improve the efficiency of $\widehat{\bm \beta}_{\rm SP}$. 

It is thus reasonable to choose the smallest subspace $\mathcal{S}$ containing $m({\bf x})$. We propose a data-driven approach to selecting the basis functions that span $\mathcal{S}$. We remark that sparse selection does not itself solve the identification problem, but it does provide a parsimonious representation where such a sparse representation holds.

As a closed subspace of Hilbert space $\mathcal{L}^2$, $\mathcal{S}$ is also a Hilbert space. Since $\mathcal{L}^2$ is separable, it possesses an orthonormal basis \cite[p. 169]{dudley2003}; a subsequence of this basis can form a basis for $\mathcal{S}$. Specifically, let $\{\phi_j({\bf x}), j\ge 0\}$ be a basis of $\mathcal{L}^2$, and let the subsequence $\{\phi_{s_j}({\bf x}), j\geq 1\}$ be the basis of $\mathcal{S}$ denoted by $\{\psi_j({\bf x}), j\geq 1\}$ in Section \ref{Sec2}; thus, $\psi_j({\bf x})=\phi_{s_j}({\bf x})$ for $j\geq 1$. 
In the rest of this section, we show how to select $\{\psi_j({\bf x}), j\geq 1\}$ from its parent sequence $\{\phi_j({\bf x}), j\ge 0\}$. Thus, $\mathcal{S}=\overline{\rm span}\{\psi_j({\bf x}), j=1, 2, \ldots\}$, where $\overline{\mathbf{A}}$ denotes the closure of the set $\mathbf{A}$. 

We recall the function space ${\mathcal L}^2$ and suppose that $\{\phi_j(\mathbf{x}): j\geq 0\}$ is its orthonormal basis. As discussed in the supplementary Appendix B.1 and then implemented in the supplementary Appendix B.2, the choice of $\{\phi_j({\bf x}), j\ge 0\}$ can be quite general and flexible, without necessarily imposing orthogonality on $\{\phi_j({\bf x}), j\ge 0\}$. In the rest of this subsection, our goal becomes to find the basis of the subspace $\mathcal{S}$ from the basis $\{\phi_j({\bf x}), j\ge 0\}$ in ${\mathcal L}^2$. 

Since $m({\bf x})\in \mathcal{S}\subsetneqq{{\mathcal L}^2}$, we may expand $m({\bf x})$ in terms of $\{\phi_j({\bf x}), j\ge 0\}$ as $m({\bf x})=\sum_{j=1}^\infty \phi_j({\bf x}) \, \xi_j$, where $\{\xi_j: j\geq 1\}$ may differ from the coefficients $\{\gamma_j: j\geq 1\}$ in the expansion of $m({\bf x})$ in terms of $\{\psi_j({\bf x}), j\ge 1\}$ used above. We omit the constant term $\phi_0({\bf x})$ because $\e[m({\bf x})]=0$. Since $m({\bf x})\in\mathcal{S}$, Parseval's identity gives $\e[m^2({\bf x})]=\sum_{j=1}^\infty \xi_j^2<\infty$ under the orthogonality conditions $\e[\phi_j^2(\mathbf{x})]=1$ and $\e[\phi_j(\mathbf{x}) \, \phi_{\ell}(\mathbf{x})]=0$ for $j\neq \ell$. In practice, many $\xi_j$ are either zero or statistically insignificant. 

From the perspective of sieve estimation, we cannot estimate an infinite-dimensional function directly; see \cite[p. 5561]{chen2007}. Instead, we use a sequence of sieve spaces $\cdots \subset {\mathcal L}^2_k\subset {\mathcal L}^2_{k+1}\subset \cdots \subset {\mathcal L}^2$, where ${\mathcal L}^2_k={\rm span}\{\phi_j({\bf x}), 1\leq j\leq k\}$, that approaches ${\mathcal L}^2$. Correspondingly, ${\mathcal S}_{(k)}={\rm span}\{\phi_{s_j}({\bf x}), 1\leq s_j\le k\}$ forms a sieve for ${\mathcal S}$. Once we identify each element of ${\mathcal S}_{(k)}$ with high probability, we obtain $\widehat{\mathcal S}_{(k)}$ by the data-driven method.

{To proceed,} given a sample $\{({\bf x}_i, y_i), i=1,\cdots,n\}$, we rewrite the model as follows:
\begin{equation}\label{4.2D}
y_i={\bf x}_i^\top {\bm\beta}_0 + \sum_{j=1}^k \phi_j({\bf x}_i) \, \xi_j +\delta_{k}({\bf x}_i)+e_i,
\end{equation}
where $k$ is the truncation parameter, and $\delta_{k}({\bf x}_i)=\sum_{j=k+1}^\infty \phi_j({\bf x}_i) \, \xi_j$. The parameters $\{\xi_j\}$ need to meet certain sparsity conditions that will be specified soon.

Based on \eqref{4.2D}, we propose the following LASSO selection procedure.
{\small

\begin{equation}\label{4.3D}
{\bf Step \ 1}: \ (\widehat{\bm \beta}^*, \, \widehat{\bm\xi}^*)=\underset{{(\bm \beta, \, \bm\xi)}}{\arg\min} \frac{1}{n}\sum_{i=1}^n\left[y_i-{\bf x}_i^\top {\bm\beta} - \sum_{j=1}^{k} \phi_j({\bf x}_i)\xi_j\right]^2+\lambda \sum_{j=1}^{k} |\xi_j|,
\end{equation}
where $\widehat{\bm\xi}^*=(\widehat{\xi}_1^*,\ldots, \widehat{\xi}_k^*)^\top$, and $k$ and $\lambda$ are the user--chosen truncation and tuning parameters.

\begin{equation}\label{4.3D.2}
{\bf Step \ 2}: \, (\widehat{\bm \beta}^\dag, \, \widehat{\bm\xi}^\dag)=\underset{{(\bm \beta, \, \bm\xi)}}{\arg\min} \frac{1}{n}\sum_{i=1}^n\left[y_i-{\bf x}_i^\top {\bm\beta} - \sum_{j=1}^{k} \phi_j({\bf x}_i)\xi_j\right]^2+\lambda \sum_{j=1}^{k}\zeta_j |\xi_j|,
\end{equation}
where $\{\zeta_j\}$ is a set of predetermined weights, such as $1/\widehat{\xi}_j^*$ when $\widehat{\xi}_j^* \neq 0$, and $\widehat{\xi}_j^*$ is the $j$th element of $\widehat{\bm\xi}^*$. The LASSO estimate $\widehat{\bm\xi}^\dag=(\widehat{ \xi}_{1}^\dag,\ldots, \widehat{ \xi}_{k}^\dag)^\top$ enables us to obtain
\be
    \widehat{\mathcal{S}}_{(k)}=\text{span}\left\{\phi_j({\bf x}) \, |\, \widehat{ \xi}_{j}^\dag\ne 0, j\geq 0\right\},
\label{subset}
\ee
for which we show that $\Pr\left(\widehat{\mathcal{S}}_{(k)}=\mathcal{S}_{(k)}\right)\rightarrow 1$ as $(n, k) \rightarrow (\infty, \infty)$ in Theorem \ref{Lem.Lasso} below. 

}

\begin{assumption}\label{AS.lasso}

(i) Let $\mathscr{C}=\{\ell \mid \xi_\ell\ne 0, \, \ell \le k\}$ with cardinality $\sharp \mathscr{C} =k_0$ and $\bar{\mathscr{C}}=\{\ell \mid \xi_\ell = 0, \, \ell \le k\}$, where $k_0$ can be either fixed or varying to infinity, and $k$ varies to infinity. \ (ii) There exists some $J> 2$ such that $\e\left[|e_i|^J\right]<\infty$, and {\footnotesize $k \, \left(\left[n^{\frac{J-2}{2}}(\log k)^{\frac{J}{2}}\right]\right)^{-1} \to 0$} and $\frac{k_0 \sqrt{\log k}}{\sqrt{n}}\to 0$ as $(k, n)\to (\infty, \infty)$.  (iii) $\min_{\ell \in \mathscr{C}} |\xi_{\ell}| \gg \sqrt{\frac{(k_0\log  k)}{n}} \, (1+  \max_{\ell\in \mathscr{C}} \zeta_{\ell} )$ and $\min_{\ell \in \bar{\mathscr{C}}} \zeta_{\ell}\gg k_0 \left(1+  \max_{\ell\in \mathscr{C}} \zeta_{\ell}\right)$.
\end{assumption}

Assumption \ref{AS.lasso}(i) is a standard sparsity restriction and implies that $\mathscr{C}\cap \bar{\mathscr{C}}=\emptyset$ and $\sharp \bar{\mathscr{C}}=k-k_0$. Assumption \ref{AS.lasso}(ii) further restricts the truncation parameter relative to the sample size. In some special cases, for example when $m(\cdot)$ is either an even or an odd function and $k_0\le k/2$, this condition is easily satisfied. 

Assumption \ref{AS.lasso}(iii) is verifiable. Let us consider $\min_{\ell \in \mathscr{C}} |\xi_{\ell}| \gg \sqrt{\frac{(k_0\log  k )}{n}} (1+  \max_{\ell\in \mathscr{C}} \zeta_{\ell} )$. Take $\zeta_j=\frac{1}{\widehat{\xi}_j^*}$ when $\widehat{\xi}_j^*\neq 0$. In this case, $\min_{\ell \in \mathscr{C}} |\xi_{\ell}|\asymp 1$ and $\max_{\ell\in \mathscr{C}} \zeta_{\ell} \asymp 1$ because $\widehat{\xi}_j^*$ is consistent. Therefore, the inequality holds. 

Similarly, Condition $\min_{\ell \in \bar{\mathscr{C}}} \zeta_{\ell}\gg k_0 \left(1+  \max_{\ell\in \mathscr{C}} \zeta_{\ell}\right)$ also holds. In the case $\zeta_j=\frac{1}{\widehat{\xi}_j^*}$ (in the case of $\widehat{\xi}_j^*\neq 0$), $\min_{\ell \in \bar{\mathscr{C}}} \zeta_{\ell}$ diverges at the rate of $\sqrt{n}/\sqrt{k_0\log k}$ and $\max_{\ell\in \mathscr{C}} \zeta_{\ell} \asymp 1$. Provided $\frac{k_0 \sqrt{\log k}}{\sqrt{n}}\to 0$ specified in Assumption \ref{AS.lasso}, this condition holds. \ We establish Theorem \ref{Lem.Lasso}; its proof is given in Appendix C.

\begin{theo}\label{Lem.Lasso}
Let the conditions of Theorem \ref{MTPclt} and Assumption \ref{AS.lasso}(i)(ii) hold, and let $\lambda\asymp  \frac{\sqrt{\log k}}{\sqrt{n}} $. We obtain as $(k, n)\to (\infty, \infty)$

{\small 
(i)  $\|\VEC(\boldsymbol{\beta}_0 - \widehat{\boldsymbol{\beta}}^*, \boldsymbol{\xi}-\widehat{\boldsymbol{\xi}}^* ) \| =O_P\left(\frac{\sqrt{k_0\log k}}{\sqrt{n}}\right)$; \, (ii) \  $\|\VEC(\boldsymbol{\beta}_0 - \widehat{\boldsymbol{\beta}}^*, \boldsymbol{\xi}-\widehat{\boldsymbol{\xi}}^* ) \|_1 =O_P\left(\frac{k_0\sqrt{\log k}}{\sqrt{n}}\right)$}.

\noindent If, in addition, Assumption \ref{AS.lasso}(iii) also holds, we then have as $(k, n)\to (\infty, \infty)$

(iii) $\Pr\left(\widehat{\mathcal{S}}_{(k)}=\mathcal{S}_{(k)}\right)\rightarrow 1$; and (iv) {\small $\left({\bf X}^\top \, {\bf M}_v \, \widehat{\bm \Omega}\, {\bf M}_v \, {\bf X} \right)^{-1/2} \left({\bf X}^\top \, {\bf M}_v \, {\bf X} \right)\, (\widehat{\bm \beta}^\dag - {\bm \beta}_0)\to_D N(\mathbf{0}, {\bf I}_d)$}, where the quantities involved are the same as in Theorem \ref{MTPclt}.
\end{theo}

Theorem \ref{Lem.Lasso} (iii) implies that we can correctly choose the subset $\mathcal{S}_{(k)}$ asymptotically, and Theorem \ref{Lem.Lasso} (iv) shows that the asymptotic normality remains true for $\widehat{\bm \beta}^\dag$. 

Our experience suggests that the cardinality of $\mathcal{S}_{(k)}$ is small in practice. If it is still relatively large, a model selection method, such as the generalized cross--validation method discussed in the supplementary Appendix B.1.2 below, may be employed to determine the cardinality.

\section{The return to schooling case study}\label{Sec5}

\renewcommand{\theequation}{5.\arabic{equation}} 
\setcounter{equation}{0}

\noindent{\bf Example 5.1}: In this example, we examine the returns to schooling (see, for example, \citealp{card2001}) using data from the 1979 National Longitudinal Survey of Youth\footnote{The data are obtained from \url{https://davidcard.berkeley.edu/data_sets.html}.}. After removing observations with missing values, the sample contains $n=2639$ observations. We consider the following variables.

In Table \ref{tb.em.1}, $y$ is $\log$(wage), the outcome, and $x$ is the endogenous regressor, measured as the child's years of education. Following \cite[Chapter 15]{wooldridge2016} and \cite{cff2025}, we consider the following set of IVs: $z_1$ is the father's years of education; $z_2$ equals 1 if the individual grew up near an accredited four-year college and 0 otherwise; $z_3$ equals 1 if the individual grew up near an accredited two-year college and 0 otherwise; $z_4$ equals 1 if the individual grew up near an accredited four-year public college and 0 otherwise; and $z_5$ equals 1 if the individual grew up near an accredited four-year private college and 0 otherwise.

{
\begin{table}[htb!]
\centering
\renewcommand{\arraystretch}{0.8}
\caption{Summary Statistics}\label{tb.em.1}
\begin{tabular}{lrr}
\toprule
 & Mean & sd \\
$\log$(wage) (y) & 6.329 & 0.444 \\
edu $(x)$ & 13.268 & 2.634 \\
dad edu $(z_1)$ & 9.956 & 3.229 \\
4-year college $(z_2)$ & 0.690 & 0.463 \\
2-year college $(z_3)$ & 0.437 & 0.496 \\
4-year pub college $(z_4)$ & 0.490 & 0.500 \\
4-year pri college $(z_5)$ & 0.200 & 0.400 \\
\bottomrule
\end{tabular}
\end{table}}

We consider the following linear model: $y_i= \alpha_0+ x_i \, \beta_0+\varepsilon_i$. We rescale $x_i$'s as follows:
\begin{equation}\label{data1}
    \widetilde{x}_i = \pi\frac{x_i-\underline{x}}{\overline{x} - \underline{x}},
\end{equation}
where $\underline{x} =\min_{i}x_i$ and $\overline{x}=\max_i x_i$. Similarly, we normalize $z_{1i}$'s.  

This is in the same spirit of the data transformation process as in \cite{cff2025}. Thus, $\widetilde{x}_i\in [0,\pi]$. After the data transformation \eqref{data1}, the model becomes
\begin{equation}\label{eq52}
    y_i =\widetilde{\alpha}_0+\widetilde{x}_i \widetilde{\beta}_0 + \varepsilon_i
\end{equation}
where $\widetilde{\beta}_0 =  \frac{1}{\pi}(\overline{x} - \underline{x})\beta_0$, $\widetilde{\alpha}_0 =\underline{x}\beta_0  + \alpha_0$, and $\varepsilon_i =\widetilde{m}(\widetilde{x}_i)+e_i$ with $\widetilde{m}(\widetilde{x}_i) = m(\widetilde{x}_i \cdot \frac{\overline{x} - \underline{x}}{\pi}+ \underline{x}) =m(x_i)=\e[\varepsilon_i|x_i]$. 

Before examining the estimation results, we restrict the admissible confounding component to a sparse, smooth nonlinear function of rescaled education, while reserving the affine component for the structural relationship of interest. 

In the rest of this section, we also employ the following orthonormal series: {\small $\{\phi_0({x}) =1, \phi_j({x})= \sqrt{\frac{2}{\pi}} \, \cos( j {x}): j\geq 1\}$}. As discussed below Theorem \ref{MTPclt} and then in Appendix B.1.2, we need not require such an orthonormality under $\e[\phi_j^2(\mathbf{x})]=1$ and $\e[\phi_j(\mathbf{x}) \, \phi_{\ell}(\mathbf{x})]=0$ for $j\neq \ell$ when dealing with real data problems. 

The cosine dictionary is a natural approximation device on the compact interval $([0,\pi])$: its elements are bounded, computationally stable, and capable of representing a broad range of smooth nonlinear departures.

 For the purpose of comparison, we consider the following methods (with their implementation being available from using the R--package at \url{https://cran.r-project.org/package=SIVMethod}):

(i) SP method; \, 

(ii) LS method: (a)  regressing $y$  on $(1, \widetilde{x})$; \, (b) regressing $y$  on $(1, \widetilde{x}, z_{1}, z_{1}^2)$; \, 

(iii) 2SLS for model \eqref{eq52} under the following scenarios: 

(a) Just identified: using $z_j$ as IV only for $j=1,2,\ldots, 5$ respectively; 

(b) Over identified: using two IVs $z_{1}$ and $z_{j}$, where $j=2,3,4,5$; 

(c) Over identified: using $z_{1}$ and $\phi_2(\cdot)$, $\phi_3(\cdot)$ and $\phi_8(\cdot)$, selected by the SP method, as IVs;   

(d) Over identified: using $\phi_2(\cdot)$, $\phi_3(\cdot)$ and $\phi_8(\cdot)$ only as IVs.
\smallskip

We then consider the case where the level of the dad's education is also included in $\mathbf{x}$:
\be
y_i =\widetilde{\alpha}_0 +\widetilde{x}_i\widetilde{\beta}_0 + z_{1i} \beta_1 +\varepsilon_i \ \ \mbox{with} \ \ \varepsilon_i=m(\widetilde{x}_i,z_{1i})+e_i,
\label{SP*}
\ee
where $\mathbf{x} = (x, z_1)^{\top}$, and $x$ and $z_1$ denote the child's and father's years of education, respectively. We also estimate this model using the SP method. We select $\mathcal{S}$ from $\{\phi_0(\widetilde{x}),\phi_1(\widetilde{x}),\ldots \}\otimes \{\phi_0(z),\phi_1(z),\ldots \}$. The proposed LASSO method selects $\phi_1(\widetilde{x}) \, \phi_6(z)$, $\phi_2(\widetilde{x})$, and $\phi_3(\widetilde{x}) \, \phi_1(z)$.  

The selection of the basis functions is somewhat consistent throughout this empirical analysis, as we also identify $\phi_2(\widetilde{x})$ in the LASSO procedure. The estimated $m(x,z)$ is $\widehat{m}(x,z)= 0.0680 \, \phi_1(\widetilde{x})\, \phi_6(z)-0.0678 \, \phi_2(\widetilde{x}) -0.1232\, \phi_3(\widetilde{x}) \, \phi_1(z)$. We refer to this model as SP* in what follows, and still focus on reporting $\widetilde{\beta}_0$. For each method, we report the estimated value of $\widetilde{\beta}_0$ and its 95\% confidence interval (referred to CI in Table \ref{tb.em.2}). Moreover, we report the in-sample root mean squared errors of each method: {\footnotesize $\text{RMSE} =\sqrt{\frac{1}{n}\sum_{i=1}^n (y_i-\widehat{y}_i)^2}$}.

\begin{table}[htb!]
\centering \small
\renewcommand{\arraystretch}{0.6}
\caption{Estimation Results}\label{tb.em.2} 

\begin{tabular}{llrrrrrr}
\toprule
 &  & $\widetilde{\beta}_0$ & CI  & RMSE & $\beta_0$ & CI   \\
SP &  & 0.3604 & (0.2949, 0.4258) & 0.4131 & 0.0708 & (0.0579, 0.0836) \\
SP* &  & 0.3837  & (0.3233, 0.4441) &  0.4120 & 0.0753 & (0.0635, 0.0872) \\
LS & (a) & 0.3096 & (0.2790, 0.3402) & 0.4143 & 0.0608 & (0.0548, 0.0668) \\
 & (b) &  0.3017 & (0.2682, 0.3351) & 0.4137 & 0.0592 & (0.0527, 0.0658) \\
2SLS & $z_1$ & 0.3665 & (0.2903, 0.4427) & 0.4154 & 0.0720 & (0.0570, 0.0869) \\
 & $z_2$ & 0.9140 & (0.6396, 1.1883) & 0.5189 & 0.1795 & (0.1256, 0.2333) \\
 & $z_3$ & 1.5497 & (0.6224, 2.4769) & 0.7633 & 0.3043 & (0.1222, 0.4863) \\
 & $z_4$ & 0.8551 & (0.6177, 1.0925) & 0.5012 & 0.1679 & (0.1213, 0.2145) \\
 & $z_5$ & 0.5692 & (-0.3972, 1.5355) & 0.4355 & 0.1118 & (-0.0780, 0.3015) \\
 & $z_1, z_2$ & 0.4084 & (0.3338, 0.4830) & 0.4175 & 0.0802 & (0.0655, 0.0948) \\
 & $z_1, z_3$ & 0.3777 & (0.3016, 0.4539) & 0.4158 & 0.0742 & (0.0592, 0.0891) \\
 & $z_1, z_4$ & 0.4118 & (0.3376, 0.4861) & 0.4177 & 0.0809 & (0.0663, 0.0954) \\
 & $z_1, z_5$ & 0.3673 & (0.2912, 0.4434) & 0.4154 & 0.0721 & (0.0572, 0.0871) \\
 & $z_1, \phi_2,\phi_3, \phi_8$  & 0.2987 & (0.2642, 0.3331) & 0.4144 & 0.0580 & (0.0519, 0.0654) \\
 & $\phi_2,\phi_3, \phi_8$  & 0.2955 & (0.2609, 0.3301) & 0.4144 & 0.0580 & (0.0512, 0.0648) \\
\bottomrule
\end{tabular}
\end{table}

We summarize the results in Table \ref{tb.em.2}. All methods confirm a positive relationship between wage and education, and the SP method offers slightly better model fitting in terms of RMSE. The LASSO method selects $\{\phi_2(x), \phi_3(x), \phi_8(x)\}$. It is noteworthy that Table \ref{tb.em.2} provides the estimates of both $\widetilde{\beta}_0$ and $\beta_0$. After accounting for the relationship $\widetilde{\beta}_0 =  \frac{1}{\pi}(\overline{x} - \underline{x})\beta_0$, the estimates of $\beta_0$ are reported in the last column. 

Thus, taking the SP result as an example, an additional year of education is associated with an average wage increase of around 7.1\%. The SP estimate is very close to the 2SLS estimates obtained using $z_1$ alone, $(z_1, z_3)$, or $(z_1, z_5)$ as instruments.  

Table \ref{tb.em.2} also shows that the proposed SP* model accommodates externally observed variables, such as the father's education, in $\mathbf{x}$. Among SP, SP*, and LS(b), SP* has the lowest in-sample RMSE. The 2SLS estimates vary substantially across instrument sets, creating uncertainty for practical interpretation. By contrast, the SP and SP* estimates are stable across the series functions and truncation parameters considered here.

{\samepage
\begin{figure}[htb!]
    \centering
    \caption{Estimated $\widehat{m}(x)$}\label{Fig1}
    \includegraphics[scale=0.6]{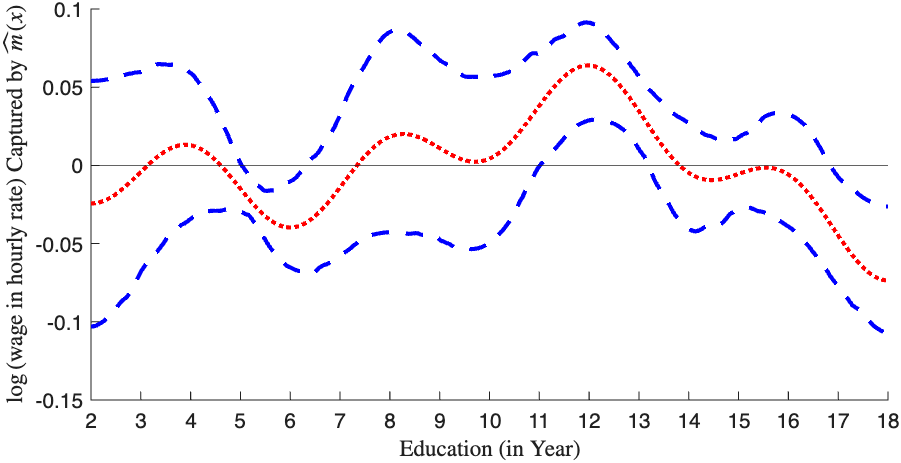}
\end{figure}
}

As additional evidence of substantial nonlinear endogeneity in the returns-to-schooling data, Figure \ref{Fig1} plots the estimated $m(x)$ with its 95\% confidence interval, obtained by wild bootstrap. The SP estimates support a nonlinear specification of $m(x)$. 

We then would like to investigate whether $\varepsilon$ could be uncorrelated with all $\{z_j: 1\leq j\leq 5\}$ used above. We compute the sample averages of the following quantities:
$\e[\widetilde{x}\cdot m(\widetilde{x})]$ and $\e[z_j\cdot \varepsilon]$, where $\widetilde{x}$, $m(\widetilde{x})$ are the same as in Example 5.1, $z_j$'s are the different IV variables adopted above, and $m(\cdot)$ and $\varepsilon = y- \widetilde{\alpha}_0- \widetilde{x}\, \widetilde{\beta}_0$ are replaced with their respective estimates via the SP method. 

\begin{table}[htb!]
\centering \small
\renewcommand{\arraystretch}{0.7}
\caption{Estimation Results}\label{tb.em.3}
\begin{tabular}{lrr} 
\toprule
 & Average & 95\% CI \\
$\e[\widetilde{x}\cdot m(\widetilde{x})]$ & 0.009 & (0.004, 0.014) \\
$\e[z_1\cdot \varepsilon]$ & 0.022 & (-0.045, 0.293) \\
$\e[z_2\cdot \varepsilon]$ & 0.027 & (0.014, 0.040) \\
$\e[z_3\cdot \varepsilon]$ & 0.024 & (0.013, 0.035) \\
$\e[z_4\cdot \varepsilon]$ & 0.026 & (0.014, 0.037) \\
$\e[z_5\cdot \varepsilon]$ & 0.002 & (-0.005, 0.008) \\
 \bottomrule
\end{tabular}
\end{table}

As shown in Table \ref{tb.em.3}, $\e[\widetilde{x}\cdot m(\widetilde{x})]$ is clearly non-zero, which justifies the existence of endogeneity. As a valid IV, we would expect that each $\e[z_j\cdot \varepsilon]$ should be close to zero. However, only $z_1$ and $z_5$ generate averages sufficiently close to 0. Numerically, it offers a reason to explain why the results associated with $z_1$ and $(z_1, z_5)$ in Table 2 are close to that from the SP approach. However, the other quantities: $\e[z_j\cdot \varepsilon]$, for $2\leq j\leq 4$, listed in Table \ref{tb.em.3}, wouldn't justify the validity of those IVs. 
\medskip

Appendix B.4 of the supplementary document further discusses about possible source of endogeneity potentially inherited in the return to schooling dataset. The proposed SP estimation method is applicable to stationary time--series data, subject to standard stochastic mixing conditions (see, for example, Chapter 6 of \cite{hlg2000}). An additional example about market-return analysis is also included in Appendix B.4.

\section{Simulation evaluations}\label{Sec6}

\renewcommand{\theequation}{6.\arabic{equation}} 
\setcounter{equation}{0}

In this section, we consider a number of scenarios to demonstrate the finite-sample performance of the proposed SP estimation method under orthonormality in Section \ref{Sec6}.1. The corresponding evaluation of the testing theory is also given in Section \ref{Sec6}.2 below. 

To show that the orthogonality on $\{\psi_j(\cdot): j\geq 1\}$ may be relaxed, we present extensive numerical evaluations in Appendix B of the supplementary document to demonstrate that the SP method still works well numerically (with its implementation being available from using the R--package at \url{https://cran.r-project.org/package=SIVMethod}).

\subsection{Finite--sample properties of the estimation theory}

\smallskip

\noindent{\bf Example 6.1}:  We consider the following data generating processes:

\begin{itemize}
    \item Case A: $y_i = \alpha_0 +x_i \, \beta_0 + \varepsilon_i$, $(\alpha_0,\beta_0) =(1,1)$, $\varepsilon_i = m(x_i) + e_i$, $e_i\sim N(0,1)$, $x_i \sim U(0, \pi)$ and $m(x)= 2.5\cos(3 \, x)+0.5\cos(5 \, x)$;
    \item Case B: Consider Case A, but choose $m(x)$ as $m(x)=\sum_{j=1}^{\infty} \gamma_j\cos(j x)$, in which $\gamma_j =4\cdot (0.9)^j$ if $j\in \{4(\ell-1)+3\mid \ell=1,2,\ldots \}$, and $\gamma_j =0$, otherwise;
    \item Case C: $y_i =x_i^2 \, \beta_0 +\varepsilon_i$, $x_i\sim U(0, \pi)$ and $m(x_i)=12\pi (x_i-\frac{\pi }{2})$;
    \item Case D:  Consider Case A with one additional exogenous regressor $y_i = \alpha_0 +x_i \, \beta_0+w_i\beta_1 + \varepsilon_i$, where $w_i\sim N(1,1)$ and $\beta_1=1$.
\end{itemize}

In the supplementary Appendix A.5, we show that $\e[x_im(x_i)]\ne 0$ and $ \e[m(x_i)]=0$ for all cases, so the endogeneity exists. As shown in the supplementary Appendices A.5 and B.3, Case B satisfies Assumption \ref{MTP2}(iii). Because $g(x, \beta_0) \equiv: x^2 \, \beta_0$ is nonlinear in $x$, Case C is a nonlinear regression setting. Case D is a special case of model (\ref{con1}) listed in the end of Section \ref{Sec3}.1. Table \ref{TB.N1} below shows that the estimation method also works well numerically for Case D.

For the numerical exercise in this section, we use an orthonormal basis in ${\mathcal L}^2$ of the form $\{\phi_0({x}) =1, \phi_j({x})= \sqrt{2} \, \cos( j {x}): j\geq 1\}$, where $\e[\phi_j^2(\mathbf{x})]=1$ and $\e[\phi_j({x}) \, \phi_{\ell}({x})]=0$ for $j\neq \ell$. After using LASSO to identify the $\phi_j(x)$ with nonzero coefficients, we construct $\{\psi_j(\cdot): j\geq 1\}$ and apply the proposed SP method to obtain $\widehat{\beta}_{\rm SP}$. Appendix A.5 of the supplementary document demonstrates for Case C that $\{\psi_j(\cdot): j\geq 1\}$ may be constructed as a subset of $\{\phi_0({x}) =1, \phi_j({x})= \sqrt{2} \, \cos( j {x}): j\geq 1\}$.

{\footnotesize
\begin{table}[htb!]
\caption{Estimation Results} \label{TB.N1} 
\centering\small
\setlength{\tabcolsep}{6pt}
\renewcommand{\arraystretch}{0.6} 
\begin{tabular}{lrrrrrrr}
\toprule
    & $n$   & $\Delta \beta$  & $\text{sd}_\beta$ & $|\frac{\Delta \beta_{\rm SP}}{\Delta \beta_{\rm LS}}|$ & $\Delta \beta$  & $\text{sd}_\beta$ & $|\frac{\Delta \beta_{\rm SP}}{\Delta \beta_{\rm LS}}|$  \\
    &     & \multicolumn{3}{c}{Case A}  & \multicolumn{3}{c}{Case D ($\beta_0$)} \\
LS  & 200 & -0.233 & 0.168 & \multicolumn{1}{r|}{}      & -0.225   & 0.171   &             \\
    & 300 & -0.222 & 0.134 & \multicolumn{1}{r|}{}      & -0.226   & 0.137   &             \\
    & 400 & -0.234 & 0.121 & \multicolumn{1}{r|}{}      & -0.230   & 0.118   &             \\
SP & 200 & -0.008 & 0.084 & \multicolumn{1}{r|}{0.035} & -0.004   & 0.079   & 0.017       \\
    & 300 & 0.002  & 0.064 & \multicolumn{1}{r|}{0.007} & -0.003   & 0.065   & 0.013       \\
    & 400 & -0.001 & 0.055 & \multicolumn{1}{r|}{0.003} & 0.001    & 0.055   & 0.006       \\
    &     & \multicolumn{3}{c|}{Case B}                  & \multicolumn{3}{c}{Case D ($\beta_1$)} \\
LS  & 200 & -0.282 & 0.253 & \multicolumn{1}{r|}{}      & 0.003    & 0.146   &             \\
    & 300 & -0.278 & 0.199 & \multicolumn{1}{r|}{}      & 0.003    & 0.119   &             \\
    & 400 & -0.284 & 0.176 & \multicolumn{1}{r|}{}      & -0.007   & 0.106   &             \\
SP & 200 & -0.001 & 0.102 & \multicolumn{1}{r|}{0.004} & 0.000    & 0.074   & 0.028       \\
    & 300 & 0.001  & 0.084 & \multicolumn{1}{r|}{0.003} & -0.001   & 0.058   & 0.490       \\
    & 400 & -0.004 & 0.072 & \multicolumn{1}{r|}{0.014} & -0.003   & 0.050   & 0.441       \\
    &     & \multicolumn{3}{c}{Case C}                  &          &         &             \\
LS  & 200 & 11.256 & 0.232 & \multicolumn{1}{r|}{}      &          &         &             \\
    & 300 & 11.257 & 0.184 & \multicolumn{1}{r|}{}      &          &         &             \\
    & 400 & 11.252 & 0.162 & \multicolumn{1}{r|}{}      &          &         &             \\
SP & 200 & 0.160  & 0.356 & \multicolumn{1}{r|}{0.014} &          &         &             \\
    & 300 & 0.132  & 0.285 & \multicolumn{1}{r|}{0.012} &          &         &             \\
    & 400 & 0.092  & 0.225 & \multicolumn{1}{r|}{0.008} &          &         &         \\
\bottomrule
\end{tabular}
\end{table}}

We repeat the simulation procedure $R=1000$ times and report the following measures in Table \ref{TB.N1}: $\Delta \beta = \overline{\beta}  - \beta_0$ and $\text{sd}_\beta = \sqrt{\frac{1}{R}\sum_{r=1}^R (\widehat{\beta}_r - \overline{\beta})^2}$, where $\overline{\beta} = \frac{1}{R}\sum_{r=1}^R\widehat{\beta}_r$ with $\widehat{\beta}_r \in\{\widehat{\beta}_{\rm LS},\widehat{\beta}_{\rm SP}\}$ in each replication. 

Additionally, we calculate $\left|\frac{\Delta \beta_{\rm SP}}{\Delta \beta_{\rm LS}}\right|$ in Table \ref{TB.N1} in order to show the relative magnitude of the biases associated with both OLS and SP methods. As indicated in Table \ref{TB.N1}, the OLS method is apparently biased and has large standard deviation. The SP method works well for all cases. Although the magnitudes of biases vary in different cases, the ratio $\left|\frac{\Delta \beta_{\rm SP}}{\Delta \beta_{\rm LS}}\right|$ is reasonably stable across all cases. The only exception is the ratio associated with $\beta_1$ of Case D, which should be expected due to the fact that $w_i$ is an exogenous variable.  

{\footnotesize
\begin{table}[htb!]
    \caption{Selection Results} \label{TB.N2}
    \centering \small
\setlength{\tabcolsep}{6pt}
\renewcommand{\arraystretch}{0.6}
\begin{tabular}{lrrrrrrrrrrrr}
\toprule
        & \multicolumn{3}{c}{Case A}                 & \multicolumn{3}{c}{Case B}                 & \multicolumn{3}{c}{Case C}                 & \multicolumn{3}{c}{Case D} \\
$n$       & 200   & 300   & 400                        & 200   & 300   & 400                        & 200   & 300   & 400                        & 200     & 300     & 400    \\
\hline
$P_1$  & 0.00 & 0.00 & \multicolumn{1}{r|}{0.00} & 0.00 & 0.00 & \multicolumn{1}{r|}{0.00} & 1.00 & 1.00 & \multicolumn{1}{r|}{1.00} & 0.00    & 0.00    & 0.00   \\
$P_2$  & 0.01 & 0.00 & \multicolumn{1}{r|}{0.00} & 0.03 & 0.02 & \multicolumn{1}{r|}{0.01} & 0.25 & 0.17 & \multicolumn{1}{r|}{0.13} & 0.01    & 0.00    & 0.00   \\
$P_3$  & 1.00 & 1.00 & \multicolumn{1}{r|}{1.00} & 1.00 & 1.00 & \multicolumn{1}{r|}{1.00} & 1.00 & 1.00 & \multicolumn{1}{r|}{1.00} & 1.00    & 1.00    & 1.00   \\
$P_4$ & 0.01 & 0.00 & \multicolumn{1}{r|}{0.00} & 0.02 & 0.02 & \multicolumn{1}{r|}{0.01} & 0.02 & 0.01 & \multicolumn{1}{r|}{0.00} & 0.01    & 0.00    & 0.00   \\
$P_5$  & 0.81 & 0.96 & \multicolumn{1}{r|}{0.98} & 0.03 & 0.01 & \multicolumn{1}{r|}{0.00} & 1.00 & 1.00 & \multicolumn{1}{r|}{1.00} & 0.86    & 0.95    & 0.98   \\
$P_6$  & 0.01 & 0.00 & \multicolumn{1}{r|}{0.00} & 0.03 & 0.01 & \multicolumn{1}{r|}{0.01} & 0.01 & 0.01 & \multicolumn{1}{r|}{0.00} & 0.01    & 0.00    & 0.00   \\
$P_7$  & 0.00 & 0.00 & \multicolumn{1}{r|}{0.00} & 1.00 & 1.00 & \multicolumn{1}{r|}{1.00} & 1.00 & 1.00 & \multicolumn{1}{r|}{1.00} & 0.00    & 0.00    & 0.00   \\
$P_8$  & 0.00 & 0.00 & \multicolumn{1}{r|}{0.00} & 0.03 & 0.01 & \multicolumn{1}{r|}{0.01} & 0.02 & 0.01 & \multicolumn{1}{r|}{0.00} & 0.00    & 0.00    & 0.00   \\
$P_9$  & 0.00 & 0.00 & \multicolumn{1}{r|}{0.00} & 0.02 & 0.01 & \multicolumn{1}{r|}{0.01} & 0.97 & 1.00 & \multicolumn{1}{r|}{1.00} & 0.00    & 0.00    & 0.00   \\
$P_{10}$ & 0.00 & 0.00 & \multicolumn{1}{r|}{0.00} & 0.03 & 0.01 & \multicolumn{1}{r|}{0.01} & 0.02 & 0.01 & \multicolumn{1}{r|}{0.00} & 0.00    & 0.00    & 0.00   \\
$P_{11}$ & 0.00 & 0.00 & \multicolumn{1}{r|}{0.00} & 1.00 & 1.00 & \multicolumn{1}{r|}{1.00} & 0.70 & 0.85 & \multicolumn{1}{r|}{0.92} & 0.00    & 0.00    & 0.00   \\
$P_{12}$ & 0.00 & 0.00 & \multicolumn{1}{r|}{0.00} & 0.03 & 0.01 & \multicolumn{1}{r|}{0.01} & 0.02 & 0.01 & \multicolumn{1}{r|}{0.00} & 0.00    & 0.00    & 0.00   \\
$P_{13}$ & 0.00 & 0.00 & \multicolumn{1}{r|}{0.00} & 0.02 & 0.02 & \multicolumn{1}{r|}{0.01} & 0.36 & 0.49 & \multicolumn{1}{r|}{0.57} & 0.00    & 0.00    & 0.00   \\
$P_{14}$ & 0.00 & 0.00 & \multicolumn{1}{r|}{0.00} & 0.01 & 0.01 & \multicolumn{1}{r|}{0.00} & 0.01 & 0.01 & \multicolumn{1}{r|}{0.00} & 0.00    & 0.00    & 0.00   \\
$P_{15}$ & 0.00 & 0.00 & \multicolumn{1}{r|}{0.00} & 0.96 & 1.00 & \multicolumn{1}{r|}{1.00} & 0.16 & 0.20 & \multicolumn{1}{r|}{0.23} & 0.00    & 0.00    & 0.00   \\
\bottomrule
\end{tabular}
\end{table}
}

We now examine the LASSO selection method. Table \ref{TB.N2} reports $P_j =\frac{1}{R}\sum_{r=1}^R I(\phi_j\in {\mathcal S}_r)$, where $\mathcal{S}_r$ denotes the set selected by LASSO in the $r$th replication. Thus, $P_j$ is the frequency with which $\phi_j$ is selected across the $R$ replications. The results indicate that the LASSO method performs well. 

The derivations for Example 6.1 in Appendix A.5 of the supplementary document show that the odd-indexed $\{\phi_j: j= 2\ell -1; \ell\geq 1\}$ have nonzero coefficients, whereas the even-indexed $\{\phi_j: j= 2\ell; \ell\geq 1\}$ have zero coefficients. Accordingly, Table \ref{TB.N2} shows high selection frequencies for the relevant odd-indexed terms and very low frequencies for the even-indexed terms.  

\subsection{Finite--sample properties of the testing theory}

\noindent{\bf Example 6.2}: We consider Cases A and B of Example 6.1, with the following modification to evaluate size and power.

\begin{eqnarray*}
    y_i = x_i \, \beta_0 + \varepsilon_i \quad \mbox{with} \quad \varepsilon_i = a_{nj} \, \Delta(x_i) + e_i,\notag
\end{eqnarray*}
where $x_i \sim U(0, \pi)$, $a_{n,j} =c_j \, \sqrt{\frac{k_{\max}}{n}}$ for $1\leq j\leq 4$, and the functional form of $\Delta(\cdot)$ is chosen in Cases A and B below. We set $c_1\equiv 0$ to evaluate size and use $c_2 \equiv 0.75$, $c_3 \equiv 1$, and $c_4 \equiv 1.25$ to evaluate local power. The basis functions are $\{\psi_{j}(x)=\sqrt{2/\pi}\cos(jx): j \ge 1 \}$.

When calculating the test statistic, we chose $k_{\min}=\lceil \log \log n\rceil$ and $k_{\max}=\lceil 3 \log \log n\rceil$. By construction, we have $a_{n,j}^2 \, n \, k_{\max}^{-1/2} \asymp \sqrt{\log\log n }\to \infty$, so the requirement to ensure local power properties is fulfilled. The functional form of $\Delta(x)$ is chosen as follows:
\begin{itemize}
    \item Case A: $\Delta(x)= 4\cos(3 x)+4\cos(5 x)$;
    \item Case B: $\Delta(x)=\sum_{j=1}^{\infty} \gamma_j\cos(j x)$ in which $\gamma_j =4\cdot (0.9)^j$ if $j\in \{4(\ell-1)+3\mid \ell=1,2,\ldots\}$; $\gamma_j =0$, otherwise.
\end{itemize}

To improve the finite--sample performance of the size in each case, we use a standard wild bootstrap to obtain critical values.
\begin{enumerate}
    \item Obtain $\widetilde{e}_i$ as in Section 3.2 and generate the bootstrap sample $\{y_i^* \}$, where $y_i^* \equiv x_i \, \widehat{\beta}+ \widetilde{e}_i \eta_i$ and the $\eta_i$ are i.i.d. draws from $N(0,1)$.

    \item Given $\{y_i^*, x_i\}$, run the regression to obtain $\{\widetilde{e}_i^*\}$ as in Step 1 and calculate the corresponding bootstrap statistic $L_n^*$.

    \item We repeat the above steps, say, 400 times to obtain the 95\% coverage set $\mathscr{C}_n$.
\end{enumerate}

After $R=1000$ simulation replications, we report the rejection rate for $n=200,300,400$: $\text{RJ} =\frac{1}{R} \sum_{r=1}^R I(L_{n,r}\not \in \mathscr{C}_{n,r}),$ where $L_{n,r}$ and $\mathscr{C}_{n,r}$ denote the values of $L_n$ and $\mathscr{C}_n$, respectively, in the $r$th replication. Table \ref{Tab_test} summarizes the results.

\begin{table}[htb!]\centering\small
\renewcommand{\arraystretch}{0.8}
\caption{Rejection Rates}\label{Tab_test}
\begin{tabular}{lrlrrrlrrr}
\toprule
    &       &  & \multicolumn{3}{c}{Case A} &  & \multicolumn{3}{c}{Case B}     \\ \cline{4-6} \cline{8-10}
$n$   & $c_1=0$  &  & $c_2$      & $c_3$      & $c_4$     &  & $c_2$      & $c_3$      & $c_4$       \\
200 & 0.041 &  & 0.999 & 1.000 & 1.000 &  & 0.709 & 0.944 & 0.992 \\
300 & 0.045 &  & 0.997 & 1.000 & 1.000 &  & 0.738 & 0.957 & 0.997 \\
400 & 0.049 &  & 0.999 & 1.000 & 1.000 &  & 0.741 & 0.951 & 0.996 \\
\bottomrule
\end{tabular}
\end{table}

In Table \ref{Tab_test}, the rejection rates with $c_1=0$ are close to the nominal rate of 5\%. For $c_2$, $c_3$, and $c_4$, the rejection rates describe the local power of the proposed test. Local power varies between the two cases. In Case B, the rejection rates are slightly lower than in Case A, possibly because of the accuracy and convergence rate of the truncation terms. 

In summary, Table \ref{Tab_test} shows that the proposed test can detect endogeneity of order $n^{-1/2} \, \delta_n$, where $\delta_n$ diverges to $\infty$ at the slowly varying rate $\log(\log(n))$.

\section{Conclusions and Extensions}\label{Sec7}

\renewcommand{\theequation}{7.\arabic{equation}} 
\setcounter{equation}{0}

This paper proposes the SP method for addressing certain types of endogeneity in linear models. We establish the corresponding estimation and testing theory and evaluate the methods using simulated and real data. The main finite--sample results are given in Sections \ref{Sec5} and \ref{Sec6}. Appendix B of the supplementary document provides additional numerical evidence across a wider range of finite--sample scenarios, complementing the large--sample theory established in Sections \ref{Sec2}--\ref{Sec4}. 

We then briefly discuss that the proposed SP method is applicable to deal with some nonlinear and non--separable models associated with endogeneity.

\subsection{Nonlinear regression}

Note that when we specify $\e[\varepsilon|\mathbf{x}] =m(\mathbf{x}, \bm{\gamma}_0)$ parametrically, model (\ref{a4}) becomes a parametrically nonlinear model of the form:
\be
y = \mathbf{x}^{\top}\bm{\beta}_0 + m(\mathbf{x}, \bm{\gamma}_0) + e,
\label{extend1}
\ee
where we assume that
\be
 \bm{\Sigma}_{gm} \equiv \e\left[\mathbf{x}\, \mathbf{x}^{\top}\right]  -  \e\left[\mathbf{x} \, \mathbf{m}_1^{\top}(\mathbf{x}, \bm{\gamma}_0)\right] \bm{\Sigma}_{mm}^{-1} \e\left[\mathbf{m}_1(\mathbf{x}, \bm{\gamma}_0)\, \mathbf{x}^{\top})\right]>{\bf 0},
\label{extend2}
\ee
in which $\bm{\Sigma}_{mm} = \e\left[\mathbf{m}_1(\mathbf{x}, \bm{\gamma}_0)\, \mathbf{m}_1^{\top}(\mathbf{x}, \bm{\gamma}_0)\right]>{\bf 0}$ with $\mathbf{m}_1(\mathbf{x}, \bm{\gamma}_0) =\frac{\partial m(\mathbf{x}, \bm{\gamma})}{\partial \bm{\gamma}}|_{\bm{\gamma} = \bm{\gamma}_0}$.  

It can be shown that equation (\ref{extend2}) is required for the identifiability of $(\bm{\beta}_0, \bm{\gamma}_0)$. 

In a similar way to that of ${\bm\Sigma}_g$ involved in Theorem \ref{MTPclt}, the corresponding covariance matrix becomes
\be
    \mathbf{\Omega}_{gm} \equiv: \begin{pmatrix}
        \e\left[\mathbf{x}\, \mathbf{x}^{\top}\right]  & \e\left[\mathbf{x} \, \mathbf{m}_1^{\top}(\mathbf{x}, \bm{\gamma}_0)\right] \\
        \e\left[\mathbf{m}_1(\mathbf{x}, \bm{\gamma}_0) \, \mathbf{x}^{\top}\right] & \e\left[\mathbf{m}_1(\mathbf{x}, \bm{\gamma}_0) \, \mathbf{m}_1^{\top}(\mathbf{x}, \bm{\gamma}_0)\right]
    \end{pmatrix}^\top,
   \label{extend3} 
\ee
which is invertible under Condition (\ref{extend2}). The corresponding estimation theory follows similarly from that of the standard nonlinear LS method.

Meanwhile, our discussion for models (\ref{a4}) and (\ref{extend1}) can be generalized to a wide class of nonlinear models of the form:
\be
y = g(\mathbf{x}, \bm{\theta}_0) + m(\mathbf{x}) + e,
\label{extend4}
\ee
where $g(\cdot, \bm{\theta}_0)$ is a parametrically known function indexed by a vector of unknown parameters, $\bm{\theta}_0$, and $m(\cdot)$ can be specified either nonparametrically as in (\ref{a4}) or parameterically as in (\ref{extend1}).

\subsection{Semiparametric regression}

Observe that model: 
\be
y = \mathbf{x}_u^{\top} \bm{\alpha}_0 + q(\mathbf{x}_v) + \varepsilon
\label{con2}
\ee
may be motivated as follows:
\be
y = \mathbf{x}_u^{\top} \bm{\alpha}_0 + \eta \ \ \mbox{with $\e[\eta]=0$ but $\e[\eta \, \mathbf{x}_u]\neq {\bf 0}$},
\label{seminon1}
\ee
where the proposed SP method gives us the following model:
\be
y = \mathbf{x}_u^{\top} \bm{\alpha}_0 + \eta \ \ \mbox{with $\eta = \e[\eta|\mathbf{x}_v] + e$},
\label{seminon2}
\ee
where $\mathbf{x}_v$ is either an observed variable or an IV available to the econometrician satisfying
\be
\e[e|(\mathbf{x}_u, \mathbf{x}_v)] = 0
\label{seminon3}
\ee
under which we may identify and estimate $\bm{\alpha}_0$ consistently in the same way as has been done in the relevant literature. We are interested in the case where Condition (\ref{seminon3}) may not be true.

When condition (\ref{seminon3}) is not satisfied, we consider model (\ref{con2}) for the case where 
\be
\e[\varepsilon]=0, \ \ \e[\varepsilon|(\mathbf{x}_u, \mathbf{x}_v)]\neq 0 \ \ \ \mbox{and} \ \ \ \e[\varepsilon|\mathbf{x}_v]\neq 0.
\label{con12}
\ee

Note that the proposed approach below is valid regardless of whether $\e[\varepsilon|\mathbf{x}_u]\neq 0$ or $\e[\varepsilon|\mathbf{x}_u]=0$. Let $q_1(\mathbf{x}_v) = \e[y|\mathbf{x}_v]$, $\mathbf{q}_2(\mathbf{x}_v) =\e[\mathbf{x}_u|\mathbf{x}_v]$, $\widetilde{\mathbf{x}}_u=\mathbf{x}_u - \mathbf{q}_2(\mathbf{x}_v)$, $\widetilde{y}=y - q_1(\mathbf{x}_v)$, $q_3(\mathbf{x}_v) = \e[\varepsilon|\mathbf{x}_v]$, $\eta = \varepsilon -  q_3(\mathbf{x}_v)$, $m(\widetilde{\mathbf{x}}_u) = \e[\eta|\widetilde{\mathbf{x}}_u]$ and $e=\eta - m(\widetilde{\mathbf{x}}_u)$. Then model (\ref{con1}) can be written as
\be
\widetilde{y} = y - q_1(\mathbf{x}_v) = (\mathbf{x}_u - \mathbf{q}_2(\mathbf{x}_v))^{\top} \bm{\alpha}_0 + \varepsilon - q_3(\mathbf{x}_v) = \widetilde{\mathbf{x}}_u^{\top} \bm{\alpha}_0 + m(\widetilde{\mathbf{x}}_u) + e
\label{con13}
\ee
with $\e[e|\widetilde{\mathbf{x}}_u]=0$. Notationally, model (\ref{con13}) is the same as model (2.3) with $(y, \mathbf{x})$ being replaced by $(\widetilde{y}, \widetilde{\mathbf{x}}_u)$, respectively.  The main parameter--of-interest, $\bm{\alpha}_0$, can then be estimated in the same way as in Section \ref{Sec3}.1 based on a sampling version of the form:
\be
\widetilde{y} = \widetilde{\mathbf{x}}_u^{\top} \bm{\alpha}_0 + m(\widetilde{\mathbf{x}}_u) + e,
\label{con14}
\ee
where $\widetilde{y} = y - \widehat{q}_1(\mathbf{x}_v)$ and $\widetilde{\mathbf{x}}_u = \mathbf{x}_u - \widehat{\mathbf{q}}_2(\mathbf{x}_v)$, in which $\widehat{q}_1(\cdot)$ and $\widehat{\mathbf{q}}_2(\cdot)$ can be constructed as in Section 3, before $m(\cdot)$ can be estimated in the same way as in Section 3.

It is noted that our discussion covers various data structures for $(\mathbf{x}_u, \mathbf{x}_v)$, including the case where $\mathbf{x}_u$ is a binary variable, and $\mathbf{x}_v$ is a vector of either fixed--design or stationary time--series variables.

\subsection{Binary regression}

Regression models with binary dependent variables have many theoretical investigations and empirical studies, and the relevant literature is comprehensive, such as \cite{lewbel1998}, and \cite{blundell2004}, for example. To show that the proposed SP method is also useful to address certain types of endogeneity involved in binary models, we start with the following linear probability model:
\be
{\rm P}[y=1|\mathbf{x}] =  \mathbf{x}^{\top} \bm{\beta}_0 + \varepsilon,
\label{jiti4.1}
\ee
where $\e[\varepsilon \, \mathbf{x}] \neq 0$. We then rewrite model (\ref{jiti4.1}) as
\be
y =  \mathbf{x}^{\top} \bm{\beta}_0 + m(\mathbf{x}) + e,
\label{jiti4.2}
\ee
where $m(\mathbf{x}) = \e[\varepsilon|\mathbf{x}]\neq 0$, $e = \varepsilon - m(\mathbf{x})$, and $\e[e|\mathbf{x}]=0$ by definition.  Model (\ref{jiti4.2}) can then be estimated in the same way as in Section \ref{Sec3}.1, although $y$ is now a binary response variable. 

Meanwhile, we consider a non--separable binary model of the form: 
\be
y = {\rm I}[\mathbf{x}^{\top} \bm{\beta}_0 - \varepsilon\geq 0] =  {\rm I}[\mathbf{x}^{\top} \bm{\beta}_0 - m(\mathbf{x}) \geq e],
\label{jiti4.3}
\ee
which can be rewritten as $\e[y|\mathbf{x}] = {\rm P}_{e}(e\leq \mathbf{x}^{\top} \bm{\beta}_0 - m(\mathbf{x})) = F_{e}(\mathbf{x}^{\top} \bm{\beta}_0 - \mathbf{v}(\mathbf{x})^{\top} \bm{\gamma})$ when ignoring the approximation error term and $m(\mathbf{x})$ is replaced by $\mathbf{V}_k(\mathbf{x})^{\top} \, \bm{\gamma}$ as in Section \ref{Sec3}.1, where $F_{e}(\cdot)$ denotes the cumulative distribution function (CDF) of $e$, and $(\mathbf{V}_k(\mathbf{x}), \bm{\gamma})$ is the same as defined in Section \ref{Sec3}.1.

When $F_{e}(\cdot)$ is parametrically known, the parameter--of--interest, $\bm{\beta}_0$, can then be estimated consistently by MLE. When $F_e(\cdot)$ is nonparametrically unknown, it can be expanded by an infinite sum of known CDFs before a semiparametric MLE estimation method may be developed.

Because these models in (\ref{extend4}), (\ref{con14}) and (\ref{jiti4.3}) raise distinct identification issues and require separate estimation and testing theory, we leave their treatments for future work.

{\footnotesize 
\bibliography{Refs1.bib}

@article{M1997,
 author = {A. Craig MacKinlay},
 journal = {Journal of Economic Literature},
 number = {1},
 pages = {13--39},
 title = {Event Studies in Economics and Finance},
 volume = {35},
 year = {1997}
}

@book{campbell1998econometrics,
  title={The Econometrics of Financial Markets},
  author={Campbell, John Y and Lo, Andrew W and MacKinlay, A Craig},
  year={1997},
  publisher={Princeton University Press}
}

@article{ns2021,
    author = {Whitney K. Newey and Sami Stouli},
    title = "Control variables, discrete instruments, and identification of structural functions",
    journal = {Journal of Econometrics},
    volume = {222},
    number = {1},
    pages = {73--88},
    year = {2021}
}

@article{ts2023,
    author = {Tianyu Zhang and Noah Simon},
    title = "Regression in tensor product spaces by the method of sieves",
    journal = {Electronic Journal of Statistics},
    volume = {17},
    number = {4},
    pages = {3660--3727},
    year = {2023}
}

@article{nagaev1979large,
	author = {S. V. Nagaev},
  title = "Large Deviations of Sums of Independent Random Variables",
  volume = {7},
  journal = {Annals of Probability},
  number = {5},
  pages = {745-789},
  year = {1979},
}

@article{BRT2009,
  title={Simultaneous analysis of {LASSO} and Dantzig selector},
  author={Bickel, P. J. and Ritov, Y. and Tsybakov, A. B.},
  year={2009},
  Journal = {Annals of Statistics},
  Volume = {37},
  Number = {4},
  Pages = {1705-1732}
}

@article{raskutti10a,
  author  = {Garvesh Raskutti and Martin J. Wainwright and Bin Yu},
  title   = {Restricted Eigenvalue Properties for Correlated Gaussian Designs},
  journal = {Journal of Machine Learning Research},
  year    = {2010},
  volume  = {11},
  number  = {78},
  pages   = {2241-2259}
}

@article{gao2002,
    author = {Jiti Gao and Howell Tong and Rodney Wolff},
    title = "{Adaptive orthogonal series estimation in additive stochastic regression models}",
    journal = {Statistica Sinica},
    volume = {11},
    number = {4},
    pages = {409--428},
    year = {2002}
}

@article{gao2006,
    author = {Jiti Gao and Zudi Lu and Dag Tj{\o}stheim},
    title = "{Estimation in semiparametric spatial regression}",
    journal = {Annals of Statistics},
    volume = {34},
    number = {3},
    pages = {1395--1435},
    year = {2006}
}

@article{sp2001,
    author = {Enrique Sentana and Gabriele Fiorentini},
    title = "{Identification, estimation and testing of conditionally heteroskedastic factor models}",
    journal = {Journal of Econometrics},
    volume = {102},
    number = {2},
    pages = {143--164},
    year = {2001}
}

@article{lewis2025,
    author = {Daniel J. Lewis},
    title = "{Identification based on higher moments in macroeconometrics}",
    journal = {Annual Review of Economics},
    volume = {17},
    number = {1},
    pages = {665--693},
    year = {2025}
}

@article{RR2003,
    author = {Roberto Rigobon},
    title = "{Identification through heteroskedasticity}",
    journal = {The Review of Economics and Statistics},
    volume = {85},
    number = {4},
    pages = {777--792},
    year = {2003}
}

@article{AL2012,
    author = {Arthur Lewbel},
    title = "{Using heteroscedasticity to identify and estimate mismeasured and endogenous regressor models}",
    journal = {Journal of Business and Economic Statistics},
    volume = {30},
    number = {1},
    pages = {67--80},
    year = {2012}
}

@article{BCCK2015,
    author = {A. Belloni and V. Chernozhukov and D. Chetverikov and K. Kato},
    title = "{Some new asymptotic theory for least squares: pointwise and uniform results}",
    journal = {Journal of Econometrics},
    volume = {186},
    number = {4},
    pages = {345--366},
    year = {2015}
}

@article{linton1995,
author  = "Oliver Linton",
title   = "{Second order approximation in a partially linear regression model}",
journal = "Econometrica",
year    = "1995",
volume  = "63",
number  = "4",
pages   = "1079--1113"
}

@article{Andrews1991,
  author = {Donald W. Andrews},
  journal = {Econometrica},
  number = {2},
  pages = {307--345},
  title = {Asymptotic normality of series estimators
  for non-- and semi--parametric regression models},
  volume = {59},
  year = {1991}
}

@article{hausman1978,
  author = {Jerry A. Hausman},
  journal = {Econometrica},
  number = {6},
  pages = {1251--1271},
  title = {Specification tests in econometrics},
  volume = {46},
  year = {1978}
}

@article{IN2009,
  author = {Guido W. Imbens and Whitney K. Newey},
  journal = {Econometrica},
  number = {2},
  pages = {467--475},
  title = {Identification and estimation of triangular simultaneous equations models without additivity},
  volume = {77},
  year = {2009}
}

@article{JMW2015,
  author = {Jeffrey M. Wooldridge},
  journal = {Journal of Human Resources},
  number = {2},
  pages = {420--445},
  title = {Control function methods in applied econometrics},
  volume = {50},
  year = {2015}
}

@incollection{chen2007,
title={{Large Sample Sieve Estimation of Semi--Nonparametric Models}},
author={Xiaohong Chen},
booktitle={Chapter 76 of Handbook of Econometrics (edited by J. J. Heckman
and E. E. Leamer)},
volume={6B},
pages={5549--5632},
year={2007},
publisher={Elsevier, New York}
}

@book{dudley2003,
	author={R. M. Dudley},
	title={{Real Analysis and Probability}},
	year={2003},
	publisher={Cambridge University Press, London}
}

@book{dg2025,
	author={Chaohua Dong and Jiti Gao},
	title={{Modern Series Methods in Econometrics and Statistics}},
	year={2025},
	publisher={Springer, New York}
}

@book{vdv1998,
	author={van der Vaart},
	title={{Asymptotic Statistics}},
	year={1998},
	publisher={Cambridge University Press, London}
}

@article{cff2025,
 title={Iterative estimation of nonparametric regressions with
 continuous endogenous variables and discrete instruments},
  author={Centorrino, Samuele and F{\'e}ve, Fr{\'e}d{\'e}rique and
  Florens, Jean-Pierre},
  journal={Journal of Econometrics},
  volume={247},
  pages={105950},
  year={2025}
}

@article{newey2013,
 author = {Whitney K. Newey},
  title = {Nonparametric Instrumental Variables Estimation},
 journal = {The American Economic Review, Papers and Proceedings of
 the 125 Annual Meeting of the American Economic Association},
 number = {3},
 pages = {550--556},
 volume = {103},
 year = {2013}
}

@article{blundell2004,
 author = {Blundell, Richard and Powell, James},
  title = {Endogeneity in semiparametric binary response models},
 journal = {Review of Economic Studies},
 number = {3},
 pages = {655--679},
 volume = {71},
 year = {2004}
}

@article{card2001,
 author = {David Card},
  title = {Estimating the return to schooling: progress on some persistent
  econometric problems},
 journal = {Econometrica},
 number = {5},
 pages = {1127-1160},
 volume = {69},
 year = {2001}
}

@article{lewbel1998,
 author = {Lewbel, Arthur},
  title = {Semiparametric latent variable model estimation with endogenous or
  mis--measured regressors},
 journal = {Econometrica},
 number = {1},
 pages = {105-121},
 volume = {66},
 year = {1998}
}

@article{ln1995,
 author = {Linton, Oliver and Nielsen, Jens, P.},
  title = {A kernel method of estimating structured nonparametric regression
  based on marginal integration},
 journal = {Biometrika},
 number = {1},
 pages = {93-100},
 volume = {82},
 year = {1995}
}

@article{gallant1981,
 author = {Gallant, A. R.},
  title = {On the bias in flexible functional forms and an essentially
  unbiased form: The {Fourier} flexible form},
 journal = {Journal of Econometrics},
 number = {2},
 pages = {211--245},
 volume = {15},
 year = {1981}
}

@article{dlp2021,
title = {A weighted sieve estimator for nonparametric time series models with
nonstationary variables},
journal = {Journal of Econometrics},
volume = {222},
number = {3},
pages = {909-932},
year = {2021},
author = {Chaohua Dong and Oliver Linton and Bin Peng}

}

@article{dgl2023,
title = {High dimensional semiparametric moment restriction models},
journal = {Journal of Econometrics},
volume = {232},
number = {3},
pages = {320-345},
year = {2023},
author = {Chaohua Dong and Jiti Gao and Oliver Linton}

}

@article{bck2007,
 author = {Richard Blundell and Xiaohong Chen and Dennis Kristensen},
 journal = {Econometrica},
 number = {6},
 pages = {1613--1669},
 title = {Semi--nonparametric {IV} estimation of shape--invariant Engel curves},
 volume = {76},
 year = {2007}
}

@article{npv1999,
 author = {Whitney K. Newey and James Powell and Francis Vella},
 journal = {Econometrica},
 number = {3},
 pages = {565--603},
 title = {Nonparametric estimation of triangular simultaneous equations models},
 volume = {67},
 year = {1999}
}

@article{np2003,
 author = {Whitney K. Newey and James Powell},
 journal = {Econometrica},
 number = {6},
 pages = {1565--1578},
 title = {Instrumental variables estimation for nonparametric models},
 volume = {71},
 year = {2003}
}

@article{ass2019,
author = {Andrews, Isaiah and Stock, James H. and Sun, Liyang},
title = {Weak instruments in instrumental variables regression: theory and practice},
journal = {Annual Review of Economics},
volume = {11},
number = {1},
pages = {727-753},
year = {2019}
}

@article{pmr1988,
author = {Robinson, Peter M.},
title = {Root--{N} consistent semiparametric regression},
journal = {Econometrica},
volume = {56},
number = {3},
pages = {931--954},
year = {1988}
}

@book{hlg2000,
	author={Wolfgang H\"{a}rdle and Hua Liang and Jiti Gao},
	title={{Partially Linear Models}},
	year={2000},
	publisher={Springer/Verlag, New York}
}

@book{wooldridge2016,
  title={Introductory Econometrics: A Modern Approach 6th Ed.},
  author={Wooldridge, Jeffrey M},
  year={2016},
  publisher={Cengage Learning, Boston}
}

@book {peterhall1980,
    AUTHOR = {P. G. Hall and C. C. Heyde},
     TITLE = {{Martingale Limit Theory and Its Applications}},
    SERIES = {},
 PUBLISHER = {Academic Press},
   ADDRESS = {New York},
      YEAR = {1980},
     PAGES = {},
      ISBN = {},
   MRCLASS = {},
  MRNUMBER = {},
MRREVIEWER = {}}
}
\newpage

\begin{center}
{\large \bf Supplementary Document to\\ Identification, Estimation and Inference Based on \\Constructed Instruments}

\medskip

{\sc \small Chaohua Dong$^{\ast}$, Jiti Gao$^{\dag}$, Oliver Linton$^{\star}$ and Bin Peng$^{\dag}$}

{\small $^{\ast}$Zhongnan University of Economics and Law, $^{\dag}$Monash University and $^{\star}$University of Cambridge}
\end{center}

{\small

\begin{appendix}

In this supplementary document,  Appendix \ref{Ap.A1} provides some additional interpretations and justifications for the main assumptions in Sections 2--4 and then the simulation settings considered in Section 6. Appendix \ref{Ap.B} outlines that the proposed SP method works without orthogonality and also discusses some implementational issues. Appendix \ref{Ap.C} collects the proofs of the main results in Sections 3 and 4. 

\section{Interpretations and justifications}\label{Ap.A1}

\renewcommand{\theequation}{A.\arabic{equation}}

\setcounter{subsection}{0}

\subsection{Geometric interpretation of Assumption 2.1}

We have defined the Hilbert space ${\mathcal L}^2=\{g({\bf x}):\; \e[g^2({\bf x})]<\infty\}$, in which $\e[\xi_1\xi_2]$ is an inner product for $\xi_1, \xi_2\in {\mathcal L}^2$; in Assumption 2.1, we stipulate ${\mathcal S}$ as the closed proper subspace of ${\mathcal L}^2$. Accordingly, we have direct sum decomposition ${\mathcal L}^2={\mathcal S}\oplus {\mathcal S}^\bot$ where ${\mathcal S}^\bot$ is the orthogonal complement subspace of ${\mathcal S}$ that is the set of all elements from ${\mathcal L}^2$ orthogonal with each element in ${\mathcal S}$. Notice that both ${\mathcal S}$ and ${\mathcal S}^\bot$ are Hilbert space too, with the norm defined in ${\mathcal L}^2$. 

{\footnotesize
\tikzset{global scale/.style={
    scale=#1,
    every node/.append style={scale=#1}
  }
}
{
\begin{figure}[h]
\begin{center}
\begin{tikzpicture}[scale=0.4]
\begin{scope}[thick]
\draw[blue](-3,0)--(-1.5,2)--(8,2)--(6.5,0)--(-3,0);
\draw[blue](-1,1)--(6,1);
\draw[red](-1,1)--(6,6);\draw[red, dashed](6,6)--(6,1); \draw(5.8, 3.5)node[right]{${\mathcal M_S}({\bf x})$};
\draw(3, 4)node[left]{${\bf x}$};
\draw(3, 1.4)node[below]{${\mathcal P_S}({\bf x})$};
\draw(6.3, -0.1)node[above]{${\mathcal S}$};
\draw[red](-4,-1)--(-4,7)--(9.5,7)--(9.5,-1)--(-4,-1);
\draw(-4, 6)node[right]{${\mathcal L}^2$};
\end{scope}
\end{tikzpicture}
\end{center}
\caption{Projecting ${\bf x}$ into ${\mathcal S}$}\label{f1}
\end{figure}}
}

At element level, for any $\xi\in{\mathcal L}^2$, one can uniquely decompose $\xi=\xi_1+\xi_2$ with $\xi_1\in{\mathcal S}$ and $\xi_2\in{\mathcal S}^\bot$. Accordingly, two projection mappings can be defined as ${\mathcal P_S}(\xi)=\xi_1$ and ${\mathcal M_S}(\xi)=\xi_2$. Certainly, ${\mathcal M_S}\equiv {\mathcal I-P_S}$ where ${\mathcal I}$ is the identity operator.

In Assumption 2.1, the condition $\lambda^\top {\bf x}\not \in {\mathcal S}$ for any $\lambda\neq 0$ implies that ${\mathcal M_S}(\lambda^\top{\bf x})$ does not degenerate, as illustrated in Figure \ref{f1}.

Consequently, as shown in Lemma 2.1 of Section 2.2, we have the following identifiability condition:
\be
{\bm\Sigma}_x=\e[{\mathcal M_S}({\bf x}){\mathcal M_S}({\bf x})^\top]=\e[{\bf x}{\bf x}^\top] - \sum_{j=1}^\infty \e[{\bf x}\psi_j({\bf x})] \, \e[{\bf x}^\top\psi_j({\bf x})]>{\bf 0},
\label{cdx}
\ee
ensures that $\bm{\beta}_0 = {\bm\Sigma}_x^{-1} \, \bm{\Sigma}_{xy}$ is identifiable, where $\bm{\Sigma}_{xy} = \e[\mathbf{x} \, y] - \sum_{j=1}^\infty \e[{\bf x} \, \psi_j({\bf x})] \, \e[y \, \psi_j({\bf x})]$. Furthermore, Lemma 2.2 shows that $\bm{\beta}_0$ is invariant to the choice of the basis $\{\psi_j(\cdot): j\geq 1\}$.

\subsection{Probability density function of $\mathbf{x}$ being unknown}

Since the space ${\mathcal L}^2$ is defined as the set of all $g({\bf x})$ such that $\e[g^2({\bf x})]<\infty$, it is not easy to determine which function of ${\bf x}$ belongs to the space when the density of $\mathbf{x}$ is unknown. In this subsection, we show a way of how to relax the orthonormality structure imposed in Sections 2--4.

Let ${\mathcal{X}}$ be the support of ${\bf x}$ and $\pi(x)$ be a density on ${\mathcal{X}}$ satisfying $c\pi(x)\le f_{\bf x}(x)\le C\pi(x)$ for some constants $C>c>0$, where $f_{\bf x}(x)$ is the density of ${\bf x}$. Define ${\mathcal L}^2({\mathcal{X}}, \pi(x))=\{h(x): \int_{\mathcal{X}} h^2(x)\pi(x)dx<\infty\}$. It is clear that $h({\bf x})\in {\mathcal L}^2$ iff $h(x)\in {\mathcal L}^2({\mathcal{X}}, \pi(x))$. 

Accordingly,  ${\mathcal S}_{\pi}=\{g(x):\; g({\bf x})\in {\mathcal L}^2\}$ is a proper closed subspace of ${\mathcal L}^2({\mathcal{X}}, \pi(x))$. 

Moreover, suppose $\{\varphi_j(x), j=0, 1, \cdots\}$ is an orthonormal basis in ${\mathcal L}^2({\mathcal{X}}, \pi(x))$. where $\varphi_0(\cdot)$ is a constant function, and $\psi_j(x)=\varphi_{\ell(j)}(x)$ with $\ell(j)$ being a function in $j$ is an orthonormal basis in ${\mathcal S}_{\pi}$.

When $m({\bf x})=\e[\varepsilon|{\bf x}]\in {\mathcal S}$, the function $m(x)\in {\mathcal S}_{\pi}$ and therefore, $m(\mathbf{x})=\sum_{j=1}^\infty \psi_j(\mathbf{x}) \, \gamma_j$, where $\gamma_j=\int_{\mathcal{X}} m(x)\psi_j(x)\pi(x)dx$. Hence, our model is written as
\begin{align}\label{nonor1}
  y_i={\bf x}_i^\top {\bm\beta}_0+\sum_{j=1}^\infty \psi_j({\bf x}_i) \, \gamma_j +e_i.
\end{align}

For a given truncation parameter $k$, using the same notation $\mathbf{V}_k({\bf x}_i)$, ${\bm \gamma}$ and $\delta_k({\bf x}_i)$, ${\bf y, X, V, e}$ and ${\bm \delta}$ as in Section 3, we have
\begin{align}\label{nonor2}
 {\bf y}={\bf X} {\bm\beta}_0+{\bf V} {\bm\gamma}+{\bm \delta}+{\bf e},
\end{align}
which is the same as the main estimation equation in Section 3, but we no longer require ${\bf G}_k= \e[\mathbf{V}_k({\bf x}) \, \mathbf{V}_k({\bf x})^\top]$ to be an identity matrix with respect to the distribution of ${\bf x}$. 

Nonetheless, it has been shown (see, for example, Proposition 2.1 in \cite{BCCK2015}) that the eigenvalues of the matrix ${\bf G}_k = \e[\mathbf{V}_k({\bf x}) \, \mathbf{V}_k({\bf x})^\top] = \DF{1}{n}\e[{\bf V}^\top {\bf V}]$ are bounded above and away from zero, uniformly over $k$. This implies that $\DF{1}{n}{\bf V}^\top {\bf V}$ is positive definite with high probability when $n$ is large enough. 

Similar to the relevant notation introduced in Sections 2 and 3, define ${\bf M}_v={\bf I}_k - {\bf V}({\bf V}^\top {\bf V})^{-1} {\bf V}^\top$. Multiplying ${\bf M}_v$ on the both sides of \eqref{nonor2} yields
\begin{align}\label{nonor3}
 {\bf M}_v{\bf y}={\bf M}_v{\bf X} {\bm\beta}_0+{\bf M}_v({\bm \delta}+{\bf e}).
\end{align}

Then, we can estimate $ {\bm\beta}_0$ by the same $\widehat{\bm\beta}_{\rm SP}$ as defined in Section 3 when the invertibility of $\DF{1}{n}{\bf X}^\top {\bf M}_v{\bf X}$ in probability is proven. As a matter of fact, we can actually derive as $n\rightarrow \infty$
\bea
 \DF{1}{n}{\bf X}^\top {\bf M}_v{\bf X} & = & \DF{1}{n}{\bf X}^\top {\bf X}-\DF{1}{n}{\bf X}^\top {\bf V}\left(\DF{1}{n}{\bf V}^\top {\bf V}\right)^{-1} \DF{1}{n}{\bf V}^\top{\bf X}
 \nonumber\\
&=&\DF{1}{n}\sum_{i=1}^n {\bf x}_i \, {\bf x}_i^\top -\DF{1}{n}\sum_{i=1}^n {\bf x}_i \, \mathbf{V} _k({\bf x}_i)^\top \left(\DF{1}{n}{\bf V}^\top {\bf V}\right)^{-1}\DF{1}{n}\sum_{i=1}^n \mathbf{V}_k({\bf x}_i){\bf x}_i^\top
\nonumber\\
&=&\e[{\bf x} \, {\bf x}^\top]-\e[{\bf x} \, \mathbf{V}_k({\bf x})^\top] \, {\bf G}_k^{-1}\e[\mathbf{V}_k({\bf x}) \, {\bf x}^\top]+o_P(1)
\nonumber\\
& = & \e[{\bf x} \, {\bf x}^\top]-\e[{\bf x} \, \widetilde{\mathbf{V}}_k^{\top}({\bf x})] \, \e[\widetilde{\mathbf{V}}_k({\bf x}) \, {\bf x}^\top]+o_P(1) \equiv:  \bm{\Sigma}_x(k) + o_P(1)
\label{newcondition}
\eea
in the same way as in the proof of Theorem 3.1(i), where {\small $\widetilde{\mathbf{V}}_k(\mathbf{x}) = {\bf G}_k^{-1/2} \, \mathbf{V}_k(\mathbf{x})$}. 

Equation (\ref{newcondition}) implies that $\DF{1}{n}{\bf X}^\top {\bf M}_v{\bf X}$ is positive definite in probability if $\bm{\Sigma}_x(k)>{\bf 0}$, which can be shown in the same way as in the proof of Lemma 2.1 under Assumption 2.1. 

In this case,  ${\mathcal P}_{{\mathcal S},k}({\bf x})\equiv: \e\left[{\bf x} \, \mathbf{V}_k({\bf x})^\top)\right] \, {\bf G}_k^{-1} \, \mathbf{V}_k({\bf x})$ is the projection of ${\bf x}$ onto the subspace of ${\mathcal S}$ generated by $\mathbf{V}_k({\bf x})$, while ${\mathcal M}_{{\mathcal S},k}({\bf x})\equiv {\bf x}-{\mathcal P}_{{\mathcal S},k}({\bf x})$ is the residue after projection. 

To compare with the relevant discussion in Sections 2 and 3, we may show that 
\begin{align*}
    \DF{1}{n}{\bf X}^\top {\bf M}_v{\bf X}=\e\left[{\mathcal M}_{{\mathcal S},k}({\bf x}){\mathcal M}_{{\mathcal S},k}({\bf x})^\top\right]+o_P(1) = \bm{\Sigma}_x(k) + o_P(1),
\end{align*}
the left-hand side of which can be a consistent estimator of ${\bm \Sigma}_x(k)$. Consequently, Theorem 3.1 remains valid with some minor modifications on its proof.

\subsection{SP estimation efficiency}

Note that under endogeneity, the OLS estimation gives
{\footnotesize
\bea
&& \widehat{\bm{\beta}}_{\rm LS}=\widehat{S}_n^{-1} \, \sum_{i=1}^n{\bf x}_iy_i=\bm{\beta}_0+ \widehat{S}_n^{-1} \, \sum_{i=1}^n{\bf x}_i\varepsilon_i
\nonumber\\
&& = \bm{\beta}_0+ \widehat{S}_n^{-1}\, \ell_{\rm endo} + \widehat{S}_n^{-1} \, \sum_{i=1}^n[{\bf x}_im({\bf x}_i)-\ell_{\rm endo}+{\bf x}_ie_i],
\nonumber
\eea
where $\ell_{\rm endo} = \e[{\bf x}m({\bf x})]$, $\widehat{S}_n = \DF{1}{n}\sum_{i=1}^n{\bf x}_i{\bf x}_i^{\top}$, and
\be
\sqrt{n}\left[\widehat{\beta}_{\rm LS}-\bm{\beta}_0- \widehat{S}_n^{-1}\, \ell_{\rm endo} \right] = \widehat{S}_n^{-1} \, \DF{1}{\sqrt{n}}\sum_{i=1}^n[{\bf x}_im({\bf x}_i)-\ell_{\rm endo}+{\bf x}_i \, e_i].
\nonumber
\ee}

Obviously, $\widehat{\bm{\beta}}_{\rm LS}$ is biased when $\ell_{\rm endo} \neq 0$. By contrast, $\widehat{\bm{\beta}}_{\rm SP}$ is asymptotically unbiased. Meanwhile, it can be shown that the conditional covariance matrix of $\widehat{\bm{\beta}}_{\rm LS}$ is given by 
\bea
&&\widehat{S}_n^{-1} \, \DF{1}{n}\sum_{i=1}^n\e\left\{[{\bf x}_i \, m({\bf x}_i)-\ell_{\rm endo}+{\bf x}_ie_i][{\bf x}_im({\bf x}_i)-\ell_{\rm endo}+{\bf x}_ie_i]^\top|{\bf X}\right\} \, \widehat{S}_n^{-1}
\nonumber\\
&&=\widehat{S}_n^{-1} \, \DF{1}{n}\sum_{i=1}^n\left\{[{\bf x}_im({\bf x}_i)-\ell_{\rm endo}][{\bf x}_im({\bf x}_i)-\ell_{\rm endo}]^\top\right\}\, \widehat{S}_n^{-1}
\nonumber\\
&&+\widehat{S}_n^{-1}  \to_P \e^{-1} [{\bf x}{\bf x}^\top]\, {\rm var}[{\bf x} \, m({\bf x})]\, \e^{-1}[{\bf x}{\bf x}^\top] + \e^{-1}[{\bf x}{\bf x}^\top] \equiv: \bm{\Sigma}_{\rm LS},
\nonumber
\eea
due to the independence of $\{\mathbf{x}_i\}$ and $\e[e_i|{\bf x}_i]=0$ as well as letting $\sigma_e^2\equiv 1$. 

We will then need to check when $\bm{\Sigma}_{\rm LS}\geq\bm{\Sigma}_{\bf x}^{-1}$. Observe that
{\footnotesize
\be
\bm{\Sigma}_{\bf x}\bm{\Sigma}_{\rm LS}\bm{\Sigma}_{\bf x} = \bm{\Sigma}_{\bf x}\, \e^{-1}[{\bf x}{\bf x}^\top]\, {\rm var}[{\bf x}m({\bf x})]\, \e^{-1}[{\bf x}{\bf x}^\top]\, \bm{\Sigma}_{\bf x}+ \bm{\Sigma}_{\bf x} \, \e^{-1}[{\bf x}{\bf x}^\top]\, \bm{\Sigma}_{\bf x}.
\nonumber
\ee}

In order to ensure {\footnotesize $\bm{\Sigma}_{\bf x}\, \e^{-1}[{\bf x}{\bf x}^\top]\, \left\{{\rm var}[{\bf x} \, m({\bf x})]- \e[\mathcal{P}_{\mathcal S}({\bf x}){\mathcal P}_{\mathcal S}({\bf x})^\top]\bm{\Sigma}_{\bf x}^{-1}\, \e[{\bf x}{\bf x}^\top]\right\} \, \e^{-1}[{\bf x} \, {\bf x}^\top]\, \bm{\Sigma}_{\bf x}\geq {\bf 0}$}, it suffices to require
\be
{\rm var}[{\bf x}m({\bf x})]\geq \e[{\bf x}{\bf x}^\top]\, \{\e[{\mathcal M}_{\mathcal S}({\bf x}){\mathcal M}_{\mathcal S}({\bf x})^\top]\}^{-1}\e[{\mathcal P}_{\mathcal S}({\bf x}){\mathcal P}_{\mathcal S}({\bf x})^\top],
\nonumber
\ee
which reduces to requiring $\e\left[{\bf x} \, \mathbf{x}^{\top}\, m^2({\bf x})\right]\geq \e[{\bf x}{\bf x}^\top]\, \bm{\Sigma}_x^{-1} \, \sum_{j=1}^{\infty} \e[{\bf x} \, \psi_j(\mathbf{x})] \, \e\left[{\bf x}^\top \, \psi_j(\mathbf{x})\right] = {\bf 0}$ trivially when $\e[{\bf x} \, \psi_j(\mathbf{x})]={\bf 0}$ and consequently $\ell_{\rm endo}=0$. In this case, it is obvious that $\widehat{\bm{\beta}}_{\rm LS}$ is unbiased, but it is still less efficient than $\widehat{\bm{\beta}}_{\rm SP}$.

\subsection{Verification of Assumption 3.1(iv)}

Let us show a few cases where Assumption 3.1(iv) can be satisfied.

Suppose that the support of ${\bf x}$ is $[0, \pi]$ for $d=1$ or $[0, \pi]^d$ for $d>1$. In what follows we first consider the case $d=1$. Let $\varphi_0({x})\equiv 1$ and $\varphi_j({x})=\sqrt{2/\pi}\cos(j{x})$ be the orthogonal basis in $L^2[0,\pi]$. Take $\psi_j({x})=\varphi_{\ell(j)}({x})$ for some function $\ell(j)$ in $j$ specified later. Notice that $m({x})=\sum_{j=1}^\infty \gamma_j \psi_j({\bf x})$ and $\e[\delta_k^2({x})]=\sum_{j=k+1}^\infty \gamma_j^2$ under $\e[\psi_j^2(x)]=1$ and $\e[\psi_j(x)\, \psi_i(x)]=0$ for $j\neq i$. Suppose $m({x})$ is differentiable, then integration by parts gives
\begin{align}\label{ass3.1exa1}
  \gamma_j=&\int_0^{\pi} m({x})\psi_j({x})d{x}=\int_0^{\pi} m({x})\varphi_{\ell(j)}({x})d{x}\notag \\
  =&\sqrt{\frac{2}{\pi}}\int_0^{\pi} m({x})\cos(\ell(j){x})d{x}=\sqrt{\frac{2}{\pi}}\DF{1}{\ell(j)}\int_0^{\pi} m({x})d\sin(\ell(j){x})\notag\\
  =&-\sqrt{\frac{2}{\pi}}\DF{1}{\ell(j)}\int_0^{\pi} m'({\bf x})\sin(\ell(j){\bf x})d{\bf x}\notag\\
  =&\sqrt{\frac{2}{\pi}}\DF{1}{\ell(j)^2}\int_0^{\pi} m^{\prime}({x})d\cos(\ell(j){x})\notag\\
  =&\sqrt{\frac{2}{\pi}}\DF{1}{\ell(j)^2} m^{\prime}({x})\cos(\ell(j){x})\Big|_0^\pi-\sqrt{\frac{2}{\pi}}\DF{1}{\ell(j)^2}\int_0^{\pi} m^{\prime \prime}({x})\cos(\ell(j){x})d{x}\notag\\
  =&O(\ell(j)^{-2}).
\end{align}

The same derivation gives ${\lambda}_j = \e[{x} \, \psi_j({x})]=O(\ell(j)^{-2})$.

{\bf Case (1)}. Let $\ell(j)=j^{s+1}$ for some positive integer $s$. Then, $\gamma_j=Cj^{-2(s+1)}$ and hence, $\e[\delta_k^2({x})]=O(k^{-2s-1})$; \, {\bf Case (2)}. Let $\ell(j)=2^j$, and then $\e[\delta_k^2({x})]=O(4^{-k})$.

{\bf Case (3)}. If we further suppose that $m({\bf x})$ is differentiable up to $s+1$ order and $m^{(r)}(0)=m^{(r)}(\pi)=0$ for odd $r$ and $r\le s+1$, then we continue \eqref{ass3.1exa1} to calculate

\begin{align}\label{ass3.1exa2}
  \gamma_j=&-\sqrt{\frac{2}{\pi}}\DF{1}{\ell(j)}\int_0^{\pi} m^{\prime}({x})\sin(\ell(j){x})d{x}\notag\\
  =&-\sqrt{\frac{2}{\pi}}\DF{1}{\ell(j)^2}\int_0^{\pi} m^{\prime}({x})d\cos(\ell(j){x})\notag\\
  =&\sqrt{\frac{2}{\pi}}\DF{1}{\ell(j)^2}\int_0^{\pi} m^{\prime \prime}({x})\cos(\ell(j){x})d{x}\notag\\
  =&-\sqrt{\frac{2}{\pi}}\DF{1}{\ell(j)^3}\int_0^{\pi} m^{\prime \prime \prime}({x})\sin(\ell(j){x})d{x}\notag\\
  =&\cdots \notag\\
  =&\begin{cases}
    -\sqrt{\frac{2}{\pi}}\DF{1}{\ell(j)^{s+1}}\int_0^{\pi} m^{(s+1)}({x})\sin(\ell(j){x})d{x}, & \text{for}\, s \; \text{even}, \\
    \sqrt{\frac{2}{\pi}}\DF{1}{\ell(j)^{s+1}}\int_0^{\pi} m^{(s+1)}({x})\cos(\ell(j){x})d{x}, & \text{for}\, s \; \text{odd},
  \end{cases}  
\end{align}
where $m^{(s)}(\cdot)$ denotes the $s$--th derivative of $m(\cdot)$. Let $\ell(j)=aj+b$ for some integers $a>1$, $0\le b\le a-1$. Then,  $\e[\delta_k^2({x})]=O(k^{-2s-1})$.

Cases (1) and (2) are examples where $\{\psi_j({x})\}$ is a sparse subsequence of $\{\varphi_j({x})\}$, while Case (3) is less sparse that includes as special cases $\ell(j)=2j$, $\ell(j)=2j+1$ and $\ell(j)=5j+4$, to name a few. In all these cases, when $s$ or $k$ is suitable large the condition is fulfilled.

When $d>1$, the smoothness order, $s$, can be replaced by $s/d$ for $s>d$ in Case (3) where the boundary condition is adjusted accordingly.

We then come back to Case C of Example 6.1. Note from the derivations in Appendix A.5 below in the case of $d=1$ that 
\bea
\lambda_j & = & \e[x \, \psi_j(x)] =\left\{\begin{array}{ll}
        0 & \text{if $j$ is even}, \\
        \frac{c_1}{j^2} & \text{if $j$ is odd},
    \end{array}\right.
    \label{f4.1}\\
    \gamma_j & = & \e[m(x) \, \psi_j(x)] =\left\{\begin{array}{ll}
        0 & \text{if $j$ is even}, \\
        \frac{c_2}{j^2} & \text{if $j$ is odd},
    \end{array}\right.
    \label{f4.2}
    \eea
    where both $c_1$ and $c_2$ are constants.
    
    Equations (\ref{f4.1}) and (\ref{f4.2}) imply that for some constant $c_3$ and large enough $n$
    \be
    \sum_{j=k+1}^{\infty} {\lambda}_j \, \gamma_j = \frac{c_3}{k^3} = o(n^{-1/2})
    \label{f4.3}
    \ee
    which is satisfied when choosing $k = \left[n^{c_4}\right]$ for $c_4>\frac{1}{6}$. This, along with Assumption 3.1(iii), requires $\frac{1}{6}<c_4<\frac{1}{4}$.
    
    Similarly, other cases used in the simulation evaluations in Example 6.1 may be verified to ensure $\sum_{j=k+1}^{\infty} \bm{\lambda}_j \, \gamma_j = o(n^{-1/2})$.

\subsection{Verification of the simulation designs in Example 6.1}

We consider Case C of Example 6.1 only here. As the derivations for Cases A, B and D of Example 6.1 are similar and simpler, and for Example 6.2 are also similar, we omit the details. To show that $\e[m(x_i)]=0$ and $\e[x_im(x_i)]\ne 0$, it is sufficient to consider
\begin{eqnarray*}
    \int_{0}^\pi \frac{1}{\pi} \left(z-\frac{\pi }{2}\right)\mathrm{d}z =\frac{\pi }{2} -\frac{\pi }{2}=0 \ \ \mbox{and} \ \
    \frac{1}{\pi}\int_0^{\pi}  z^2\mathrm{d}z- \frac{1}{2}\int_0^{\pi}  z\mathrm{d}z=\frac{1}{3}\pi^2 -\frac{1}{4}\pi^2 =\frac{1}{12}\pi^2,
\end{eqnarray*}
so endogeneity exists. We then show that 
\begin{eqnarray*}
    E[x_i^2 \cos(j x_i)] &=&\frac{1}{\pi}\int_{0}^\pi z^2 \cos (j z)\mathrm{d}z=\frac{1}{\pi j}\int_{0}^\pi  z^2\mathrm{d}\sin(jz) = -\frac{2}{\pi j}\int_{0}^\pi\sin(jz) z \mathrm{d}z \notag \\
    &=&\frac{2}{\pi j^2}\int_{0}^\pi z \mathrm{d}\cos(jz)=\frac{2}{\pi j^2} z\cos(jz) |_0^\pi -\frac{2}{\pi j^2}\int_{0}^\pi\cos(jz)\mathrm{d}z  \notag \\
    &=& \frac{2}{ j^2} \cos(j\pi) =\left\{\begin{array}{ll}
        \frac{2}{ j^2} & \text{if $j$ is even} \\
        -\frac{2}{ j^2} & \text{if $j$ is odd}
    \end{array} \right.
\end{eqnarray*}
and 
\begin{eqnarray*}
E[x_i\cos(j x_i)] & = & \frac{1}{\pi}\int_0^\pi z\cos(jz)\mathrm{d}z -\frac{1}{2}\int_0^\pi \cos(jz)\mathrm{d}z = \frac{1}{\pi j}\int_0^{\pi}z\mathrm{d} \sin(jz) =\notag \\
&= & -\frac{1}{\pi j}\int_0^{\pi}\sin(jz)\mathrm{d}z=\frac{1}{\pi j^2}\cos(jz)|_0^\pi =\left\{\begin{array}{ll}
        0 & \text{if $j$ is even}, \\
        -\frac{2}{\pi j^2} & \text{if $j$ is odd}.
    \end{array}\right.
\end{eqnarray*}

Therefore, we have shown that $m(x)$ can be fully expanded by the odd terms of $\{\cos(j \, x)\}$, while fully expanding the regressor $x^2$ requires all of $\{\cos(j \, x)\}$.

\renewcommand{\theequation}{B.\arabic{equation}}

\setcounter{subsection}{2}

\section{Additional Discussions and Results}\label{Ap.B}

\subsection{Implementational Issues}

We offer our recommendations on the choice of orthonormal series functions involved in the proposed SP estimation before we introduce a generalized cross--validation method for the selection of $k$, the number of truncation terms involved in the approximation of $m(\mathbf{x})$. 

\subsubsection{Choice of series functions}

The main model and estimation method proposed in Section 3 remains valid for the case where the dimensionality, $d$, can be large but fixed, although we assume in Section 2 that we focus on the case where $d$ is known and small. 

Let us discuss the expansion of $m(\mathbf{x})$ used in the main paper for the case where $\mathbf{x}$ is multivariate. It follows from using the so--called ``unravelling procedure" discussed in \cite{ts2023} for example that we can expand $m({\bf x})$ by a general series expression of the form:
\be
m(\mathbf{x}) = \sum_{j=1}^{\infty} \psi_j(\mathbf{x}) \, \gamma_j \equiv: \bm{\Psi}_{\infty}^{\top}(\mathbf{x}) \, \bm{\gamma}_{\infty},
\label{appb1}
\ee
where $\bm{\Psi}_{\infty}(\mathbf{x}) \equiv \bm{\Psi}(\mathbf{x}) = \left(\psi_1(\mathbf{x}), \cdots, \psi_k(\mathbf{x}), \cdots, \psi_{\infty}(\mathbf{x})\right)^{\top}$ and $\bm{\gamma}_{\infty} = \left(\gamma_1, \cdots, \gamma_k, \cdots, \gamma_{\infty}\right)^{\top}$.
\smallskip

\noindent{\bf 1}. \, In general, one can construct $\{\psi_j(\mathbf{x}): j\geq 1\}$ for the case of $d\geq 2$. For ease of implementation, we propose using $\psi_j(\mathbf{x}) = \prod_{k=1}^d \psi_{jk}(x_k)$ when $d\geq 2$, where $\{\psi_{jk}(\cdot): k\geq 1\}$ is an array of univariate series functions chosen from either univariate probabilistic Hermite or trigonometric polynomials, or a mixture of both, as discussed in \cite{gallant1981} and \cite{andrews1991}, for example. 
\smallskip

\noindent{\bf 2}. \, Our experience with the finite--sample studies in Sections 5 and 6 of the main paper as well as in Appendix B.2 below shows that the multiplicative form: $\psi_j(\mathbf{x}) = \prod_{k=1}^d \psi_{jk}(x_k)$ works well numerically and better than some other competing forms available in the relevant literature.
\smallskip

Meanwhile, standard dimension reduction methods may also be employed, such as an additivity structure of the form: $m(\mathbf{x}) = \sum_{j=1}^d m_j(x_j)$ as in \cite{ln1995}, \cite{gao2002}, \cite{gao2006}, and  \cite{dlp2021} for example, in which it is expected that the construction of $\mathcal{S}$ should be a direct sum of $\mathcal{S}_j$ for $1\leq j\leq d$.  There are also dimension reduction methods proposed by involving certain types of single--index or multi--index modelling methods (see \cite{dgl2023}, and \cite{dg2025}, for example).
     
As discussed in Section 4.2.3 and Appendix C.4 of \cite{dg2025}, we need not know the distributional structure of the data under analysis in practice as long as the support of $\{\mathbf{x}_i: i\geq 1\}$ becomes available to the practitioner. Along with the fact that the identification of $\bm{\beta}_0$ is invariant to the choice of $\{\psi_j(\cdot): j\geq 1\}$ as demonstrated in Section 2.2, we propose using Steps {\bf 1} and {\bf 2}.

In theory, the choice of $\{\psi_j(\cdot): j\geq 1\}$ can be flexible as long as $\bm{\Sigma}_x>{\bf 0}$. As discussed in Section 4 of the main submission, moreover, the proposed LASSO selection is applicable to the case where $\{\psi_j(\cdot): j\geq 1\}$ are not necessarily orthogonal.

\subsubsection{Generalized cross--validation}

In the rest of Appendix B, without imposing orthonormality and employing the LASSO selection method, we present some numerical evidence with the corresponding theoretical justification to support our SP method. We first introduce an alternative method for selecting the truncation parameter when sparsity assumptions are not ideal for practical analysis. We emphasize that the primary objective is not the specific selection of $k$, but rather the execution of a rigorous robustness check for our SP method.

Specifically, we employ the so--called ``Generalized Cross--Validation" (GCV) method (see, Chapter 2 of \cite{hlg2000}, for example) as follows:
\be
\widehat{k}_n =\argmin_{k\in K_n} \left(1 - \frac{k+1}{n}\right)^{-2} \cdot  \frac{1}{n} \sum_{i=1}^n \left(y_i - \mathbf{x}_i^{\top} \widehat{\bm{\beta}} - \widehat{m}_{\phi}(\mathbf{x}_i)\right)^2,
\label{4.18}
\ee
where $\widehat{m}_{\phi}(\mathbf{x}_i) = \sum_{j=1}^{k} \phi_j({\bf x}_i) \, \widehat{\gamma}_{j, \phi}$ with $\phi_j(\cdot)$ being an orthonormal sequence of basis functions chosen by the user from $L^2([0, \pi])$ and $\widehat{\gamma}_{j, \phi}$ being defined in the same way as in Section \ref{Sec3}, in which $K_n = \left\{1,2, \cdots, \left[c_1 \, n^{\frac{1}{5}}\right]\right\}$, in which $c_1$ is the user choice in an individual scenario such that $K_n$ can have up to certain integers. An asymptotic consistency for $\widehat{k}_n$ can be established in a similar way to \cite{gao2002}.

\subsection{Additional simulated examples}

\noindent{\bf Example B1}: Consider the following data generating process for a univariate setting:
\be
y_i = x_i \, \beta_0 + \varepsilon_i \quad \mbox{with} \quad \varepsilon_i = m(x_i) + e_i,\notag
\ee
where $e_i\sim N(0,1)$, $\beta_0=1$, and $i =1,\ldots , n$. Additionally, we consider the following cases for $x_i$'s and $m(x)$:

\begin{enumerate}
    \item[] Case A: $x_i \sim U(-2, 1)$, and $m(x) = 15 \cos(3\pi x)$;
    \item[] Case B: $x_i \sim N(1, 1)$, and $m(x) = x^2-2$;
    \item[] Case C: $x_i \sim U(-1, 2)$, and $m(x) = \exp(x) - x - \frac{\exp(2)-\exp(-1)}{3} +\frac{1}{2}$;
    \item[] Case D: Consider Case A, and add an extra exogenous variable to the model, so $y_i = x_i \, \beta_0 +w_i \, \alpha_0+ \varepsilon_i$, where $w_i\sim N(1,1)$ and $\alpha_0=1$;
\end{enumerate}

In Appendix B.3 below, we show that Cases A--D all fulfil the following endogeneity condition:
\begin{equation}\label{condition.1}
    \mathbb{E}[m(x_i)]=0\quad\text{and}\quad \mathbb{E}[x_i \, m(x_i)]\neq 0.
\end{equation}

In what follows, we provide some detailed implementation and make comments on each case. We start with Cases A--C. 

\begin{enumerate}
\item {\bf LS}--Estimate $\beta_0$ by the OLS method:
\begin{eqnarray*}
    \widehat{\beta}_{\rm LS}= (\mathbf{X}^{\top}\mathbf{X} )^{-1} \mathbf{X}^{\top}  \mathbf{y} ,
\end{eqnarray*}
where we let $\mathbf{X}= (x_1, \ldots, x_n )^{\top}$ and $\mathbf{y} = (y_1, \ldots, y_n )^{\top}$ for notational simplicity.
\item {\bf SP}--Estimate $\beta_0$ by the SP method of the form:
\begin{equation*}
    \widehat{\beta}_{\rm SP} = (\mathbf{X}^{\top} \mathbf{M}_v \mathbf{X} )^{-1} (\mathbf{X}^{\top} \mathbf{M}_v \mathbf{y}),
\end{equation*}
where $\mathbf{M}_v= \mathbf{I}_n -  \mathbf{V} (\mathbf{V}^{\top} \mathbf{V} )^{-1} \mathbf{V}^{\top}$, and $\mathbf{V}= (\mathbf{v}(x_1), \ldots, \mathbf{v}(x_n) )^{\top}$. Here, we take the suggestions from \cite[p. 228]{gallant1981}, and choose $\mathbf{v}(x)$ as $\mathbf{v}(x)=(1, H_2(x),  p_1(x), p_2(x),\ldots, p_{k-2}(x))^\top$. An optimal choice of $k$ as $\widehat{k}_n$ is selected by the GCV method defined in \eqref{4.18}.
\end{enumerate}

In practice, one can further add more polynomials to $\mathbf{v}(x)$ if necessary. The above setting is already sufficient to evaluate our theory, so we no longer explore higher order polynomials. We repeat the above procedure $R=1000$ times, and report the following measures
\begin{eqnarray}\label{measure_1}
\Delta \beta &=& \overline{\beta}  - \beta_0,\quad  \quad \text{sd}_\beta = \left\{\frac{1}{R}\sum_{r=1}^R(\widehat{\beta}_r - \overline{\beta})^2 \right\}^{1/2},\quad \overline{k} = \frac{1}{R}\sum_{r=1}^R \widehat{k}_{n,r} ,
\end{eqnarray}
where $\overline{\beta} = \frac{1}{R}\sum_{r=1}^R\widehat{\beta}_r$ with $\widehat{\beta}_r \in\{\widehat{\beta}_{\rm LS},\widehat{\beta}_{\rm SP}\}$ in each replication, $\widehat{k}_{n,r}$ stands for the value of $\widehat{k}_n$ chosen by the GCV in the $r^{th}$ replication, and $\overline{k}$ reported in the following tables represents the largest integer part.

For Case D, we also consider LS and SP, but choose

\begin{equation*}
    \mathbf{X} =\begin{pmatrix}
        x_1 & \ldots & x_n \\
        w_1 & \ldots & w_n
    \end{pmatrix}^\top .
\end{equation*}

The rest setting is identical to those specified for Cases A--C. On top of the above measures in \eqref{measure_1}, we report $\Delta \alpha$ and $\text{sd}_\alpha$ which are calculated in the same way as $\Delta \beta$ and $\text{sd}_\beta$.

{\small
\begin{table}[htb!]
\caption{Results of Cases A-D of Example B1}\label{Tb41}
\centering \small
\setlength{\tabcolsep}{4pt}
\renewcommand{\arraystretch}{0.7} 
\begin{tabular}{lrrrrlrrrrr}
\hline\hline
 &  & \multicolumn{3}{c}{Case A} &   & \multicolumn{3}{c}{Case C} & &  \\
 & $n$ & $\Delta \beta$ & $\text{sd}_\beta$ & $\overline{k}$ &  & $\Delta \beta$ & $\text{sd}_\beta$ & $\overline{k}$ &  &  \\
LS & 200 & -0.091 & 0.757 &  &   & 0.782 & 0.101  \\
 & 300 & -0.126 & 0.599 &  &   & 0.784 & 0.080  \\
 & 400 & -0.107 & 0.533 &  &   & 0.783 & 0.067 \\
SP & 200 & 0.004 & 0.137 & 5 &   & 0.067 & 0.131 & 4 &  &  \\
 & 300 & -0.007 & 0.107 & 5 &   & 0.063 & 0.109 & 4  &  &  \\
 & 400 & 0.006 & 0.098 & 5 &   & 0.070 & 0.096 & 4  &  &  \\
\hline \hline
 &  & \multicolumn{3}{c}{Case B} & &   \multicolumn{5}{c}{Case D}  \\
 & $n$ & $\Delta \beta$ & $\text{sd}_\beta$ & $\overline{k}$ &    & $\Delta \beta$ & $\text{sd}_{\beta}$ & $\Delta \alpha$ & $\text{sd}_{\alpha}$ & $\overline{k}$ \\
LS & 200 & 0.981 & 0.160 &   &  & -0.187 & 0.838 & -0.062 & 0.584 &  \\
 & 300 & 0.992 & 0.125 &  &   & -0.163 & 0.664 & -0.039 & 0.473 &  \\
 & 400 & 0.992 & 0.110 &  &   & -0.154 & 0.562 & -0.040 & 0.402 &  \\
SP & 200 & -0.003 & 0.129 & 2 &  &  -0.003 & 0.135 & -0.001 & 0.072 & 5 \\
 & 300 & -0.005 & 0.104 & 2 & & 0.004 & 0.110 & 0.003 & 0.058 & 5 \\
 & 400 & 0.002 & 0.086 & 2 &  & -0.003 & 0.094 & 0.001 & 0.051 & 5 \\
 \hline
 \hline
\end{tabular}
\end{table}

}

We summarise the results in Table \ref{Tb41}, and provide the key findings below.

\begin{enumerate}
    \item Across all scenarios, the SP method consistently delivers strong performance. In contrast, the LS method consistently demonstrates noticeable bias and, in some instances, exhibits considerable standard deviation.

    \item In Case C, simple algebra shows that $\exp(x)-x =\sum_{j=0}^\infty \frac{x^j}{j!} - x$, so the linear term disappear from the right hand side. While truncation residuals are inherent in this scenario, the proposed SP method still achieves significantly lower bias and smaller standard deviation.

    \item For Case D, the proposed SP method surpasses the LS method by yielding unbiased estimates and reduced standard deviations for both parameters.

    \item The optimal value of $k$ is case-dependent. Specifically, for optimal finite sample performance, we anticipate (i) $\widehat{k}_n\ge 5$ in Cases A and D, and (ii) $\widehat{k}_n\ge 2$ for Cases B and C. These expectations align with the results presented in Table \ref{Tb41}. Furthermore, the value of $k$ suggests that a slightly over-specified $k$ still yields reasonable finite sample performance. This is a positive finding, as it indicates that the choice of $k$ is quite robust in practice.

\end{enumerate}

We now have a look at a bivariate case for $\mathbf{x} = (x_1, x_2)^{\top}$ below.
\medskip

\noindent{\bf Example B2}: Consider the following DGP for a bivariate setting:
\begin{equation}
y_i =  \mathbf{x}_i^{\top} \bm{\beta}_0 + \varepsilon_i \quad\text{and}\quad \varepsilon_i = m(\mathbf{x}_i) + e_i,
\label{5.6}
\end{equation}
where $\mathbf{x}_i =(x_{1i},x_{2i})^\top$, $\bm{\beta}_0=\left(\beta_{01}, \beta_{02}\right)^{\top} = \left(1, 1\right)^{\top}$, and $e_i \sim N(0, 1)$ for $1\leq i\leq n$.

As discussed in the end of Appendix B.3 below, the following conditions can be verified:
\begin{equation}\label{cond.bi}
    \mathbb{E}[m(\mathbf{x}_i)]=0\quad\text{and}\quad \mathbb{E}[\mathbf{x}_i \, m(\mathbf{x}_i)]\neq \mathbf{0}.
\end{equation}

We start with two additive cases:

\begin{itemize}
    \item[] Case A: $m(\mathbf{x}) = x_1^2 +x_2^2-13/3$,  $x_{1i}\sim U(-2,3)$, and $x_{2i}\sim N(1,1)$;

    \item[] Case B: $m(\mathbf{x}) = 15\cos(3\pi x_1)+x_2^2-2$, $x_{1i}\sim U(-2,1)$, and $x_{2i}\sim N(1,1)$.
\end{itemize}

Accordingly, we let
\begin{equation*}
    \mathbf{v}(\mathbf{x}) = (1,H_2(x_1),  p_1(x_1), \ldots, p_{k-2}(x_1), H_2(x_2),  p_1(x_2),\ldots, p_{k-2}(x_2) )^\top,
\end{equation*}
where $k\ge 2$ ensures the length of $\mathbf{v}(\mathbf{x})$ is at least 3. Apparently, in this case, $\mathbf{v}(\mathbf{x})$ is $(2k-1)\times 1$. An optimal choice of $k$ as $\widehat{k}_n$ is selected by the GCV method:
\begin{equation}
\widehat{k}_n = \argmin_{ k\in K_n} \left(1 - \frac{2k}{n}\right)^{-2} \cdot  \frac{1}{n} (\mathbf{y} - \mathbf{X} \, \widehat{\beta}_{\rm SP} - \mathbf{V} \, \widehat{\pmb{\gamma}} )^{\top}  (\mathbf{y} - \mathbf{X} \, \widehat{\beta}_{\rm SP} - \mathbf{V} \, \widehat{\pmb{\gamma}}),\notag
\end{equation}
where $\widehat{\gamma}  = (\mathbf{V}^{\top} \mathbf{V} )^{-1} \mathbf{V}^{\top}  (\mathbf{y} - \mathbf{X} \, \widehat{\beta}_{\rm SP})$, and $K_n$ is as defined in (\ref{4.18}).

Meanwhile, we also consider the following multiplicative cases:
\begin{itemize}
    \item[] Case C: $m(\mathbf{x}) = x_1^2\cdot (x_2^2-2)$,  $x_{1i}\sim U(-2,3)$, and $x_{2i}\sim N(1,1)$;

    \item[] Case D: $m(\mathbf{x}) = 15\cos(3\pi x_1)\cdot (x_2^2+2)$, $x_{1i}\sim U(-2,1)$, and $x_{2i}\sim N(1,1)$.
\end{itemize}

Accordingly, we let $\mathbf{v}(\mathbf{x})$ include the distinct elements from
\begin{equation*}
    \{1,H_2(x_1),  p_1(x_1), \ldots, p_{k-2}(x_1)\}\otimes \{ 1,H_2(x_2), p_1(x_2),\ldots, p_{k-2}(x_2)\},
\end{equation*}
where $k\ge 2$ ensures the length of $\mathbf{v}(\mathbf{x})$ is at least 4.  An optimal choice of $k$ as $\widehat{k}_n$ is selected by the GCV method:
\begin{equation}
\widehat{k}_n = \arg\min_{ k\in K_n} \left(1 - \frac{k^2 +1}{n}\right)^{-2} \cdot  \frac{1}{n} (\mathbf{y} - \mathbf{X} \, \widehat{\beta}_{\rm SP} - \mathbf{V} \, \widehat{\pmb{\gamma}} )^{\top}  (\mathbf{y} - \mathbf{X} \, \widehat{\beta}_{\rm SP} - \mathbf{V} \, \widehat{\pmb{\gamma}}),
\end{equation}
where the quantities involved are the same as those in (\ref{4.18}).

{\small
\begin{table}[htb!] 
\caption{Results of Bivariate Case of Example B2}\label{Tb42}
\centering \small
\setlength{\tabcolsep}{4pt}
\renewcommand{\arraystretch}{0.7} 
\begin{tabular}{lrrrrrrlrrrrr}
\hline\hline
 &  & \multicolumn{5}{c}{Case A} &  & \multicolumn{5}{c}{Case B} \\
 & $n$ & $\Delta \beta_1$ & $\text{sd}_{\beta_1}$ & $\Delta \beta_2$ & $\text{sd}_{\beta_2}$ & $\overline{k}$ &  & $\Delta \beta_1$ & $\text{sd}_{\beta_1}$ & $\Delta \beta_2$ & $\text{sd}_{\beta_2}$ & $\overline{k}$\\
LS & 200 & 0.712 & 0.128 & 0.809 & 0.192 &  &  & 0.409 & 0.792 & 1.099 & 0.622 &  \\
 & 300 & 0.714 & 0.110 & 0.811 & 0.164 &  &  & 0.405 & 0.667 & 1.109 & 0.484 &  \\
 & 400 & 0.715 & 0.092 & 0.809 & 0.142 &  &  & 0.422 & 0.593 & 1.119 & 0.421 &  \\
SP & 200 & 0.000 & 0.062 & -0.002 & 0.128 & 3 &  & 0.003 & 0.130 & 0.002 & 0.130 & 6 \\
 & 300 & -0.001 & 0.052 & -0.004 & 0.099 & 3 &  & 0.000 & 0.110 & 0.007 & 0.101 & 6 \\
 & 400 & 0.000 & 0.043 & 0.000 & 0.091 & 3 &  & 0.003 & 0.101 & -0.001 & 0.088 & 6 \\
 &  & \multicolumn{5}{c}{Case C} &  & \multicolumn{5}{c}{Case D} \\
LS & 200 & -0.554 & 0.472 & 2.414 & 0.607 &  &  & -0.467 & 3.953 & -0.158 & 4.543 &  \\
 & 300 & -0.522 & 0.376 & 2.452 & 0.477 &  &  & -0.402 & 3.021 & -0.017 & 3.701 &  \\
 & 400 & -0.517 & 0.332 & 2.458 & 0.444 &  &  & -0.614 & 2.567 & -0.117 & 3.023 &  \\
SP & 200 & -0.001 & 0.062 & -0.006 & 0.134 & 3 &  & -0.005 & 0.151 & -0.011 & 0.136 & 6 \\
 & 300 & 0.002 & 0.054 & -0.001 & 0.102 & 3 &  & 0.000 & 0.117 & -0.006 & 0.104 & 6 \\
 & 400 & 0.001 & 0.043 & 0.000 & 0.087 & 3 &  & 0.002 & 0.098 & -0.002 & 0.091 & 6 \\
 \hline\hline
\end{tabular}
\end{table}}

For the bivariate settings, the findings in Table \ref{Tb42} support that the proposed SP method works well numerically in a very similar spirit to that of Table \ref{Tb41} for Example B1.
\medskip

We then add some simulation results for the case where $m(\mathbf{x})$ itself contains a linear component in each case under examination in Example B3 below.  Before we provide you with the simulation results, we would like to show you why the SP based estimation method still works in some linearity cases.

Consider the following univariate linear model:
\be
y_i = x_i \, \beta_0 + (x_i - \e[x_i]) \, \gamma_0 + e_i = \e[x_i] \, \beta_0 + (x_i - \e[x_i]) \, (\beta_0 + \gamma_0) + e_i
\label{extra1}
\ee
where $\e[e_i|x_i]=0$, and both $\beta_0$ and $\gamma_0$ are identifiable and estimable in the case of $\e[x_i]\neq 0$. In other words, the case of $\e[x_i]\neq 0$ can help identify $\beta_0$. 

In the beginning of Section 3.3 of the main paper, we already show that when $\e[x_i]=0$, $\beta_0$ cannot be uniquely identifiable, although $(\beta_0 + \gamma_0)$ can be collectively identifiable. 
\medskip

\noindent{\bf Example B3}: We now consider model $y_i = x_i \, \beta_0 + \varepsilon_i$, $\varepsilon_i = m(x_i) + e_i$, $e_i\sim N(0,1)$ and $x_i \sim U[0, \pi]$ under the following scenarios:
\begin{itemize}
    \item Case E: $m(x) = 1.5 \, (\cos(x) + 2 \,x - \pi) $;
    \item Case F: $m(x) = 1 + (2 \, x - \pi) - \frac{3}{\pi^2} \, x^2$; and
    \item Case G: $m(x) = 1.5 \, (x - 0.5 \, \pi)$.
\end{itemize}

Apparently, the restrictions: $\e[m(x_i)]=0$ and $\e[x_i \, m(x_i)]\neq 0$ can be verified in a similar way to those listed in Appendix A.5 for Example 6.1. We summarize the simulation results in Table \ref{tb.linear} below, wherein the biases and standard deviations (in parentheses) are calculated in exactly the same way as before. 

\begin{table}[htb!]
    \caption{Estimation Results for Cases E--G}\label{tb.linear}
    \centering \small
    \setlength{\tabcolsep}{6pt}
\renewcommand{\arraystretch}{0.6}
\begin{tabular}{rrrrr}
\toprule
 & &  $n=100$ & $n=250$ & $n=500$  \\
 Case E & LS & 0.4569 (0.0051) & 0.4604 (0.0022) & 0.4600 (0.0110) \\
& SP & 0.0036 (0.0040) & 0.0013 (0.0015) & 0.0022 (0.0008) \\ 
Case F & LS & 0.2616 (0.0031) & 0.2621 (0.0012) & 0.2624 (0.0007) \\
& SP & 0.0033 (0.0041) & 0.0013 (0.0015) & 0.0016 (0.0008) \\
Case G & LS & 0.3733 (0.0042) & 0.3758 (0.0017) & 0.3756 (0.0009) \\
& SP & 0.0031 (0.0041) & 0.0011 (0.0015) & 0.0016 (0.0008)\\ 
\bottomrule
\end{tabular}
\end{table}

Table \ref{tb.linear} obviously shows that the SP method works well numerically. Moreover, the numerical results are better than those in Example 6.1, probably because the elementary functional forms of $m(\cdot)$ considered in Cases E--G were all accurately approximated by $\left\{\psi_j(x) = \sqrt{\frac{2}{\pi}} \, \cos(j\, x): j\geq 1\right\}$ for the case of $x\in U[0, \pi]$.

\subsection{Verifications of the above simulation designs}

In what follows, $x_i$ stands for a random variable, while $x$ stands for a real number.

\begin{lemma}\label{Lem_sim_1}
$\{\cos(i\pi x)\mid i \in 0\cup \mathbb{N}\}$ is a set of orthogonal functions in $L^2([a,b])$ for $\forall a,b\in \mathbb{Z}$ and $a\ne b$.
\end{lemma}

\begin{proof}[Proof of Lemma \ref{Lem_sim_1}]
\item

For $i\ne j$, write

\begin{eqnarray*}
&&    \int_a^b\cos(i\pi x) \cos(j\pi x)\mathrm{d}x= \frac{1}{2}\int_a^b \cos((i-j)\pi x)\mathrm{d}x+\frac{1}{2}\int_a^b \cos((i+j)\pi x)\mathrm{d}x \notag \\
    &&=\frac{1}{2(i-j)\pi}\sin((i-j)\pi x) \big|_a^b+\frac{1}{2(i+j)\pi}\sin((i+j)\pi x) \big|_a^b = 0,
\end{eqnarray*}
where the last step is obvious.

Also, we write

\begin{eqnarray*}
    \int_a^b\cos(i\pi x)^2\mathrm{d}x &=&\frac{1}{2}\int_a^b \cos(0\pi x)\mathrm{d}x+\frac{1}{2}\int_a^b \cos(2i\pi x)\mathrm{d}x \notag \\
    &=&\frac{b-a}{2}+\left\{\begin{array}{ll}
        \frac{b-a}{2} & \text{for }i=0 \\
        \frac{1}{4i\pi}\sin(2i\pi x) |_a^b & \text{for }i\ge 1
    \end{array} \right. = \left\{\begin{array}{ll}
        b-a & \text{for }i=0 \\
        \frac{b-a}{2} & \text{for }i\ge 1
    \end{array} \right. .
\end{eqnarray*}
The proof is completed.
\end{proof}

\begin{proof}[Justification of Example B1]
\item

We consider the cases one by one.

Case A: \, Let $\ell$ be an odd positive integer. We note that

\begin{eqnarray*}
    \mathbb{E}[\cos(\ell\pi x_i)] = \int_{-2}^1 \frac{1}{3} \cos(\ell\pi x)\mathrm{d}x =\frac{1}{9\pi}\sin(\ell\pi x)\big|_{-2}^1=0.
\end{eqnarray*}

Also,
\begin{eqnarray*}
&&    \mathbb{E}[x_i\cos(\ell\pi x_i)] = \int_{-2}^1 \frac{1}{3} x \cdot \cos(\ell\pi x)\mathrm{d}x =\frac{1}{3\ell\pi}\int_{-2}^1 x\, \mathrm{d}\sin(\ell\pi x)\notag \\
&& = \frac{1}{3\ell\pi}x\cdot\sin(\ell\pi x) \big|_{-2}^1 -\frac{1}{3\ell\pi}\int_{-2}^1 \sin(\ell\pi x) \mathrm{d}x = \frac{1}{3(\ell\pi)^2}\cos(\ell\pi x)\big|_{-2}^1=\frac{-2}{3(\ell\pi)^2},
\end{eqnarray*}
and
\begin{eqnarray*}
    \mathbb{E}[x_i^2\cos(\ell\pi x_i)] &=&\int_{-2}^1 \frac{1}{3} x^2 \cdot \cos(\ell\pi x)\mathrm{d}x =\frac{1}{3\ell\pi}\int_{-2}^1 x^2\, \mathrm{d}\sin(\ell\pi x)\notag \\
    &=&-\frac{2}{3\ell\pi}\int_{-2}^1 x \cdot\sin(\ell\pi x)\mathrm{d}x= \frac{2}{3(\ell\pi)^2}\int_{-2}^1 x \,\mathrm{d}\cos(\ell\pi x) \notag \\
    &=&\frac{2}{3(\ell\pi)^2}x\cdot \cos(\ell\pi x)\big|_{-2}^1-\frac{2}{3(\ell\pi)^2}\int_{-2}^1\cos(\ell\pi x)\mathrm{d}x =\frac{2}{3(\ell\pi)^2}.
\end{eqnarray*}

Case B: \, We note that if $x_i\sim N(\mu, 1)$, then $\mathbb{E}[H_j(x_i)] =\mu^j$. The reason is as follows.

\begin{eqnarray*}
   && \mathbb{E}[H_j(x_i)] = \int_{\mathbb{R}}H_j(x )f(x-\mu)  \mathrm{d}x = \int_{\mathbb{R}}H_j(x+\mu )f(x)  \mathrm{d}x\notag \\
    &&= \int_{\mathbb{R}}\sum_{k=0}^j\binom{j}{k}\mu^{j-k} H_j(x) f(x)  H_0(x)\mathrm{d}x = \mu^{j} \int_{\mathbb{R}} H_0(x)^2 f(x) \mathrm{d}x=\mu^{j} ,
\end{eqnarray*}
where $f(x)$ is the density function of $N(0,1)$, $H_0(x)\equiv 1$, and the fourth equality follows from the orthogonality of $\{ H_j(x)\}$. We them immediately obtain that

\begin{eqnarray*}
    \mathbb{E}[x_i^2-2]= \mathbb{E}[H_2(x_i)- H_0(x_i)]=\mu^2-\mu^0,
\end{eqnarray*}
and $\mathbb{E}[x_i(x_i^2-2)] =\mathbb{E}[H_3(x_i)+H_1(x_i)]= \mu^3+\mu^1$. Thus, for $\mu=1$, $\mathbb{E}[x_i^2-2]=0$ and $\mathbb{E}[x_i(x_i^2-2)]=2$.

\medskip

Case C: \, We note that $\exp(x)-x =\sum_{j=0}^\infty \frac{x^j}{j!} - x$, so the liner term disappear from the right hand side. Simple algebra shows that
\begin{eqnarray*}
    \int_{-1}^2 \frac{1}{3}(\exp(x)-x)\mathrm{d}x = \frac{1}{3}\exp(x)\big|_{-1}^2 -\frac{1}{6}x^2\big|_{-1}^2 = \frac{1}{3}\exp(2)-\frac{1}{3}\exp(-1)-\frac{1}{2}.
\end{eqnarray*}
Thus, $\mathbb{E}[m(x)]=0$.

We also have
\begin{eqnarray*}
&&    \int_{-1}^2\frac{1}{3} x\exp(x)\mathrm{d}x = \frac{1}{3} x\exp(x)\big|_{-1}^2-\int_{-1}^2 \frac{1}{3}\exp(x)\mathrm{d}x\notag \\
&& =\frac{2}{3}\exp(2) +\frac{1}{3}\exp(-1) -\frac{1}{3}\exp(2)+\frac{1}{3}\exp(-1) = \frac{1}{3}\exp(2) +\frac{2}{3}\exp(-1)
\end{eqnarray*}
and $\int_{-1}^2\frac{1}{3} x^2 \mathrm{d}x=\frac{1}{9}x^3\big|_{-1}^2=1$. Therefore, we have
\begin{eqnarray*}
    \mathbb{E}[x \, m(x)] = \frac{1}{3}\exp(2) +\frac{2}{3}\exp(-1) -1 -\frac{1}{3}\exp(2) +\frac{1}{3}\exp(-1)+\frac{1}{2} = \exp(-1)-\frac{1}{2}\ne 0.
\end{eqnarray*}

\medskip

\textbf{Remark}: We note that, for Cases A, C and D, it is easy to see that $\mathbb{E}[\mathbf{v}(x_j) \mathbf{v}(x_j)^\top]$ is not diagonal but has full rank based on the above development. For Case B, it is mathematically more involved, as we need to use Euler's formula: $\cos(x)=\frac{1}{2}(\exp(\mathrm{i}x)+\exp(-\mathrm{i}x))$, where $\mathrm{i}$ is an imagine unit. Thus,

\begin{eqnarray*}
    \mathbb{E}[\cos(x_j)]&=&\frac{1}{2}(\mathbb{E}[\exp(\mathrm{i}x_j)]+\mathbb{E}[\exp(-\mathrm{i}x_j)])\notag \\
    &=&\frac{1}{2}\exp\left(\mathrm{i}\mu_x-\frac{\sigma_x^2}{2}\right)+\frac{1}{2}\exp\left(-\mathrm{i}\mu_x-\frac{\sigma_x^2}{2}\right) = \exp\left( -\frac{\sigma_x^2}{2}\right)\cos(\mu_x),
\end{eqnarray*}
where $\mu_x$ and $\sigma_x$ are the mean and standard deviation of the random variable $x_j$, and the second equality follows from a direct calculation using the characteristic function of a normal distribution. As this is simulation exercise, we no longer purse the final form of $\mathbb{E}[\mathbf{v}(x_j) \mathbf{v}(x_j)^\top]$ further, which also varies with respect to the choice of the Hermite function(s) involved.
\end{proof}

\begin{proof}[Justification of Example B2]
\item

We now show that all cases fulfil \eqref{cond.bi}.

In Example B2, Cases A and B can be easily justified by following the justification of Cases A and B of Example E1, as $m(\mathbf{x})$ admits an additive form: $m(\mathbf{x}) =m_1(x_1)+ m_2(x_2)$ in both cases.

In Cases C and D, $m(\mathbf{x})$ always admits an multiplicative form: $m(\mathbf{x}) =m_1(x_1)\cdot m_2(x_2)$. In connection with the justification of Cases A and B of Example B1, we can also verify \eqref{cond.bi}.
\end{proof}

\subsection{More empirical analysis}\label{Ap.B2}

\subsubsection{Possible sources of endogeneity}

As briefly pointed out in the main sections of this paper, an additional novelty of the SP approach is that we are able to estimate $\varepsilon$ of model: $y=\mathbf{x}^{\top} \bm{\beta}_0 + \varepsilon$ by $\widehat{\varepsilon} = y - \mathbf{x}^{\top} \widehat{\boldsymbol{\beta}}$. As a consequence, we may be able to make the best use of $\widehat{\varepsilon}$ to establish an approximate regression model of the form: 
\be
u_i \equiv: \widehat{\varepsilon}_i = G(\mathbf{w}_i; {\eta}_i), \ 1\leq i\leq n
\label{nons1}
\ee
for the econometrician to reveal possible sources of endogeneity, where each $\mathbf{w}_i$ is a vector of variables possibly containing omitted variables and/or errors in variables, which may be highly correlated with $\mathbf{x}$, and ${\eta}$ is the error term involved.

Due to the fact that $\varepsilon$ is usually unobservable and latent, it is reasonable to assume the non--separable functional form in (\ref{nons1}). If each $\mathbf{w}_i$ is available to the econometrician, the discussion in  Section 3  suggests that one may be able to reveal possible sources of endogeneity by estimating the functional form of $G(\mathbf{w}; \cdot)$ with respect to $\mathbf{w}$.

We then estimate the following three nonparametric regressions models:
\be
    \text{Model 1}: \, \widetilde{x} = g_1(z_1) +\xi_1,\ \  \text{Model 2}: \, \varepsilon = g_2(z_1) +\xi_2,\ \ \text{Model 3}: \, z_1 = g_3(\widetilde{x}) +\xi_3,
    \nonumber
    \ee
where $\widetilde{x}$, $z_1$ and $\varepsilon$ have been defined previously, $g_1(\cdot)$, $g_2(\cdot)$ and $g_3(\cdot)$ measure possible relationships among them, and $\xi_1$, $\xi_2$ and $\xi_3$ are the error terms. 

The nonparametrically estimated relationships, along with their corresponding $95\%$ confidence intervals, are plotted in Figure \ref{Fig2}.  The second plot clearly shows that the relationship between $\varepsilon$ and $z_1$ is clearly nonlinear, although it is symmetrically fluctuating around zero. For simplicity, we use the generalized cross--validation (GCV) method of \cite{gao2002} to select the truncation parameter in each setting. In the second case, the GCV suggests the number of truncation parameters being $20$. The second plot infers that the typical IV method might be modelled beyond linear regression.

{\samepage
\begin{figure}[htb!]
    \centering
    \caption{Robustness Check via Nonparametric Estimates}\label{Fig2}
    \includegraphics[scale=0.6]{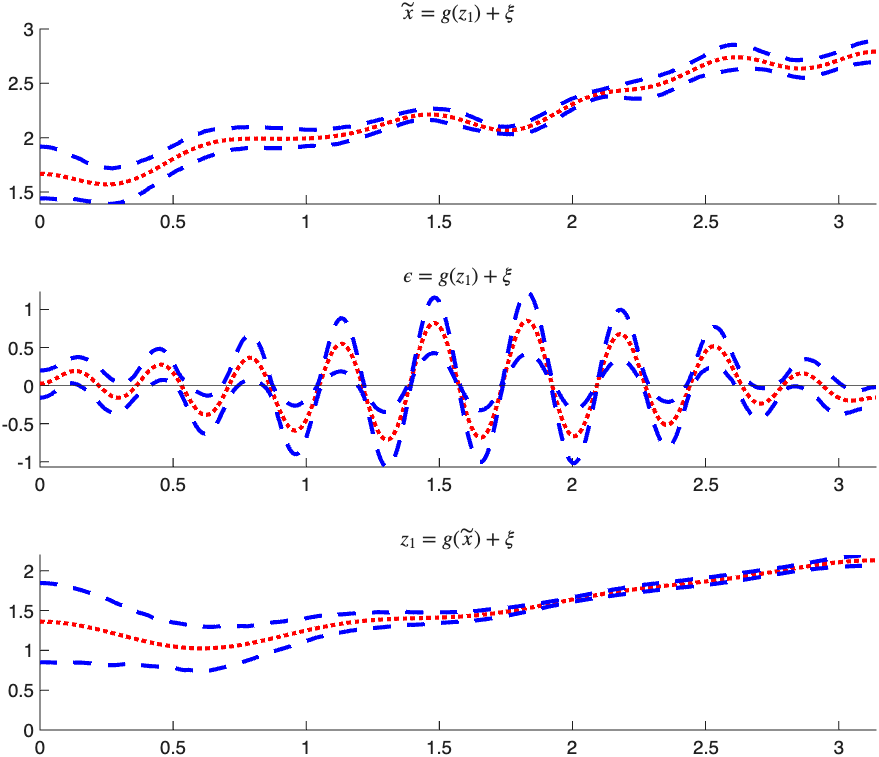}
\end{figure}
} 

The first and third plots together infer the relationship between $x$ and $z_1$ has a linear trend but not exactly linear, as the truncation parameters for these two regressions are 13 and 7 respectively according to the GCV approach. This further explains why the results associated the SP method and the IV method (using $z_1$) are close to each other in Table 2. 

In summary, due to the unobservable nature of $\varepsilon$, it is difficult to establish and reveal a relationship between $\varepsilon$ and potential instrumental variables.  Without assuming the knowledge of such relationships, the existence and the validity of potential IVs, the proposed SP approach offers a simple and robust way to deal with certain type of endogeneity involved in the return to schooling case study.

\subsubsection{Market return analysis} 

We apply the SP method to revisit the classic market model (\citealp{M1997, campbell1998econometrics}). See Eq. (4.3.2) of \cite{campbell1998econometrics} for example. Consider the following linear model:
\be
    y_t = \alpha_0+ x_t \, \beta_0 + \varepsilon_t, \ 1\leq t\leq n,
\label{m5.1}
\ee
where $\varepsilon_t=m(x_t) + e_t$, $y_t$ represents daily stock return for a chosen stock, and $x_t$ stands for the market portfolio. In this study, we  use daily return of the SP500 index as $x_t$. We aim to argue that endogeneity should also be taken into account for the market model. 
 
For the response $y_t$, we collect data for Apple, Google, Meta and Microsoft from Yahoo finance covering period from the first trading day of 2012 to the last trading day of 2020, which gives $n=2169$ observations for us to conduct regressions.  As in our study about the return to schooling, we normalize $x_t$ to get $\widetilde{x}_t$, so the latter belongs to $[0,\pi]$. Thus, the model used for regression is again written as 
\be
    y_t = \widetilde{\alpha}_0+\widetilde{\beta}_0\widetilde{x}_t + \varepsilon_t.
    \label{m5.2}
    \ee

For each stock, we consider (i) The SP method; \ (ii) The LS method; and (iii) The 2sLS method using the basis functions selected by the SP method as IVs.

The estimated versions of $m(x)$ for different companies are given as follows:
\begin{enumerate}
    \item[] Apple: $\widehat{m}(x) = 0.00348\phi_2(x)$;
    \item[] Google: $\widehat{m}(x) = 0.03176\phi_1(x) +0.00750\phi_3(x)-0.00187\phi_4(x)$;
    \item[] Meta: $\widehat{m}(x)= 0.02730\phi_1(x) -0.00647\phi_5(x) $;
    \item[] Microsoft: $\widehat{m}(x) = 0.00352\phi_2(x) -0.00123\phi_4(x)$.
\end{enumerate}

\begin{table}[htb!]
\centering \small
\caption{Estimation Results}\label{newtb.1}
\setlength{\tabcolsep}{4pt}
\renewcommand{\arraystretch}{0.5}
\begin{tabular}{llrrrrrr}
     \toprule
 &  & $\widetilde{\beta}_0$ & \multicolumn{1}{c}{95\% CI}   & RMSE  & $\beta_0$ & \multicolumn{1}{c}{95\% CI} & $\phi_j$'s selected \\
 &  &  &  & {\footnotesize ($\times 100$)} & & &\\
Apple & SP & 0.0756 & (0.0714, 0.0797) & 1.3425 & 1.1111 & (1.0497, 1.1724) & $j=2$ \\
& LS & 0.0769 & (0.0733, 0.0805) & 1.3430 & 1.1307& (1.0775, 1.1839) &  \\ 
 & 2sLS & 0.0809 & (0.0737, 0.0882) & 1.3445 & 1.1898 & (1.0833, 1.2964) &  \\
Google & SP & 0.0835 & (0.0222, 0.1448) & 1.1335 & 1.2276 & (0.3262, 2.1290) & $j\in \{1,3,4\}$ \\
& LS & 0.0717 & (0.0686, 0.0748) & 1.1373 & 1.0542 & (1.0092, 1.0993) &  \\
 & 2sLS & 0.0717 & (0.0686, 0.0747) & 1.1373 & 1.0538 & (1.0538, 1.0087) &  \\
Meta & SP & 0.0840 & (0.0378, 0.1303) & 2.0644 & 1.2357 & (0.5562, 1.9151) & $j\in \{1,5\}$ \\
& LS & 0.0723 & (0.0667, 0.0778) & 2.0691 & 1.0624 & (0.9804, 1.1443) &  \\  
& 2sLS & 0.0721 & (0.0665, 0.0777) & 2.0691 & 1.0598 & (0.9773, 1.1424) &  \\
Microsoft & SP & 0.0773 & (0.0736, 0.0810) & 1.0725 & 1.1364 & (1.0823, 1.1906) & $j\in \{2,4\}$\\
& LS & 0.0799 & (0.0770, 0.0828) & 1.0740 & 1.1752 & (1.1326, 1.2177) &  \\ 
& 2sLS & 0.0841 & (0.0795, 0.0888) & 1.0760 & 1.2371 & (1.1684, 1.3058) & \\ 
 \bottomrule
\end{tabular}
\end{table}

\begin{table}[htb!] 
\centering\small
\caption{Estimated Endogeneity ($\times 100$)}\label{newtb.2}
\renewcommand{\arraystretch}{0.6}
\begin{tabular}{lrr}
    \toprule
 & $\e[\widetilde{x} \, m(\widetilde{x})]$ & \multicolumn{1}{c}{95\% CI}  \\
Apple & -0.431 & (-0.433, -0.427) \\
Google & -0.516 & (-0.535, -0.496) \\
Meta & -0.147 & (-0.168, -0.126) \\
Microsoft & -0.538 & (-0.543, -0.533) \\
\bottomrule
\end{tabular}
\end{table}

Tables \ref{newtb.1} and \ref{newtb.2} summarize the results, revealing several key findings. 

First, Table \ref{newtb.1} indicates similar features about the SP estimates to those discussed in Section 5.2 with relatively smaller RMSEs in comparison with those of LS and 2sLS estimates. 

Second, Table \ref{newtb.2} confirms the presence of endogeneity, as the estimated values of $\e[\widetilde{x} \, m(\widetilde{x})]$ are negative across all companies, which should be expected given that these four companies account for approximately more than 19\% of the entire S\&P 500 index. 

Third, for Apple and Microsoft, the estimates of $\beta_0$'s via three methods are all greater than 1, indicating that these two stocks are more volatile than the market. For Google and Meta, after accounting for endogeneity, the confidence intervals cover wider range, indicating there might a chance the stock is more efficient than the market. However, a large proportion of these CIs cover the inefficient part.

\renewcommand{\theequation}{C.\arabic{equation}}

\setcounter{subsection}{2}

\section{Proofs for Sections 3 and 4}\label{Ap.C}

\begin{proof}[Proof of Theorem 3.1]

(i) It follows that
\begin{align*}
\widehat{\bm{\beta}}_{\rm SP}=&(\mathbf{X}^\top \mathbf{M}_v \, \mathbf{X})^{-1} \mathbf{X}^\top \mathbf{M}_v \, \mathbf{y}=\bm{\beta}_0 +(\mathbf{X}^\top \mathbf{M}_v \, \mathbf{X})^{-1} \mathbf{X}^\top \mathbf{M}_v \, {\bm\delta} +(\mathbf{X}^\top \mathbf{M}_v \, \mathbf{X})^{-1} \mathbf{X}^\top \mathbf{M}_v \, \mathbf{e}.
\end{align*}

We then have $\e[\widehat{\bm{\beta}}_{\rm SP}|\mathbf{X}] = \bm{\beta}_0 +(\mathbf{X}^\top \mathbf{M}_v \, \mathbf{X})^{-1} \mathbf{X}^\top \mathbf{M}_v \, \bm{\delta}$. Notice that, as $n\to\infty$,
\begin{align*}
(\mathbf{X}^\top \mathbf{M}_v \mathbf{X})^{-1} \mathbf{X}^\top \mathbf{M}_v \bm{\delta}=\left(\frac{1}{n} \mathbf{X}^\top \mathbf{M}_v \, \mathbf{X}\right)^{-1} \frac{1}{n} \mathbf{X}^\top \mathbf{M}_v \bm{\delta}\to_P 0,
\end{align*}
which follows from the proof of part (ii) of Theorem 3.1. The assertion holds.
\smallskip

(ii) Notice that $\widehat{\bm{\beta}}_{\rm SP}=\bm{\beta}_0 +(\mathbf{X}^\top \mathbf{M}_v \mathbf{X})^{-1} \mathbf{X}^\top \mathbf{M}_v (\bm{\delta}+\mathbf{e})$ and hence 
\begin{align*}
{\sqrt{n}\left(\DF{1}{n} \mathbf{X}^\top \mathbf{M}_v \mathbf{X}\right)(\widehat{\bm{\beta}}_{\rm SP}-\bm{\beta}_0)= \DF{1}{\sqrt{n}} \mathbf{X}^\top \mathbf{M}_v \, \mathbf{e}+ \DF{1}{\sqrt{n}} \mathbf{X}^\top \mathbf{M}_v \bm{\delta}.}
\end{align*}

We first consider the convergence of $\DF{1}{n} \mathbf{X}^\top \mathbf{M}_v \mathbf{X}$. Note that
\begin{align*}
&\DF{1}{n} \mathbf{X}^\top \mathbf{M}_v \mathbf{X}=\DF{1}{n} \mathbf{X}^\top \mathbf{X}-\DF{1}{n} \mathbf{X}^\top \mathbf{V}(\mathbf{V}^\top \mathbf{V})^{-1}
\mathbf{V}^\top \mathbf{X}\\
=&\DF{1}{n}\sum_{i=1}^n\mathbf{x}_i\mathbf{x}_i^\top-
\DF{1}{n}\sum_{i=1}^n\mathbf{x}_i\mathbf{V}_k(\mathbf{x}_i)^\top \left(\DF{1}{n}\mathbf{V}^\top \mathbf{V}\right)^{-1}\DF{1}{n}\sum_{i=1}^n\mathbf{V}_k(\mathbf{x}_i)
\mathbf{x}_i^\top\\
=&\e[\mathbf{x}\mathbf{x}^\top]-\e[\mathbf{x}\mathbf{V}_k(\mathbf{x})^\top]
\e[\mathbf{V}_k(\mathbf{x})\mathbf{x}^\top]+o_P(1)\\
\to_P&\e[\mathbf{x}\mathbf{x}^\top]-\sum_{j=1}^\infty\e[\mathbf{x}\psi_j(\mathbf{x})]
\e[\psi_j(\mathbf{x})\mathbf{x}^\top]={\bm\Sigma}_x,
\end{align*}
by Lemma 2.1 as $(k, n)\to(\infty, \infty)$. The asymptotic normality will be derived from $\DF{1}{\sqrt{n}} \mathbf{X}^\top \mathbf{M}_v \, \mathbf{e}$. Because its conditional covariance matrix is $\DF{1}{n}\e[ \mathbf{X}^\top \mathbf{M}_v \, \mathbf{e} \, \mathbf{e}^\top \mathbf{M}_v \mathbf{X}|\mathbf{X}]=\DF{1}{n} \mathbf{X}^\top \mathbf{M}_v \, \mathbf{\Omega} \, \mathbf{M}_v \mathbf{X}$,
and due to the independence data structure, we have
\begin{align*}
\DF{1}{\sqrt{n}}\left(\DF{1}{n} \mathbf{X}^\top \mathbf{M}_v \, \mathbf{\Omega} \, \mathbf{M}_v \mathbf{X}\right)^{-1/2} \mathbf{X}^\top \mathbf{M}_v \, \mathbf{e}\to_{\mathcal D} N\left(0, {\bf I}_d\right)
\end{align*}
as $(k, n)\to (\infty, \infty)$. Meanwhile, we have as $(n,k)\to (\infty, \infty)$, 
{\small
\begin{align*}
  &\DF{1}{n} \mathbf{X}^\top \mathbf{M}_v \, \mathbf{\Omega} \, \mathbf{M}_v \mathbf{X} = \DF{1}{n} \mathbf{X}^\top ({\bf I}-\mathbf{P}_v) \, \mathbf{\Omega} \, ({\bf I}-\mathbf{P}_v) \mathbf{X}\\
=&\DF{1}{n} \mathbf{X}^\top \mathbf{\Omega}  \mathbf{X}+\DF{1}{n} \mathbf{X}^\top \mathbf{P}_v \, \mathbf{\Omega} \, \mathbf{P}_v\mathbf{X}-\DF{1}{n} \mathbf{X}^\top  \mathbf{\Omega} \, \mathbf{P}_v\mathbf{X} -\DF{1}{n} \mathbf{X}^\top \mathbf{P}_v \, \mathbf{\Omega} \, \mathbf{X}\\
=&\DF{1}{n} \mathbf{X}^\top \mathbf{\Omega}  \mathbf{X}+\DF{1}{n^3} \mathbf{X}^\top \mathbf{V}\mathbf{V}^\top \, \mathbf{\Omega} \, \mathbf{V}\mathbf{V}^\top\mathbf{X} -\DF{1}{n^2} \mathbf{X}^\top  \mathbf{\Omega} \, \mathbf{V}\mathbf{V}^\top\mathbf{X} -\DF{1}{n^2} \mathbf{X}^\top \mathbf{V}\mathbf{V}^\top \, \mathbf{\Omega} \, \mathbf{X}+o_P(1)\\
=&\DF{1}{n}\sum_{i=1}^n \mathbf{x}_i\mathbf{x}^\top_i \sigma_i^2+ \DF{1}{n}\sum_{i=1}^n \mathbf{x}_i\mathbf{V}_k(\mathbf{x}_i)^\top  \DF{1}{n}\sum_{i=1}^n \mathbf{V}_k(\mathbf{x}_i)\mathbf{V}_k(\mathbf{x}_i)^\top \sigma_i^2 \DF{1}{n}\sum_{i=1}^n \mathbf{V}_k(\mathbf{x}_i)\mathbf{x}^\top_i \\
&-\DF{1}{n}\sum_{i=1}^n \mathbf{x}_i\mathbf{V}_k(\mathbf{x}_i)^\top \sigma_i^2\DF{1}{n}\sum_{i=1}^n \mathbf{V}_k(\mathbf{x}_i)\mathbf{x}_i^\top-\DF{1}{n}\sum_{i=1}^n \mathbf{x}_i\mathbf{V}_k(\mathbf{x}_i)^\top \DF{1}{n}\sum_{i=1}^n \mathbf{V}_k(\mathbf{x}_i)\mathbf{x}_i^\top \sigma_i^2\\
=& \DF{1}{n}\sum_{i=1}^n \e[\mathbf{x}_i \, \mathbf{x}_i^\top \, \sigma_i^2] + \frac{1}{n^3} \sum_{i=1}^n \sum_{j=1}^n \sum_{\ell=1}^n \e[\mathbf{x}_i \, \mathbf{V}_k^{\top}(\mathbf{x}_i)]\, \e[\mathbf{V}_k(\mathbf{x}_j) \, \mathbf{V}_k^{\top}(\mathbf{x}_j)\, \sigma_j^2]\, \e[\mathbf{V}_k(\mathbf{x}_\ell) \, \mathbf{x}_\ell^{\top}]\\
- & \frac{1}{n^2} \sum_{i=1}^n \sum_{\ell=1}^n \e[\mathbf{x}_i \, \mathbf{V}_k^{\top}(\mathbf{x}_i) \, \sigma_i^2] \, \e[\mathbf{V}_k(\mathbf{x}_\ell) \, \mathbf{x}_\ell^{\top}]
\nonumber\\
- & \frac{1}{n^2} \sum_{i=1}^n \sum_{\ell =1}^n \e\left[\mathbf{x}_i\mathbf{V}_k(\mathbf{x}_i)^\top\right] \, \e\left[\mathbf{V}_k(\mathbf{x}_\ell)\mathbf{x}_\ell^\top \sigma_\ell^2\right] + o_P(1)
\nonumber\\
& \equiv: \frac{1}{n} \e\left[\mathbf{X}^\top \mathbf{M}_v \, \mathbf{\bm{\Omega}} \, \mathbf{M}_v \mathbf{X}\right] +  o_P(1) \equiv: \bm{\Gamma}_x + o_P(1),
\end{align*}
where $\bm{\Gamma}_{x} = \lim_{n\rightarrow \infty} \frac{1}{n} \e\left[\mathbf{X}^\top \mathbf{M}_v \, \mathbf{\bm{\Omega}} \, \mathbf{M}_v \mathbf{X}\right]$}, and we have used the following approximations:
\begin{align*}
&\DF{1}{n}\sum_{i=1}^n \mathbf{x}_i\mathbf{x}^\top_i \, \sigma_i^2=\DF{1}{n}\sum_{i=1}^n \e[\mathbf{x}_i \, \mathbf{x}_i^\top \, \sigma_i^2] +O_P(n^{-1/2}),\\
&\DF{1}{n}\sum_{i=1}^n \mathbf{x}_i \, \mathbf{V}_k(\mathbf{x}_i)^\top =\frac{1}{n} \sum_{i=1}^n \e[\mathbf{x}_i \, \mathbf{V}_k(\mathbf{x}_i)^\top] +O_P(k^{1/2}n^{-1/2}),\\
&\DF{1}{n}\sum_{i=1}^n \mathbf{V}_k(\mathbf{x}_i) \, \mathbf{V}_k(\mathbf{x}_i)^\top \, \sigma_i^2=\DF{1}{n}\sum_{i=1}^n \e[\mathbf{V}_k(\mathbf{x}_i) \, \mathbf{V}_k(\mathbf{x}_i)^\top \, \sigma_i^2] +O(kn^{-1/2}),\\
&\DF{1}{n}\sum_{i=1}^n \mathbf{V}_k(\mathbf{x}_i) \, \mathbf{x}_i^\top \, \sigma_i^2= \DF{1}{n}\sum_{i=1}^n \e\left[\mathbf{V}_k(\mathbf{x}_i) \, \mathbf{x}_i^{\top} \, \sigma_i^2\right] +O_P(k^{1/2}n^{-1/2}),
\end{align*}
under the conditions of Theorem 3.1. Since they are quite easily verified, we omit their proofs.

It remains to show $\DF{1}{\sqrt{n}} \mathbf{X}^\top \mathbf{M}_v \bm{\delta}\to_P0$.  Notice that
\begin{align*}
 \frac{1}{\sqrt{n}}\mathbf{X}^\top \mathbf{M}_v \bm{\delta}=&  \frac{1}{\sqrt{n}}\mathbf{X}^\top \bm{\delta} -\frac{1}{\sqrt{n}}\mathbf{X}^\top \mathbf{V}(\mathbf{V}^\top\mathbf{V})^{-1} \mathbf{V}^\top\bm{\delta}\\
 =&\frac{1}{\sqrt{n}}\mathbf{X}^\top \bm{\delta} -\frac{1}{n}\mathbf{X}^\top \mathbf{V}\frac{1}{\sqrt{n}}\mathbf{V}^\top\bm{\delta}
 (1+o_P(1))\\
 =&\frac{1}{\sqrt{n}}\sum_{i=1}^n\mathbf{x}_i \delta_k(\mathbf{x}_i)
 -\frac{1}{n}\sum_{i=1}^n\mathbf{x}_i \, \mathbf{V}_k^\top (\mathbf{x}_i)\frac{1}{\sqrt{n}}\sum_{i=1}^n \mathbf{V}_k(\mathbf{x}_i) \delta_k(\mathbf{x}_i) \,  (1+o_P(1)).
\end{align*}

Due to the orthogonality of the basis function $\psi_j(\mathbf{x})$, we have $\e[\mathbf{V}_k(\mathbf{x}_i) \delta_k(\mathbf{x}_i)]=0$. Thus, the above second term is, in view of the i.i.d. property of $\{{\bf x}_i: i\geq 1\}$, of smaller order than that of the first term, so it is sufficient to focus on the first term. Observe that by Assumption 3.1(iii),
\bea
\frac{1}{\sqrt{n}}\sum_{i=1}^n \mathbf{x}_i \, \delta_k(\mathbf{x}_i) & = & \sqrt{n}\, \e[\mathbf{x} \delta_k(\mathbf{x})] \, (1+o_P(1)) = \sqrt{n}\sum_{j=k+1}^{\infty}\e[\mathbf{x} \psi_j(\mathbf{x})] \, \gamma_j \, (1+o_P(1))
\nonumber\\
& = & \sqrt{n}\sum_{j=k+1}^{\infty}{\bm \lambda}_j \, \gamma_j \, (1+o_P(1))=o_P(1).
\nonumber
\eea

(iii)\ Notice that
{\footnotesize
\begin{align*}
  &\DF{1}{n} \mathbf{X}^\top \mathbf{M}_v \,(\widehat{\bf \Omega}- \mathbf{\Omega}) \, \mathbf{M}_v \mathbf{X} = \DF{1}{n} \mathbf{X}^\top ({\bf I}-\mathbf{P}_v) \, (\widehat{\bf \Omega}- \mathbf{\Omega}) \, ({\bf I}-\mathbf{P}_v) \mathbf{X}\\
=&\DF{1}{n} \mathbf{X}^\top (\widehat{\bf \Omega}- \mathbf{\Omega})  \mathbf{X}+\DF{1}{n} \mathbf{X}^\top \mathbf{P}_v  (\widehat{\bf \Omega}- \mathbf{\Omega})  \mathbf{P}_v\mathbf{X}-\DF{1}{n} \mathbf{X}^\top  (\widehat{\bf \Omega}- \mathbf{\Omega})  \mathbf{P}_v\mathbf{X} -\DF{1}{n} \mathbf{X}^\top \mathbf{P}_v (\widehat{\bf \Omega}- \mathbf{\Omega}) \mathbf{X}\\
=&\DF{1}{n} \mathbf{X}^\top (\widehat{\bf \Omega}- \mathbf{\Omega})\mathbf{X}+\DF{1}{n^3} \mathbf{X}^\top \mathbf{V}\mathbf{V}^\top (\widehat{\bf \Omega}- \mathbf{\Omega}) \mathbf{V}\mathbf{V}^\top\mathbf{X}\\
&-\DF{1}{n^2} \mathbf{X}^\top (\widehat{\bf \Omega}- \mathbf{\Omega}) \mathbf{V}\mathbf{V}^\top\mathbf{X} -\DF{1}{n^2} \mathbf{X}^\top \mathbf{V}\mathbf{V}^\top (\widehat{\bf \Omega}- \mathbf{\Omega}) \mathbf{X}+o_P(1)\\
=&\DF{1}{n}\sum_{i=1}^n \mathbf{x}_i\mathbf{x}^\top_i(\widehat{e}_i^2-\sigma_i^2)+ \DF{1}{n}\sum_{i=1}^n \mathbf{x}_i\mathbf{V}_k(\mathbf{x}_i)^\top  \DF{1}{n}\sum_{i=1}^n \mathbf{V}_k(\mathbf{x}_i)\mathbf{V}_k(\mathbf{x}_i)^\top (\widehat{e}_i^2- \sigma_i^2) \DF{1}{n}\sum_{i=1}^n \mathbf{V}_k(\mathbf{x}_i)\mathbf{x}^\top_i \\
&-\DF{1}{n}\sum_{i=1}^n \mathbf{x}_i\mathbf{V}_k(\mathbf{x}_i)^\top(\widehat{e}_i^2-\sigma_i^2)\DF{1}{n}\sum_{i=1}^n \mathbf{V}_k(\mathbf{x}_i)\mathbf{x}_i^\top-\DF{1}{n}\sum_{i=1}^n \mathbf{x}_i\mathbf{V}_k(\mathbf{x}_i)^\top \DF{1}{n}\sum_{i=1}^n \mathbf{V}_k(\mathbf{x}_i)\mathbf{x}_i^\top (\widehat{e}_i^2-\sigma_i^2).
\end{align*}}

It suffices to show that
\begin{align}
I_1\equiv&\DF{1}{n}\sum_{i=1}^n \mathbf{x}_i\mathbf{x}^\top_i \, (\widehat{e}_i^2-\sigma_i^2)=o_P(1),\label{th3.1D1}\\
I_2\equiv&\DF{1}{n}\sum_{i=1}^n \mathbf{V}_k(\mathbf{x}_i)\mathbf{V}_k(\mathbf{x}_i)^\top (\widehat{e}_i^2-\sigma_i^2)=o_P(k^{-1}),\label{th3.1D2}\\
I_3\equiv&\DF{1}{n}\sum_{i=1}^n \mathbf{V}_k(\mathbf{x}_i)\mathbf{x}_i^\top (\widehat{e}_i^2-\sigma_i^2)=o_P(k^{-1/2}), \label{th3.1D3}
\end{align}
in view of $\e\left[\|\mathbf{V}_k({\bf x})\|^2\right] \leq C\, k$.

To do this, note that
\bea
\widehat{e}_i^2 & = & (y_i-{\bf x}_i^\top \widehat{\bm{\beta}}_{\rm SP}-\mathbf{V}_k({\bf x}_i)^\top \widehat{\bm{\gamma}})^2 = [e_i -{\bf x}_i^\top( \widehat{\bm{\beta}}_{\rm SP}-{\bm \beta}_0)-\mathbf{V}_k({\bf x}_i)^\top ({\bm \gamma}-\widehat{\bm{\gamma}})+\delta_k({\bf x}_i)]^2 
\nonumber\\
&= & e_i^2 -2e_i [{\bf x}_i^\top( \widehat{\bm{\beta}}_{\rm SP}-{\bm \beta}_0)+\mathbf{V}_k({\bf x}_i)^\top ({\bm \gamma}-\widehat{\bm{\gamma}})-\delta_k({\bf x}_i)] 
\nonumber\\
& + & [{\bf x}_i^\top( \widehat{\bm{\beta}}_{\rm SP}-{\bm \beta}_0)+\mathbf{V}_k({\bf x}_i)^\top ({\bm \gamma}- \widehat{\bm{\gamma}}) -\delta_k({\bf x}_i)]^2.
\nonumber
\eea

Thus, for \eqref{th3.1D1},
{\footnotesize
\bea
&& I_1=\DF{1}{n}\sum_{i=1}^n \mathbf{x}_i\mathbf{x}^\top_i(\widehat{e}_i^2-\sigma_i^2)\\
&& =\DF{1}{n}\sum_{i=1}^n \mathbf{x}_i\mathbf{x}^\top_i(e_i^2-\sigma_i^2)+\DF{2}{n}\sum_{i=1}^n \mathbf{x}_i\mathbf{x}^\top_ie_i[{\bf x}_i^\top( \widehat{\bm{\beta}}_{\rm SP}-{\bm \beta}_0)-\mathbf{V}_k({\bf x}_i)^\top ({\bm \gamma}-\widehat{\bm{\gamma}})+\delta_k({\bf x}_i)]\nonumber\\
&&+\DF{1}{n}\sum_{i=1}^n \mathbf{x}_i\mathbf{x}^\top_i[{\bf x}_i^\top( \widehat{\bm{\beta}}_{\rm SP}-{\bm \beta}_0)-\mathbf{V}_k({\bf x}_i)^\top (\widehat{\bm{\gamma}}-{\bm \gamma})+\delta_k({\bf x}_i)]^2\equiv I_{11}+I_{12}+I_{13}.
\nonumber
\eea}

Moreover, we have
\begin{align*}
\e[\|I_{11}\|^2]=&\DF{1}{n^2}\sum_{i=1}^n \e[\|\mathbf{x}_i\mathbf{x}^\top_i\|^2(e_i^2-\sigma_i^2)^2]\\
=&\DF{1}{n^2}\sum_{i=1}^n \e\{\|\mathbf{x}_i\mathbf{x}^\top_i\|^2\e[(e_i^2-\sigma_i^2)^2|{\bf x}_i]\}
\le \frac{\mu_4}{n} \, \e[\|{\bf x}\|^4]=o(1),
\end{align*}
where $\mu_4=\e[e_i^4|{\bf x}_i]$; and
{\footnotesize
\bea
\| I_{13}\| & \leq & \DF{1}{n}\sum_{i=1}^n \|\mathbf{x}_i\|^2[{\bf x}_i^\top( \widehat{\bm{\beta}}_{\rm SP}-{\bm \beta}_0)+\mathbf{V}_k({\bf x}_i)^\top ({\bm \gamma}-\widehat{\bm{\gamma}})-\delta_k({\bf x}_i)]^2
\nonumber\\
&\le & \DF{6}{n}\sum_{i=1}^n \|\mathbf{x}_i\|^2[{\bf x}_i^\top( \widehat{\bm{\beta}}_{\rm SP}-{\bm \beta}_0)]^2+\DF{6}{n}\sum_{i=1}^n \|\mathbf{x}_i\|^2[\mathbf{V}_k({\bf x}_i)^\top (\widehat{\bm{\gamma}}-{\bm \gamma})]^2
\nonumber\\
&+ &\DF{6}{n}\sum_{i=1}^n \|\mathbf{x}_i\|^2\delta_k^2({\bf x}_i)\le \| \widehat{\bm{\beta}}_{\rm SP}-{\bm \beta}_0)\|^2\DF{6}{n}\sum_{i=1}^n \|\mathbf{x}_i\|^4+\|\widehat{\bm{\gamma}}-{\bm \gamma}\|^2\DF{6}{n}\sum_{i=1}^n \|\mathbf{x}_i\|^2\|\mathbf{V}_k({\bf x}_i)\|^2
\nonumber\\
&+ &\DF{6}{n}\sum_{i=1}^n \|\mathbf{x}_i\|^2\delta_k^2({\bf x}_i)
= O_P(\|\widehat{\bm{\beta}}_{\rm SP}-{\bm \beta}_0)\|^2)+O_P(\|\widehat{\bm{\gamma}}-{\bm \gamma}\|^2) \e[\|\mathbf{x}_i\|^2\|\mathbf{V}_k({\bf x}_i)\|^2]
\nonumber\\
& + & O_P(\e[\|\mathbf{x}_i\|^2\delta_k^2({\bf x}_i)])
= O_P(\|\widehat{\bm{\beta}}_{\rm SP}-{\bm \beta}_0)\|^2)+O_P(k\|\widehat{\bm{\gamma}}-{\bm \gamma}\|^2)+
O_P\left(\{\e[\delta_k^4({\bf x}_i)]\}^{1/2}\right)
\nonumber\\
& = & o_P(1), 
\nonumber
\eea
because
\begin{align*}
k\|\widehat{\bm{\gamma}}-{\bm \gamma}\|^2\le&3k\|\DF{1}{n}{\bf V}^\top {\bf X}\|^2\|\widehat{\bm{\beta}}_{\rm SP}-{\bm \beta}_0\|^2+3k\|\DF{1}{n}{\bf V}^\top{\bf e}\|^2+3k\|\DF{1}{n}{\bf V}^\top{\bm \delta}\|^2\\
=&O_P(n^{-1}k^2)
+3k\|\DF{1}{n}\sum_{i=1}^n {\bf V}_k({\bf x}_i)e_i\|^2+k\|\DF{1}{n}\sum_{i=1}^n {\bf V}_k({\bf x}_i)\delta_k({\bf x}_i)\|^2\\
=&O_P(k^2n^{-1})+O_P(k^2n^{-1}\e[\delta_k^2({\bf x})])=o_P(1),
\end{align*}
where we have used $\e[e|{\bf x}]=0$, $\e[{\bf V}_k({\bf x})\delta_k({\bf x})]=0$}, the i.i.d. structure of $\{{\bf x}_i\}$, and $k^4=o(n)$ in Assumption 3.1, and $\e[\delta_k^2(\mathbf{x})]=\sum_{j=k+1}^\infty \gamma_j^2=o(1)$ as $k\rightarrow \infty$,
\bea
&&\e[\delta_k^4(\mathbf{x})]=\int \delta_k^4(x)f_{\bf x}(x)dx= \sum_{j=k+1}^\infty \gamma_j^4\int \psi_j^4(x) f_{\bf x}(x)dx 
\nonumber\\
&&+6 \sum_{s,j=k+1}^\infty \gamma_j^2\gamma_s^2\int \psi_j^2(x)\psi_s^2(x) f_{\bf x}(x)dx\le C_1\sum_{j=k+1}^\infty \gamma_j^4+C_2\sum_{s,j=k+1}^\infty \gamma_j^2\gamma_s^2=o(1).
\nonumber
\eea

The derivation of $I_{12}=o_P(1)$ follows similarly, so we omit the detail. For \eqref{th3.1D2},
{\footnotesize
\bea
&& I_2=\DF{1}{n}\sum_{i=1}^n \mathbf{V}_k(\mathbf{x}_i)\mathbf{V}_k(\mathbf{x}_i)^\top \left(\widehat{e}_i^2-\sigma_i^2\right)=\DF{1}{n}\sum_{i=1}^n \mathbf{V}_k(\mathbf{x}_i)\mathbf{V}_k(\mathbf{x}_i)^\top \, (e_i^2-\sigma_i^2)
\nonumber\\
&&-\DF{2}{n}\sum_{i=1}^n \mathbf{V}_k(\mathbf{x}_i)\mathbf{V}_k(\mathbf{x}_i)^\top e_i\, \left[{\bf x}_i^\top( \widehat{\bm{\beta}}_{\rm SP}-{\bm \beta}_0)+\mathbf{V}_k({\bf x}_i)^\top ({\bm \gamma}-\widehat{\bm{\gamma}})-\delta_k({\bf x}_i)\right]
\nonumber\\
&&+\DF{1}{n}\sum_{i=1}^n \mathbf{V}_k(\mathbf{x}_i)\mathbf{V}_k(\mathbf{x}_i)^\top \left[{\bf x}_i^\top \left(\widehat{\bm{\beta}}_{\rm SP}-{\bm \beta}_0\right)+\mathbf{V}_k({\bf x}_i)^\top ({\bm \gamma}- \widehat{\bm{\gamma}}) -\delta_k({\bf x}_i)\right]^2
=I_{21}+I_{22}+I_{23}.
\nonumber
\eea}

Note that
\begin{align*}
\e[\|I_{21}\|^2]=&\DF{1}{n^2}\sum_{i=1}^n \e\|\mathbf{V}_k(\mathbf{x}_i)\mathbf{V}_k(\mathbf{x}_i)^\top\|^2 (e_i^2-\sigma_i^2)^2
\le \DF{\mu_4}{n} \, \e\|\mathbf{V}_k(\mathbf{x})\|^4=\DF{C}{n} \, k^2=o(k^{-2});
\nonumber\\
\|I_{23}\|\le&\DF{1}{n}\sum_{i=1}^n \|\mathbf{V}_k(\mathbf{x}_i)\|^2 \left[{\bf x}_i^\top( \widehat{\bm{\beta}}_{\rm SP}-{\bm \beta}_0)\right]^2+\DF{1}{n}\sum_{i=1}^n \|\mathbf{V}_k(\mathbf{x}_i)\|^2\delta_k^2({\bf x}_i)\\
&+\DF{1}{n}\sum_{i=1}^n \|\mathbf{V}_k(\mathbf{x}_i)\|^2 \left[\mathbf{V}_k({\bf x}_i)^\top ({\bm \gamma}- \widehat{\bm{\gamma}})\right]^2\\
\le&\|\widehat{\bm{\beta}}_{\rm SP}-{\bm \beta}_0\|^2 \, O_P\left(\e \|\mathbf{V}_k(\mathbf{x})\|^2\|{\bf x}\|^2\right)+O_P\left(\e\|\mathbf{V}_k(\mathbf{x})\|^2\delta_k^2({\bf x})\right)\\
&+\|{\bm \gamma}- \widehat{\bm{\gamma}}\|^2 \, O_P\left(\e \|\mathbf{V}_k(\mathbf{x})\|^4\right)\\
=&O_P(k \, n^{-1})+O_P(k \, \e[\delta_k^2({\bf x})])+O_P(k^{2} \, \|{\bm \gamma}- \widehat{\bm{\gamma}}\|^2)=o_P(k^{-1}),
\end{align*}
in view of Assumption 3.1 and the order of $\|{\bm \gamma}- \widehat{\bm{\gamma}}\|^2$ in the above derivations.

We also omit the verification of $I_{22}$ due to the same reason as before. Meanwhile, as the condition ensuring $I_3$ to hold is no more restrictive than that for $I_{22}$ and $I_3$ holds automatically. We therefore have completed the proof of Theorem 3.1. 

\end{proof}

\begin{proof}[Proof of Theorem 3.2]

Notice that, for $i\ne j$, by the mutual independence of ${\bf x}_i$ and ${\bf x}_j$ and the orthogonality of $\{\psi_\ell(x)\}$ we have
\begin{align*}
&\e\left[\left(\sum_{k=k_{\rm \min}}^{k_{\rm \max}} {\bf V}^{\top}_k({\bf x}_i) {\bf V}_k({\bf x}_j)\right)^2\right]\\
=&\sum_{k=k_{\min}}^{k_{\max}} \e[{\bf V}^{\top}_k({\bf x}_i) {\bf V}_k({\bf x}_j)]^2+\sum_{k=k_{\min}}^{k_{\max}} \sum_{p=k_{\min},\ne k}^{k_{\max}} \e[V^{\top}_p({\bf x}_i) {\bf V}_p({\bf x}_j) {\bf V}^{\top}_k({\bf x}_i) {\bf V}_k({\bf x}_j)]\\
=&\sum_{k=k_{\min}}^{k_{\max}} \e \text{tr}[{\bf V}^{\top}_k({\bf x}_i) {\bf V}_k({\bf x}_j)]^2+\sum_{k=k_{\min}}^{k_{\max}} \sum_{p=k_{\min},\ne k}^{k_{\max}} \e\text{tr}[{\bf V}^{\top}_p({\bf x}_i) {\bf V}_p({\bf x}_j) {\bf V}^{\top}_k({\bf x}_i) {\bf V}_k({\bf x}_j)]\\
=&\sum_{k=k_{\min}}^{k_{\max}}  \text{tr}\{[\e {\bf V}_k({\bf x}_i) {\bf V}^{\top}_k({\bf x}_i)] [\e {\bf V}_k({\bf x}_j) {\bf V}_k^\top({\bf x}_j)]\}\\
&+\sum_{k=k_{\min}}^{k_{\max}} \sum_{p=k_{\min},\ne k}^{k_{\max}} \text{tr}\{[\e {\bf V}_p({\bf x}_j) {\bf V}^{\top}_k({\bf x}_j)][\e {\bf V}_k({\bf x}_i){\bf V}^{\top}_p({\bf x}_i)]\}\\
=&\sum_{k=k_{\min}}^{k_{\max}} k+\sum_{k=k_{\min}}^{k_{\max}} \sum_{p=k_{\min},\ne k}^{k_{\max}} \min(k,p)\\
=&\DF{1}{2}(k_{\max}+k_{\min})(k_{\max}-k_{\min}+1)+2\sum_{k=k_{\min}+1}^{k_{\max}} \sum_{p=k_{\min}}^{k-1} p\\
=&\DF{1}{2}(k_{\max}+k_{\min})(k_{\max}-k_{\min}+1)+\sum_{k=k_{\min}+1}^{k_{\max}}
(k+k_{\min}-1)(k-k_{\min})\\
=&\DF{1}{2}(k_{\max}+k_{\min})(k_{\max}-k_{\min}+1)\\
&+\sum_{k=k_{\min}+1}^{k_{\max}}k^2-k_{\min}^2(k_{\max}-k_{\min}-1)-
\sum_{k=k_{\min}+1}^{k_{\max}}(k-k_{\min})\\
=&\DF{k_{\max}(k_{\max}+1)(2k_{\max}+1)-(k_{\min}+1)(k_{\min}+2)(2k_{\min}+3)}{6}\\
&-k_{\min}^2(k_{\max}-k_{\min}-1)+k_{\min}(k_{\max}-k_{\min})\\
= & \DF{1}{3}k_{\max}^3(1+o(1)).
\end{align*}

Next, we consider the convergence of $\widehat{\sigma}_e^2\to_P\sigma_{\varepsilon}^2=\e[\varepsilon^2]$ as $n\to\infty$.
Under $H_0$, $m({\bf x})=0$ almost surely, and the model $y={\bf x}^\top {\bm\beta}+\varepsilon$ satisfies $\e[\varepsilon|{\bf x}]=0$. Moreover, $\widehat{\bm \beta}=({\bf X}^\top {\bf X})^{-1}{\bf X}^\top Y={\bm \beta}+({\bf X}^\top {\bf X})^{-1}{\bf X}^\top {\bm \varepsilon}$ where ${\bm \varepsilon}=(\varepsilon_1, \cdots, \varepsilon_n)^\top$, and
\begin{align*}
 \widetilde{e}_i=&y_i-{\bf x}_i^\top \widehat{\bm \beta}={\bf x}_i^\top({\bm\beta}-\widehat{\bm \beta})+\varepsilon_i,\ \ \ i=1,\cdots,n.
\end{align*}
It follows that
\begin{align*}
\widehat{\sigma}_e^2=&\DF{1}{n}\sum_{i=1}^n\widetilde{e}^2_i=\DF{1}{n}\sum_{i=1}^n[{\bf x}_i^\top({\bm\beta}-\widehat{\bm \beta})+\varepsilon_i]^2\\
=&\DF{1}{n}\sum_{i=1}^n\varepsilon_i^2+\DF{1}{n}\sum_{i=1}^n[{\bf x}_i^\top({\bm\beta}-\widehat{\bm \beta})]^2+\DF{2}{n}\sum_{i=1}^n[{\bf x}_i^\top({\bm\beta}-\widehat{\bm \beta})]\varepsilon_i.
\end{align*}

It is clear that due to the weak LLN,
\begin{align*}
\DF{1}{n}\sum_{i=1}^n\varepsilon_i^2\to_P \sigma_{\varepsilon}^2,
\end{align*}
as $n\to\infty$; and then we are about to show the second term to be $o_P(1)$ from which the third term will be $o_P(1)$ immediately by Cauchy-Schwarz inequality. Toward this, note that
\begin{align*}
\DF{1}{n}\sum_{i=1}^n[{\bf x}_i^\top({\bm\beta}-\widehat{\bm \beta})]^2
=&\DF{1}{n}({\bm\beta}-\widehat{\bm \beta})^\top \sum_{i=1}^n{\bf x}_i{\bf x}_i^\top({\bm\beta}-\widehat{\bm \beta})\\
=&\DF{1}{n}({\bm\beta}-\widehat{\bm \beta})^\top {\bf X}^\top {\bf X}({\bm\beta}-\widehat{\bm \beta})
=\DF{1}{n}({\bm\beta}-\widehat{\bm \beta})^\top {\bf X}^\top {\bm \varepsilon}.
\end{align*}

Hence,
\begin{align*}
\DF{1}{n}\e \|{\bf X}^\top {\bm \varepsilon}\|^2=&\DF{1}{n}\e\|\sum_{i=1}^n{\bf x}_i\varepsilon_i\|^2
=\DF{1}{n}\e\sum_{i=1}^n\|{\bf x}_i\varepsilon_i\|^2\\
=&\sigma_{\varepsilon}^2\DF{1}{n}\sum_{i=1}^n\e\|{\bf x}_i\|^2=\sigma_{\varepsilon}^2\e\|{\bf x}_1\|^2,
\end{align*}
which implies ${\bf X}^\top {\bm \varepsilon}=O_P(\sqrt{n})$, and therefore, the second term is $o_P(1)$ since $\widehat{\bm \beta}\to_P {\bm\beta}$.

This finishes the proof of $\widehat{\sigma}_e^2\to_P\sigma_{\varepsilon}^2$ as $n\to\infty$. Hence, this assertion allows us to replace $\widehat{\sigma}_e^4$ by $\sigma_{\varepsilon}^4$ due to the continuous mapping theorem when we prove the normality of $T_n/\sigma_n$. We now conclude $\sigma_n=O_P(nk_{\max}^{3/2})$ that is used in the sequel.

Notice that
\begin{align*}
T_n =& \sum_{i=1}^n \sum_{j=1, \neq i}^n \sum_{k=k_{ \min}}^{k_{\max}} {\bf V}^{\top}_k({\bf x}_i) {\bf V}_k({\bf x}_j) \widetilde{e}_i \, \widetilde{e}_j\\
=&\sum_{i=1}^n \sum_{j=1, \neq i}^n\sum_{k=k_{ \min}}^{k_{\max}} {\bf V}^{\top}_k({\bf x}_i) {\bf V}_k({\bf x}_j) [{\bf x}_j^\top({\bm\beta}-\widehat{\bm \beta})+\varepsilon_j][{\bf x}_i^\top({\bm\beta}-\widehat{\bm \beta})+\varepsilon_i]\\
=&\sum_{i=1}^n \sum_{j=1, \neq i}^n\sum_{k=k_{ \min}}^{k_{\max}} {\bf V}^{\top}_k({\bf x}_i) {\bf V}_k({\bf x}_j)\varepsilon_j \varepsilon_i\\
&+\sum_{i=1}^n \sum_{j=1, \neq i}^n\sum_{k=k_{ \min}}^{k_{\max}} {\bf V}^{\top}_k({\bf x}_i) {\bf V}_k({\bf x}_j) {\bf x}_j^\top({\bm\beta}-\widehat{\bm \beta}) {\bf x}_i^\top({\bm\beta}-\widehat{\bm \beta})\\
&+2\sum_{i=1}^n \sum_{j=1, \neq i}^n\sum_{k=k_{ \min}}^{k_{\max}} {\bf V}^{\top}_k({\bf x}_i) {\bf V}_k({\bf x}_j) {\bf x}_j^\top({\bm\beta}-\widehat{\bm \beta})\varepsilon_i\\
=&T_{1n}+T_{2n}+T_{3n}, \ \text{say}.
\end{align*}

Then, we shall show, as $n\to\infty$,
\begin{align*}
\DF{1}{\sigma_n}T_{1n}=\DF{2}{\sigma_n}\sum_{i=2}^n\left[\sum_{j=1}^{i-1}\sum_{k=k_{ \min}}^{k_{\max}} {\bf V}^{\top}_k({\bf x}_i) {\bf V}_k({\bf x}_j)\varepsilon_j\right] \varepsilon_i\to_{\cal D} N(0,1),
\end{align*}
and $\DF{1}{\sigma_n}T_{in}=o_P(1), i=2,3.$

Because $\e[\varepsilon_i|{\bf x}_i]=0$ (a.s.) under $H_0$, and $\{({\bf x}_i, \varepsilon_i): i\geq 1\}$ is a sequence of i.i.d. random variables, one can construct a filtration ${\cal F}_i=\sigma({\bf x}_s, s\le i; \varepsilon_t, t<i)$ such that $(\varepsilon_i, {\cal F}_i)$ is a martingale difference sequence. Thus, in order to show $T_{1n}/\sigma_n\to_{\cal D}N(0,1)$, it suffices to check whether it satisfies the convergence of the conditional variance process and the Linderberg condition in Corollary 3.1 of \cite{peterhall1980}.

Towards this, denote $a_{ij}\equiv  \sum_{k=k_{ \min}}^{k_{\max}} {\bf V}^{\top}_k({\bf x}_i) {\bf V}_k({\bf x}_j)$ for convenience. The conditional variance process of $T_{1n}/\sigma_n$ is
\begin{align*}
\DF{4\sigma_{\varepsilon}^2}{\sigma_n^2} \sum_{i=2}^n \left[\sum_{j=1}^{i-1}a_{ij}\varepsilon_j \right]^2
=&\DF{4\sigma_{\varepsilon}^2}{\sigma_n^2} \sum_{i=2}^n \sum_{j=1}^{i-1}a_{ij}^2\varepsilon_j^2
+\DF{8\sigma_{\varepsilon}^2}{\sigma_n^2} \sum_{i=3}^n\sum_{j_1=2}^{i-1} a_{ij_1}\varepsilon_{j_1}\sum_{j_2=1}^{j_1-1} a_{ij_2}\varepsilon_{j_2}\\
=&\DF{4\sigma_{\varepsilon}^4}{\sigma_n^2} \sum_{i=2}^n \sum_{j=1}^{i-1} a_{ij}^2
+\DF{4\sigma_{\varepsilon}^2}{\sigma_n^2} \sum_{i=2}^n \sum_{j=1}^{i-1} a_{ij}^2(\varepsilon_j^2-\sigma_{\varepsilon}^2)\\
&+\DF{8\sigma_{\varepsilon}^2}{\sigma_n^2} \sum_{i=3}^n\sum_{j_1=2}^{i-1} a_{ij_1}\varepsilon_{j_1}\sum_{j_2=1}^{j_1-1}a_{ij_2}\varepsilon_{j_2} =I_{1n}+I_{2n}+I_{3n}, \ \ \text{say}.
\end{align*}

We shall show $I_{1n}\to_P 1$, $I_{2n}=o_P(1)$ and $I_{3n}=o_P(1)$. To begin, by the definition of $\sigma_n^2$,
\begin{align*}
\e(I_{1n}-1)^2=&\DF{16}{n^4k_{\max}^6}\e\left[ \sum_{i=2}^n \sum_{j=1}^{i-1}(a_{ij}^2-\e a_{ij}^2)\right]^2 = \DF{16}{n^4k_{\max}^6} \sum_{i=2}^n \e\left[\sum_{j=1}^{i-1}(a_{ij}^2-\e a_{ij}^2)\right]^2\\
&+\DF{32}{n^4k_{\max}^6}\e \sum_{i_1=3}^n\sum_{i_2=2}^{i_1-1} \sum_{j=1}^{i_1-1}(a_{i_1j}^2-\e a_{i_1j}^2) \sum_{j=1}^{i_2-1}(a_{i_2j}^2-\e a_{i_2j}^2)\\
=&\DF{16}{n^4k_{\max}^6} \sum_{i=2}^n \sum_{j=1}^{i-1}\e\left[(a_{ij}^2-\e a_{ij}^2)\right]^2\\
&+\DF{32}{n^4k_{\max}^6} \sum_{i=3}^n \sum_{j_1=2}^{i-1}\sum_{j_2=1}^{j_1-1}\e\left[(a_{ij_1}^2-\e a_{ij_1}^2)(a_{ij_2}^2-\e a_{ij_2}^2)\right]\\
&+\DF{32}{n^4k_{\max}^6} \sum_{i_1=3}^n\sum_{i_2=2}^{i_1-1} \sum_{j_1=1}^{i_1-1}\sum_{j_2=1}^{i_2-1}\e[(a_{i_1j_1}^2-\e a_{i_1j_1}^2) (a_{i_2j_2}^2-\e a_{i_2j_2}^2)]\\
\le &\DF{C}{n^2}k_{\max}^2+\DF{C}{n}k_{\max}^2=o(1),
\end{align*}
in view of $\frac{k_{\max}^2}{n}\rightarrow 0$, because {\footnotesize $\e\left[a_{ij}^4\right]\le \e\left[\left(\sum_{k=k_{ \min}}^{k_{\max}}| {\bf V}^{\top}_k({\bf x}_i) {\bf V}_k({\bf x}_j)|\right)^4\right]\le C\left(\sum_{k=k_{\min}}^{k_{\max}}k \right)^4=Ck_{\max}^8$}. 

Similarly, $|\e(a_{ij_1}^2-\e a_{ij_1}^2)(a_{ij_2}^2-\e a_{ij_2}^2)|\le \e|a_{ij_1}^2a_{ij_2}^2|\le C k_{\max}^8$; and
\begin{align*}
\e[(a_{i_1j_1}^2-\e a_{i_1j_1}^2) (a_{i_2j_2}^2-\e a_{i_2j_2}^2)]=&\e[(a_{i_1j_1}^2-\e a_{i_1j_1}^2)]\e[(a_{i_2j_2}^2-\e a_{i_2j_2}^2)]=0,
\end{align*}
due to the i.i.d. assumption of the data. This implies that $I_{1n}\to_P1$ as $n\to \infty$.

Next, consider $I_{2n}$. Observe that
\begin{align*}
\e(I_{2n}^2)=&\DF{16}{\sigma_{\varepsilon}^4n^4k_{\max}^6} \e\left[\sum_{i=2}^n \sum_{j=1}^{i-1} a_{ij}^2(\varepsilon_j^2-\sigma_{\varepsilon}^2)\right]^2 = \DF{16}{\sigma_{\varepsilon}^4n^4k_{\max}^6} \sum_{i=2}^n \e\left[\sum_{j=1}^{i-1} a_{ij}^2(\varepsilon_j^2-\sigma_{\varepsilon}^2)\right]^2\\
&+\DF{32}{\sigma_{\varepsilon}^4n^4k_{\max}^6}\sum_{i_1=3}^n\sum_{i_2=2}^{i_1-1} \sum_{j_1=1}^{i_1-1}\sum_{j_2=1}^{i_2-1}\e[ a_{i_1j_1}^2(\varepsilon_{j_1}^2-\sigma_{\varepsilon}^2)  a_{i_2j_2}^2(\varepsilon_{j_2}^2-\sigma_{\varepsilon}^2)]\\
=&\DF{16}{\sigma_{\varepsilon}^4n^4k_{\max}^6} \sum_{i=2}^n \sum_{j=1}^{i-1} \e\left[a_{ij}^4(\varepsilon_j^2-\sigma_{\varepsilon}^2)^2\right]\\
&+\DF{32}{\sigma_{\varepsilon}^4n^4k_{\max}^6} \sum_{i=2}^n \sum_{j_1=2}^{i-1}\sum_{j_2=1}^{j_1-1}\e\left[ a_{ij_1}^2a_{ij_2}^2(\varepsilon_{j_1}^2-\sigma_{\varepsilon}^2)
(\varepsilon_{j_2}^2-\sigma_{\varepsilon}^2)\right]\\
&+\DF{32}{\sigma_{\varepsilon}^4n^4k_{\max}^6}\sum_{i_1=3}^n\sum_{i_2=2}^{i_1-1} \sum_{j_1=1}^{i_1-1}\sum_{j_2=1}^{i_2-1}\e[ a_{i_1j_1}^2(\varepsilon_{j_1}^2-\sigma_{\varepsilon}^2)  a_{i_2j_2}^2(\varepsilon_{j_2}^2-\sigma_{\varepsilon}^2)]\\
\le&\DF{C}{n^2}k_{\max}^2+\DF{32}{\sigma_{\varepsilon}^4n^4k_{\max}^6}
\sum_{i_1=3}^n\sum_{i_2=2}^{i_1-1} \sum_{j_2=1}^{i_2-1}\e[ a_{i_1j_2}^2 a_{i_2j_2}^2(\varepsilon_{j_2}^2-\sigma_{\varepsilon}^2)^2]\\
\le&\DF{C}{n^2}k_{\max}^2+\DF{C}{n}k_{\max}^2=o(1),
\end{align*}
by virtue of the data dependence again. For $I_{3n}$, note that
\begin{align*}
\e(I_{3n}^2)=&\DF{64}{\sigma_{\varepsilon}^4n^4k_{\max}^6} \e\left[\sum_{i=3}^n\sum_{j_1=2}^{i-1} a_{ij_1}\varepsilon_{j_1}\sum_{j_2=1}^{j_1-1}a_{ij_2}\varepsilon_{j_2}\right]^2 = \DF{64}{\sigma_{\varepsilon}^4n^4k_{\max}^6} \sum_{i=3}^n\e\left[\sum_{j_1=2}^{i-1} a_{ij_1}\varepsilon_{j_1}\sum_{j_2=1}^{j_1-1}a_{ij_2}\varepsilon_{j_2}\right]^2\\
&+\DF{128}{\sigma_{\varepsilon}^4n^4k_{\max}^6} \sum_{i_1=4}^n\sum_{i_2=3}^{i_1-1}\e\left[\sum_{j_1=2}^{i_1-1} a_{i_1j_1}\varepsilon_{j_1}\sum_{j_2=1}^{j_1-1}a_{i_1j_2}\varepsilon_{j_2}\right]
\left[\sum_{j_1=2}^{i_2-1} a_{i_2j_1}\varepsilon_{j_1}\sum_{j_2=1}^{j_1-1}a_{i_2j_2}\varepsilon_{j_2}\right]\\
=&\DF{64}{\sigma_{\varepsilon}^4n^4k_{\max}^6} \sum_{i=3}^n\sum_{j_1=2}^{i-1} \e\left[a_{ij_1}\varepsilon_{j_1}\sum_{j_2=1}^{j_1-1}a_{ij_2}\varepsilon_{j_2}\right]^2\\
&+\DF{128}{n^4k_{\max}^6} \sum_{i_1=4}^n\sum_{i_2=3}^{i_1-1}\sum_{j_1=2}^{i_2-1} \sum_{j_2=1}^{j_1-1}\e\left[a_{i_1j_1}a_{i_2j_1}a_{i_1j_2}a_{i_2j_2}\right]\\
=&\DF{64}{n^4k_{\max}^6} \sum_{i=3}^n\sum_{j_1=2}^{i-1} \sum_{j_2=1}^{j_1-1}\e[a_{ij_1}^2a_{ij_2}^2]
+\DF{128}{n^4k_{\max}^6} \sum_{i_1=4}^n\sum_{i_2=3}^{i_1-1}\sum_{j_1=2}^{i_2-1} \sum_{j_2=1}^{j_1-1}\sum_{k=k_{ \min}}^{k_{\max}}\\
&\qquad \e[ {\bf V}^{\top}_k({\bf x}_{i_1}) {\bf V}_k({\bf x}_{j_1}) {\bf V}^{\top}_k({\bf x}_{i_2}) {\bf V}_k({\bf x}_{j_1}) {\bf V}^{\top}_k({\bf x}_{i_1}) {\bf V}_k({\bf x}_{j_2}) {\bf V}^{\top}_k({\bf x}_{i_2}) {\bf V}_k({\bf x}_{j_2})]\\
&+\DF{128}{n^4k_{\max}^6} \sum_{i_1=4}^n\sum_{i_2=3}^{i_1-1}\sum_{j_1=2}^{i_2-1} \sum_{j_2=1}^{j_1-1}\e\left[\sum_{k=k_{ \min}}^{k_{\max}} {\bf V}^{\top}_k({\bf x}_{i_1}) {\bf V}_k({\bf x}_{j_1})\right]^2\e\left[ \sum_{k=k_{ \min}}^{k_{\max}} {\bf V}^{\top}_k({\bf x}_{i_2}) {\bf V}_k({\bf x}_{j_2})\right]^2\\
\le&\DF{C}{n}k_{\max}^2+\DF{C}{k_{\max}}+\DF{128}{n^4k_{\max}^6} \sum_{i_1=4}^n\sum_{i_2=3}^{i_1-1}\sum_{j_1=2}^{i_2-1} \sum_{j_2=1}^{j_1-1}\\
&\qquad \sum_{k=k_{ \min}}^{k_{\max}}\e\left[{\bf V}^{\top}_k({\bf x}_{i_1}) {\bf V}_k({\bf x}_{j_1})\right]^2 \sum_{k=k_{ \min}}^{k_{\max}}\e\left[{\bf V}^{\top}_k({\bf x}_{i_2}){\bf V}_k({\bf x}_{j_2})\right]^2\\
=&\DF{C}{n}k_{\max}^2+\DF{C}{k_{\max}}+\DF{C}{k_{\max}^2}=o(1),
\end{align*}
where we use the i.i.d. property of the data, and $\e[\psi_j({\bf x})]=0$.

Thus, we finish the derivation of the conditional covariance process of $T_{1n}/\sigma_n$ converging to $1$. In what follows, we are about to show that the Linderberg condition is fulfilled. In doing so, it suffices to prove, denoting $\mu_4=\e[\varepsilon^4]$,
\begin{align*}
&16\mu_{4}\DF{1}{\sigma_n^4}\sum_{i=2}^n \e\left[\sum_{j=1}^{i-1}\left(\sum_{k=k_{ \min}}^{k_{\max}} {\bf V}^{\top}_k({\bf x}_j) {\bf V}_k({\bf x}_i) \right)\varepsilon_j\right]^4 = 16\mu_4^2\DF{1}{n^4k_{\max}^6}\sum_{i=2}^n\sum_{j=1}^{i-1}\e\left[\sum_{k=k_{ \min}}^{k_{\max}} {\bf V}^{\top}_k({\bf x}_j) {\bf V}_k({\bf x}_i) \right]^4\\
&+16\mu_4\DF{6\sigma_{\varepsilon}^4}{n^4k_{\max}^6}\sum_{i=3}^n\sum_{j_1=2}^{i-1}
\sum_{j_2=1}^{j_1-1}\e\left[\sum_{k=k_{ \min}}^{k_{\max}} {\bf V}^{\top}_k({\bf x}_{j_1}) {\bf V}_k({\bf x}_i) \right]^2 \left[\sum_{k=k_{ \min}}^{k_{\max}} {\bf V}^{\top}_k({\bf x}_{j_2}) {\bf V}_k({\bf x}_i) \right]^2\\
\le&\DF{C}{n^2}k_{\max}^2+\DF{C}{n}k_{\max}^2=o(1).
\end{align*}

Therefore, invoking Corollary 3.1 of \cite{peterhall1980}, $T_{1n}/\sigma_n\to_{\cal D}N(0,1)$ as $n\to\infty$.

In what follows, we show $T_{2n}/\sigma_n=o_P(1)$ and $T_{3n}/\sigma_n=o_P(1)$. Notice that  $\sigma_n^2=O_P(n^2k_{\max}^3)$, and then
\begin{align*}
\DF{1}{\sigma_n}T_{2n}=&\DF{1}{\sigma_n}\sum_{i=1}^n \sum_{j=1, \neq i}^n\sum_{k=k_{ \min}}^{k_{\max}} {\bf V}^{\top}_k({\bf x}_i) {\bf V}_k({\bf x}_j) {\bf x}_j^\top({\bm\beta}-\widehat{\bm \beta}) {\bf x}_i^\top({\bm\beta}-\widehat{\bm \beta})\\
=&({\bm\beta}-\widehat{\bm \beta})^\top\left[\DF{1}{\sigma_n}\sum_{i=1}^n \sum_{j=1, \neq i}^n\sum_{k=k_{ \min}}^{k_{\max}} {\bf V}^{\top}_k({\bf x}_i) {\bf V}_k({\bf x}_j) {\bf x}_j {\bf x}_i^\top\right]({\bm\beta}-\widehat{\bm \beta})\\
=&\sigma_e^4\sqrt{n}({\bm\beta}-\widehat{\bm \beta})^\top\left[\DF{1}{n^2k_{\max}^{3/2}}\sum_{i=1}^n \sum_{j=1, \neq i}^n\sum_{k=k_{ \min}}^{k_{\max}}  {\bf x}_j {\bf V}^{\top}_k({\bf x}_j) {\bf V}_k({\bf x}_i) {\bf x}_i^\top\right]\sqrt{n}({\bm\beta}-\widehat{\bm \beta}).
\end{align*}

Once we show the term in the square bracket is $o_P(1)$, the assertion is fulfilled. To do so, note that
\begin{align*}
&\DF{1}{n^2k_{\max}^{3/2}}\sum_{i=1}^n \sum_{j=1, \neq i}^n\sum_{k=k_{ \min}}^{k_{\max}}  {\bf x}_j {\bf V}^{\top}_k({\bf x}_j) {\bf V}_k({\bf x}_i) {\bf x}_i^\top\\
=&\DF{1}{n^2k_{\max}^{3/2}}\left[\sum_{i=1}^n \sum_{j=1}^n\sum_{k=k_{ \min}}^{k_{\max}}  {\bf x}_j {\bf V}^{\top}_k({\bf x}_j) {\bf V}_k({\bf x}_i) {\bf x}_i^\top-\sum_{i=1}^n \sum_{k=k_{ \min}}^{k_{\max}}  \|{\bf V}_k({\bf x}_i)\|^2{\bf x}_i  {\bf x}_i^\top\right]\\
=&\DF{1}{n^2k_{\max}^{3/2}}\sum_{k=k_{ \min}}^{k_{\max}}\sum_{i=1}^n \sum_{j=1}^n \left[ {\bf x}_j {\bf V}^{\top}_k({\bf x}_j)-\e ({\bf x}_j {\bf V}^{\top}_k({\bf x}_j))+ \e ({\bf x}_j {\bf V}^{\top}_k({\bf x}_j))\right]\\
&\times \left[ {\bf x}_i^{\top}V_k({\bf x}_i)-\e( {\bf x}_i^{\top}V_k({\bf x}_i))+ \e ({\bf x}_i^{\top}V_k({\bf x}_i))\right] -\DF{1}{n^2k_{\max}^{3/2}}\sum_{i=1}^n \sum_{k=k_{ \min}}^{k_{\max}}  \|V_k({\bf x}_i)\|^2{\bf x}_i  {\bf x}_i^\top\\
=&\DF{1}{k_{\max}^{3/2}}\sum_{k=k_{ \min}}^{k_{\max}}\e ({\bf x}_1 {\bf V}^{\top}_k({\bf x}_1))\e ({\bf V}_k^{\top}({\bf x}_1){\bf x}_1)\\
&+\DF{1}{n}\sum_{j=1}^n\DF{1}{k_{\max}^{3/2}}\sum_{k=k_{ \min}}^{k_{\max}}  [ {\bf x}_j {\bf V}^{\top}_k({\bf x}_j)-\e ({\bf x}_j {\bf V}^{\top}_k({\bf x}_j))]\e({\bf V}_k^{\top}({\bf x}_1){\bf x}_1)\\
&+\DF{1}{n^2k_{\max}^{3/2}}\sum_{k=k_{ \min}}^{k_{\max}}\sum_{i=1}^n \sum_{j=1}^n \left[ {\bf x}_j {\bf V}^{\top}_k({\bf x}_j)-\e ({\bf x}_j {\bf V}^{\top}_k({\bf x}_j))\right]\left[ {\bf x}_i^{\top} {\bf V}_k({\bf x}_i)-\e( {\bf x}_i^{\top} {\bf V}_k({\bf x}_i))\right]\\
&-\DF{1}{n^2k_{\max}^{3/2}}\sum_{i=1}^n \sum_{k=k_{ \min}}^{k_{\max}}  \|{\bf V}_k({\bf x}_i)\|^2{\bf x}_i  {\bf x}_i^\top \equiv I_1+I_2+I_3-I_4, \ \ \text{say}.
\end{align*}

Moreover, $\|\e ({\bf x}_1{\bf V}^{\top}_k({\bf x}_1))\e ({\bf x}_1^\top {\bf V}_k({\bf x}_1))\|\le\|\e ({\bf x}_1^\top {\bf V}_k({\bf x}_1))\|^2 \le \e\|{\bf x}_1\|^2$ since $\e ({\bf x}_1 {\bf V}^{\top}_k({\bf x}_1))$ is the vector of coefficients projecting ${\bf x}$ onto the space generated by ${\bf V}^{\top}_k({\bf x}_1)$. Thus, $I_1=O(k_{\max}^{-1/2})=o(1)$.

For $I_4$, note that
\begin{align*}
&\DF{1}{n^2k_{\max}^{3/2}}\e\left\|\sum_{i=1}^n \sum_{k=k_{ \min}}^{k_{\max}}  \|V_k({\bf x}_i)\|^2{\bf x}_i  {\bf x}_i^\top\right\|
\le \DF{1}{n^2k_{\max}^{3/2}}\sum_{i=1}^n \sum_{k=k_{ \min}}^{k_{\max}}\e[\|V_k({\bf x}_i)\|^2\|{\bf x}_i\|^2]\\
\le& \DF{1}{n^2k_{\max}^{3/2}}\sum_{i=1}^n\sqrt{\e[\|{\bf x}_i\|^4]} \sum_{k=k_{ \min}}^{k_{\max}}k\le \DF{1}{n}k_{\max}^{1/2}=o(1),
\end{align*}
since $\e[\psi_j^4(x)]$ is uniformly bounded in $j$. This implies $I_4=o_P(1)$.

For $I_2$, note that
\begin{align*}
\e\|I_2\|^2=&\DF{1}{n^2}\e\left\|\sum_{j=1}^n\DF{1}{k_{\max}^{3/2}}\sum_{k=k_{ \min}}^{k_{\max}}  [ {\bf x}_j {\bf V}^{\top}_k({\bf x}_j)-\e ({\bf x}_j {\bf V}^{\top}_k({\bf x}_j))]\e({\bf V}_k^{\top}({\bf x}_1){\bf x}_1)\right\|^2\\
=&\DF{1}{n^2}\sum_{j=1}^n\e\left\|\DF{1}{k_{\max}^{3/2}}\sum_{k=k_{ \min}}^{k_{\max}}  [ {\bf x}_j {\bf V}^{\top}_k({\bf x}_j)-\e ({\bf x}_j {\bf V}^{\top}_k({\bf x}_j))]\e ({\bf V}_k^{\top}({\bf x}_1){\bf x}_1))\right\|^2\\
\le&\DF{1}{n^2}\sum_{j=1}^n\DF{1}{k_{\max}^{3}}\sum_{k=k_{ \min}}^{k_{\max}} \e\| [ {\bf x}_j {\bf V}^{\top}_k({\bf x}_j)-\e ({\bf x}_j {\bf V}^{\top}_k({\bf x}_j))]\|^2\sum_{k=k_{ \min}}^{k_{\max}} \| \e({\bf V}_k^{\top}({\bf x}_1){\bf x}_1))\|^2\\
\le&\DF{1}{n^2}\sum_{j=1}^n\DF{1}{k_{\max}^{2}}\sum_{k=k_{ \min}}^{k_{\max}} \e\|{\bf x}_j {\bf V}^{\top}_k({\bf x}_j)\|^2 \leq\DF{1}{n^2}\sum_{j=1}^n\DF{1}{k_{\max}^{2}}\sum_{k=k_{ \min}}^{k_{\max}}k \e\|{\bf x}_j\|^2 =\DF{C}{n},
\end{align*}
so that $I_2=o_P(1)$.

Lastly, we may estimate $I_3$ in a similar way. In fact,
\begin{align*}
\e\|I_3\|^2=&\DF{1}{n^4}\e\left\|\sum_{i=1}^n \sum_{j=1}^n \DF{1}{k_{\max}^{3/2}}\sum_{k=k_{ \min}}^{k_{\max}}\left[ {\bf x}_j {\bf V}^{\top}_k({\bf x}_j)-\e ({\bf x}_j {\bf V}^{\top}_k({\bf x}_j))\right]\left[ {\bf x}_i^{\top}{\bf V}_k({\bf x}_i)-\e( {\bf x}_i^{\top} {\bf V}_k({\bf x}_i))\right]^\top\right\|^2\\
=&\DF{1}{n^4}\e\text{tr}\left[\sum_{i=1}^n \DF{1}{k_{\max}^{3/2}}\sum_{k=k_{ \min}}^{k_{\max}}\left[ {\bf x}_i {\bf V}^{\top}_k({\bf x}_i)-\e ({\bf x}_i {\bf V}^{\top}_k({\bf x}_i))\right]\left[ {\bf x}_i^{\top} {\bf V}_k({\bf x}_i)-\e( {\bf x}_i^{\top} {\bf V}_k({\bf x}_i))\right]^\top\right]^2\\
&+\DF{2}{n^4}\e\text{tr}\sum_{i=1}^n \DF{1}{k_{\max}^{3/2}}\sum_{k=k_{ \min}}^{k_{\max}}\left[ {\bf x}_i {\bf V}^{\top}_k({\bf x}_i)-\e ({\bf x}_j {\bf V}^{\top}_k({\bf x}_j))\right]\left[ {\bf x}_i^{\top} {\bf V}_k({\bf x}_i)-\e( {\bf x}_i^{\top} {\bf V}_k({\bf x}_i))\right]^\top\\
&\qquad \times \sum_{j=1,\ne i}^n\DF{1}{k_{\max}^{3/2}}\sum_{k=k_{ \min}}^{k_{\max}}\left[ {\bf x}_j {\bf V}^{\top}_k({\bf x}_j)-\e ({\bf x}_j {\bf V}^{\top}_k({\bf x}_j))\right]\left[ {\bf x}_j^{\top} {\bf V}_k({\bf x}_j)-\e( {\bf x}_j^{\top} {\bf V}_k({\bf x}_j))\right]^\top\\
&+\DF{2}{n^4}\e\text{tr}\sum_{i=1}^n\sum_{j=1,\ne i}^n \DF{1}{k_{\max}^{3/2}}\sum_{k=k_{ \min}}^{k_{\max}}\left[ {\bf x}_i {\bf V}^{\top}_k({\bf x}_i)-\e ({\bf x}_i {\bf V}^{\top}_k({\bf x}_i))\right]\left[ {\bf x}_j^{\top} {\bf V}_k({\bf x}_j)-\e( {\bf x}_j^{\top} {\bf V}_k({\bf x}_j))\right]^\top\\
&\qquad \times \DF{1}{k_{\max}^{3/2}}\sum_{k=k_{ \min}}^{k_{\max}}\left[ {\bf x}_iV^{\top}_k({\bf x}_i)-\e ({\bf x}_i {\bf V}^{\top}_k({\bf x}_i))\right]\left[ {\bf x}_j^{\top} {\bf V}_k({\bf x}_j)-\e( {\bf x}_j^{\top}{\bf V}_k({\bf x}_j))\right]^\top\\
\le&\DF{1}{n^2} \DF{1}{k_{\max}^{3}}\left[\sum_{k=k_{ \min}}^{k_{\max}}\e\| {\bf x}_1 {\bf V}^{\top}_k({\bf x}_1)-\e ({\bf x}_1 {\bf V}^{\top}_k({\bf x}_1))\|^2\right]^2\\
&+\DF{4}{n^2}\text{tr} \DF{1}{k_{\max}^{3/2}}\sum_{k=k_{ \min}}^{k_{\max}}\e\left[ {\bf x}_iV^{\top}_k({\bf x}_i)-\e ({\bf x}_j {\bf V}^{\top}_k({\bf x}_j))\right]\left[ {\bf x}_i^{\top} {\bf V}_k({\bf x}_i)-\e( {\bf x}_i^{\top} {\bf V}_k({\bf x}_i))\right]^\top\\
&\qquad \times \DF{1}{k_{\max}^{3/2}}\sum_{k=k_{ \min}}^{k_{\max}}\e\left[ {\bf x}_j {\bf V}^{\top}_k({\bf x}_j)-\e ({\bf x}_j {\bf V}^{\top}_k({\bf x}_j))\right]\left[ {\bf x}_j^{\top} {\bf V}_k({\bf x}_j)-\e( {\bf x}_j^{\top}{\bf V}_k({\bf x}_j))\right]^\top\\
\le& \DF{C}{n^2} \DF{1}{k_{\max}^{3}}\left[\sum_{k=k_{ \min}}^{k_{\max}}\e\| {\bf x}_1\|^2\| {\bf V}^{\top}_k({\bf x}_1)\|^2\right]^2\le \DF{C}{n^2} \DF{1}{k_{\max}^{3}}\left[\sum_{k=k_{ \min}}^{k_{\max}}k\e\| {\bf x}_1\|^2\right]^2\\
\le&\DF{C}{n^2} k_{\max}=o(1),
\end{align*}
and consequently, $I_3=o_P(1)$. This finishes the proof of $T_2/\sigma_n=o_P(1)$. To show $T_3/\sigma_n=o_P(1)$, note that
\begin{align*}
\DF{1}{\sigma_n}T_3=&\DF{1}{\sigma_n}\sum_{i=1}^n \sum_{j=1, \neq i}^n\sum_{k=k_{ \min}}^{k_{\max}} {\bf V}^{\top}_k({\bf x}_i){\bf V}_k({\bf x}_j) {\bf x}_j^\top({\bm\beta}-\widehat{\bm \beta})\varepsilon_i\\
=&({\bm\beta}-\widehat{\bm \beta})^\top\DF{1}{\sigma_n}\sum_{i=1}^n \sum_{j=1, \neq i}^n\sum_{k=k_{ \min}}^{k_{\max}} {\bf V}^{\top}_k({\bf x}_i){\bf V}_k({\bf x}_j) {\bf x}_j\varepsilon_i
\end{align*}
and $\sqrt{n}({\bm\beta}-\widehat{\bm \beta})=O_P(1)$, so that it suffices to show
\begin{align*}
\DF{1}{\sigma_n\sqrt{n}}\tilde{T}_3=&\DF{1}{n^{3/2}k_{\max}^{3/2}}\sum_{i=1}^n \sum_{j=1, \neq i}^n\sum_{k=k_{ \min}}^{k_{\max}} {\bf V}^{\top}_k({\bf x}_i){\bf V}_k({\bf x}_j) {\bf x}_j\varepsilon_i\\
=&\DF{1}{n^{3/2}k_{\max}^{3/2}}\sum_{i=1}^n\sum_{k=k_{ \min}}^{k_{\max}} {\bf V}^{\top}_k({\bf x}_i) \left[\sum_{j=1,\ne i}^n {\bf V}_k({\bf x}_j) {\bf x}_j\right]\varepsilon_i=o_P(1).
\end{align*}
In fact, under $H_0$,
\begin{align*}
\e\|\DF{1}{\sigma_n\sqrt{n}}\tilde{T}_3\|^2=&\DF{1}{n^{3}k_{\max}^{3}}
\e\left\|\sum_{i=1}^n\sum_{k=k_{ \min}}^{k_{\max}} {\bf V}^{\top}_k({\bf x}_i) \left[\sum_{j=1,\ne i}^n{\bf V}_k({\bf x}_j) {\bf x}_j\right]\varepsilon_i\right\|^2\\
=&\DF{\sigma_e^2}{n^{3}k_{\max}^{3}}\sum_{i=1}^n\e\left\|\sum_{k=k_{ \min}}^{k_{\max}} {\bf V}^{\top}_k({\bf x}_i) \left[\sum_{j=1,\ne i}^n {\bf V}_k({\bf x}_j) {\bf x}_j\right]\right\|^2\\
=&\DF{\sigma_e^2}{n^{3}k_{\max}^{3}}\sum_{i=1}^n\e\left\|\sum_{j=1,\ne i}^n\sum_{k=k_{ \min}}^{k_{\max}} {\bf V}^{\top}_k({\bf x}_i)  {\bf V}_k({\bf x}_j) {\bf x}_j\right\|^2\\
=&\DF{\sigma_e^2}{n^{3}k_{\max}^{3}}\sum_{i=1}^n\sum_{j=1,\ne i}^n\e\left\|\sum_{k=k_{ \min}}^{k_{\max}} {\bf V}^{\top}_k({\bf x}_i)  {\bf V}_k({\bf x}_j) {\bf x}_j\right\|^2\\
&+\DF{\sigma_e^2}{n^{3}k_{\max}^{3}}\sum_{i=1}^n\sum_{j_1=1,\ne i}^n\sum_{j_2=1,\ne i,\ne j_1}^n\\
&\qquad \e \sum_{k=k_{ \min}}^{k_{\max}}  (\e{\bf x}_{j_1}^\top{\bf V}_k^{\top}({\bf x}_{j_1})) {\bf V}_k({\bf x}_i)   \sum_{k=k_{ \min}}^{k_{\max}} {\bf V}^{\top}_k({\bf x}_i) \e(\mathbf{V}_k({\bf x}_{j_2}) {\bf x}_{j_2})\\
\le&\DF{\sigma_e^2}{n^{3}k_{\max}^{3}}\sum_{i=1}^n\sum_{j=1,\ne i}^n\e\left(\sum_{k=k_{ \min}}^{k_{\max}}|{\bf V}^{\top}_k({\bf x}_i)  {\bf V}_k({\bf x}_j)| \| {\bf x}_j\|\right)^2\\
&+\DF{\sigma_e^2}{n^{3}k_{\max}^{3}}\sum_{i=1}^n\sum_{j_1=1,\ne i}^n\sum_{j_2=1,\ne i,\ne j_1}^n\\
&\qquad \sum_{k_1=k_{ \min}}^{k_{\max}}\sum_{k_2=k_{ \min}}^{k_{\max}}  \e{\bf x}_{j_1}^\top {\bf V}_{k_1}^{\top}({\bf x}_{j_1})\, \e[{\bf V}_{k_1}({\bf x}_i) {\bf V}^{\top}_{k_2}({\bf x}_i)] \e({\bf V}_{k_2}({\bf x}_{j_2}) {\bf x}_{j_2})\\
\le&\DF{c_1}{n}k_{\max}+\DF{c_2}{k_{\max}}=o(1),
\end{align*}
where we use the following in the derivation above, (1) {\footnotesize $\e\left(\sum_{k_1=k_{ \min}}^{k_{\max}}|\mathbf{V}^{\top}_k({\bf x}_i)\mathbf{V}_k({\bf x}_j)|\right)^4\le C k_{\max}^8$} for some absolutely constant $C$ that is the same order as $\e|a_{ij}|^4$ we used before; (2) by the orthogonality of $\psi_j(x)$, $\e[{\bf V}_{k_1}({\bf x}_i) {\bf V}^{\top}_{k_2}({\bf x}_i)]$ is a $k_1\times k_2$ matrix with 1's on the diagonal and all other elements are zero; $\e{\bf x}_{j_1}^\top {\bf V}_{k_1}^{\top}({\bf x}_{j_1})$ and $\e({\bf V}_{k_2}({\bf x}_{j_2}) {\bf x}_{j_2})$ are matrices of coefficients of projecting ${\bf x}_{j_1}$ and ${\bf x}_{j_2}$ onto ${\bf V}_{k_1}^{\top}({\bf x}_{j_1})$ and ${\bf V}_{k_2}^{\top}({\bf x}_{j_2})$, so that their norms are bounded by $\e\|{\bf x}\|$. The proof is finished.
\end{proof}

\begin{proof}[Proof of Theorem 3.3]
Under $H_1$, $m_n(\mathbf{x})=a_n \, \Delta(\mathbf{x})$ with $\|\Delta(\mathbf{x})\|>0$. Our model is written as $y={\bf x}^\top{\bm \beta}+m_n({\bf x})+e$ where $\e[e|{\bf x}]=0$. Given a sample $\{(y_i, {\bf x}_i), i=1,2,\cdots, n\}$ and a truncation parameter $k$, though $a_n\to 0$, the self-projecting IV method still can estimate ${\bm \beta}$ consistently from the sample models, because ${\bf M}_v$ does not depend on $a_n$ at all. Then,
\begin{align*}
\widetilde{e}_i=&y_i-{\bf x}^\top_i\widehat{\bm \beta}={\bf x}^\top_i({\bm \beta}-\widehat{\bm \beta}).
\end{align*}

It follows that
\begin{align*}
 \widehat{\sigma}_e^2 =& \frac{1}{n}\sum_{i=1}^n \widetilde{e}_i^2=\frac{1}{n} \sum_{i=1}^n [{\bf x}^\top_i({\bm \beta}-\widehat{\bm \beta})+\varepsilon_i]^2\\
 =&\frac{1}{n} \sum_{i=1}^n \varepsilon_i^2+\frac{1}{n} \sum_{i=1}^n [{\bf x}^\top_i({\bm \beta}-\widehat{\bm \beta})]^2+\frac{2}{n} \sum_{i=1}^n [{\bf x}^\top_i({\bm \beta}-\widehat{\bm \beta})]\varepsilon_i\\
 =&\frac{1}{n} \sum_{i=1}^n \varepsilon_i^2+({\bm \beta}-\widehat{\bm \beta})^\top\frac{1}{n}{\bf X}^\top{\bf X}({\bm \beta}-\widehat{\bm \beta})+({\bm \beta}-\widehat{\bm \beta})^\top\frac{2}{n} \sum_{i=1}^n {\bf x}_i\varepsilon_i  \to_P \e(\varepsilon^2),
\end{align*}
in view of $\sqrt{n}({\bm \beta}-\widehat{\bm \beta})=O_P(1)$, ${\bf X}^\top{\bf X}/n=\e[{\bf x}{\bf x}^\top]+o_P(1)$, and $\sum_{i=1}^n {\bf x}_i\varepsilon_i/n=\e[{\bf x}\varepsilon]+o_P(1)$ as $n\to\infty$. Hence, in what follows, the $\widehat{\sigma}_e^4$ in $\sigma_n^2$ can be replaced by $[\e(\varepsilon^2)]^2$, while noticing again $\sigma_n^2=O(n^2k_{\max}^3)$.

Observe that
\begin{align*}
T_n =& \sum_{i=1}^n \sum_{j=1, \neq i}^n \sum_{k=k_{ \min}}^{k_{\max}} {\bf V}^{\top}_k({\bf x}_i){\bf V}_k({\bf x}_j) \widetilde{e}_i \, \widetilde{e}_j = 2\sum_{i=2}^n\sum_{j=1}^{i-1}\sum_{k=k_{ \min}}^{k_{\max}} {\bf V}^{\top}_k({\bf x}_i){\bf V}_k({\bf x}_j) \widetilde{e}_i \, \widetilde{e}_j\\
=&2\sum_{i=2}^n\sum_{j=1}^{i-1}\sum_{k=k_{ \min}}^{k_{\max}} {\bf V}^{\top}_k({\bf x}_i){\bf V}_k({\bf x}_j)[{\bf x}_i^\top({\bm\beta}-\widehat{\bm \beta})+m_n({\bf x}_i)+e_i]\\
&\qquad \times[{\bf x}_j^\top({\bm\beta}-\widehat{\bm \beta})+m_n({\bf x}_j)+e_j]\\
=&2\sum_{i=2}^n\sum_{j=1}^{i-1}\sum_{k=k_{ \min}}^{k_{\max}} {\bf V}^{\top}_k({\bf x}_i){\bf V}_k({\bf x}_j)e_ie_j\\
&+2\sum_{i=2}^n\sum_{j=1}^{i-1}\sum_{k=k_{ \min}}^{k_{\max}} {\bf V}^{\top}_k({\bf x}_i){\bf V}_k({\bf x}_j)e_i[{\bf x}_j^\top({\bm\beta}-\widehat{\bm \beta})+m_n({\bf x}_j)]\\
&+2\sum_{i=2}^n\sum_{j=1}^{i-1}\sum_{k=k_{ \min}}^{k_{\max}} {\bf V}^{\top}_k({\bf x}_i){\bf V}_k({\bf x}_j)e_j[{\bf x}_i^\top({\bm\beta}-\widehat{\bm \beta})+m_n({\bf x}_i)]\\
&+2\sum_{i=2}^n\sum_{j=1}^{i-1}\sum_{k=k_{ \min}}^{k_{\max}} {\bf V}^{\top}_k({\bf x}_i){\bf V}_k({\bf x}_j)[{\bf x}_i^\top({\bm\beta}-\widehat{\bm \beta})+m_n({\bf x}_i)]\\
&\qquad \times[{\bf x}_j^\top({\bm\beta}-\widehat{\bm \beta})+m_n({\bf x}_j)] = T_{1n}+T_{2n}+T_{3n}+T_{4n}, \ \ \text{say}.
\end{align*}

Since $e_i$ has the same property as $\varepsilon_i$ under $H_0$, it follows from the proof of Theorem \ref{th4.1} that $T_{1n}/\sigma_n\to_{\cal D}N(0,1)$. To fulfil the proof, we shall show that $T_{2n}/\sigma_n$ and $T_{3n}/\sigma_n$ are $O_P(a_nk_{\max}^{1/2})$ while $T_{4n}/\sigma_n\to_P\infty$. As a matter of fact, we have
\begin{align*}
\DF{1}{\sigma_n}T_{2n}=&\DF{1}{\sigma_n}\sum_{i=2}^n\sum_{j=1}^{i-1}\sum_{k=k_{ \min}}^{k_{\max}} {\bf V}^{\top}_k({\bf x}_i){\bf V}_k({\bf x}_j)e_i[{\bf x}_j^\top({\bm\beta}-\widehat{\bm \beta})+m_n({\bf x}_j)]\\
=&\sqrt{n}({\bm\beta}-\widehat{\bm \beta})^\top\DF{1}{n^{3/2}k_{\max}^{3/2}}\sum_{i=2}^n\sum_{j=1}^{i-1}\sum_{k=k_{ \min}}^{k_{\max}} {\bf V}^{\top}_k({\bf x}_i){\bf V}_k({\bf x}_j){\bf x}_je_i\\
&+\DF{a_n}{nk_{\max}^{3/2}}\sum_{i=2}^n\sum_{j=1}^{i-1}\sum_{k=k_{ \min}}^{k_{\max}} {\bf V}^{\top}_k({\bf x}_i){\bf V}_k({\bf x}_j)\Delta({\bf x}_j) \, e_i,
\end{align*}
and
\begin{align*}
&\e\left\|\DF{1}{n^{3/2}k_{\max}^{3/2}}\sum_{i=2}^n\sum_{j=1}^{i-1}\sum_{k=k_{ \min}}^{k_{\max}} {\bf V}^{\top}_k({\bf x}_i){\bf V}_k({\bf x}_j){\bf x}_je_i\right\|^2\\
=&\DF{\sigma_e^2}{n^{3}k_{\max}^{3}} \sum_{i=2}^n\e\left\|\sum_{j=1}^{i-1}\sum_{k=k_{ \min}}^{k_{\max}} {\bf V}^{\top}_k({\bf x}_i){\bf V}_k({\bf x}_j){\bf x}_j\right\|^2=o(1)
\end{align*}
in view of the same derivation of $\DF{1}{\sigma_n\sqrt{n}}\tilde{T}_3$ in Theorem \ref{th4.1}; moreover,
\begin{align*}
&\e\left|\DF{a_n}{nk_{\max}^{3/2}}\sum_{i=2}^n\sum_{j=1}^{i-1}\sum_{k=k_{ \min}}^{k_{\max}} {\bf V}^{\top}_k({\bf x}_i){\bf V}_k({\bf x}_j)\Delta({\bf x}_j)e_i\right|^2\\
=&\DF{\sigma_e^2a_n^2}{n^{2}k_{\max}^{3}} \sum_{i=2}^n\e\left|\sum_{j=1}^{i-1}\sum_{k=k_{ \min}}^{k_{\max}} {\bf V}^{\top}_k({\bf x}_i){\bf V}_k({\bf x}_j)\Delta({\bf x}_j)\right|^2\\
=&\DF{\sigma_e^2a_n^2}{n^{2}k_{\max}^{3}} \sum_{i=2}^n\e\left|\sum_{k=k_{ \min}}^{k_{\max}} {\bf V}^{\top}_k({\bf x}_i) \sum_{j=1}^{i-1} {\bf V}_k({\bf x}_j)\Delta({\bf x}_j)\right|^2\\
\le&\DF{\sigma_e^2a_n^2}{n^{2}k_{\max}^{3}} \sum_{i=2}^n\sum_{k=k_{ \min}}^{k_{\max}} \e\|{\bf V}^{\top}_k({\bf x}_i)\|^2 \sum_{k=k_{ \min}}^{k_{\max}} \e\left\|\sum_{j=1}^{i-1} {\bf V}_k({\bf x}_j)\Delta({\bf x}_j)\right\|^2\\
=&\DF{Ca_n^2}{n^2k_{\max}}\sum_{i=2}^n\sum_{k=k_{ \min}}^{k_{\max}} \e\left\|\sum_{j=1}^{i-1}[{\bf V}_k({\bf x}_j)m({\bf x}_j)-\e{\bf V}_k({\bf x}_j)m({\bf x}_j)+\e{\bf V}_k({\bf x}_j)\Delta({\bf x}_j)]\right\|^2\\
=&\DF{Ca_n^2}{n^2k_{\max}}\sum_{i=2}^n\sum_{k=k_{ \min}}^{k_{\max}} \left[\sum_{j=1}^{i-1}\e\|{\bf V}_k({\bf x}_j)\Delta({\bf x}_j)-\e{\bf V}_k({\bf x}_j)m({\bf x}_j)\|^2+\|\e{\bf V}_k({\bf x}_j)\Delta({\bf x}_j)\|^2\right]\\
\le&\DF{Ca_n^2}{k_{\max}}\sum_{k=k_{ \min}}^{k_{\max}} \left[\e\|{\bf V}_k({\bf x}_j)\Delta({\bf x}_j)\|^2+\|\e{\bf V}_k({\bf x}_j)\Delta({\bf x}_j)\|^2\right]\\
\le&\DF{Ca_n^2}{k_{\max}}\sum_{k=k_{ \min}}^{k_{\max}}k=Ck_{\max}a_n^2,
\end{align*}
where $C$ is absolute constant that may be different in each appearance. In addition, it is evident that $\DF{1}{\sigma_n}T_{3n}$ has the same order as $\DF{1}{\sigma_n}T_{2n}$.

Furthermore, we have
\begin{align*}
\DF{1}{\sigma_n}T_{4n}=&\DF{2}{\sigma_n}\sum_{i=2}^n\sum_{j=1}^{i-1}\sum_{k=k_{ \min}}^{k_{\max}} {\bf V}^{\top}_k({\bf x}_i){\bf V}_k({\bf x}_j)[{\bf x}_j^\top({\bm\beta}-\widehat{\bm \beta})+m_n({\bf x}_j)][{\bf x}_i^\top({\bm\beta}-\widehat{\bm \beta})+m_n({\bf x}_i)]\\
=&\DF{2}{\sigma_n}\sum_{i=2}^n\sum_{j=1}^{i-1}\sum_{k=k_{ \min}}^{k_{\max}} {\bf V}^{\top}_k({\bf x}_i){\bf V}_k({\bf x}_j)m_n({\bf x}_j)m_n({\bf x}_i)\\
&+\DF{2}{\sigma_n}\sum_{i=2}^n\sum_{j=1}^{i-1}\sum_{k=k_{ \min}}^{k_{\max}} {\bf V}^{\top}_k({\bf x}_i){\bf V}_k({\bf x}_j)[{\bf x}_j^\top({\bm\beta}-\widehat{\bm \beta})][{\bf x}_i^\top({\bm\beta}-\widehat{\bm \beta})]\\
&+\DF{2}{\sigma_n}\sum_{i=2}^n\sum_{j=1}^{i-1}\sum_{k=k_{ \min}}^{k_{\max}} {\bf V}^{\top}_k({\bf x}_i){\bf V}_k({\bf x}_j)[{\bf x}_j^\top({\bm\beta}-\widehat{\bm \beta})]m_n({\bf x}_i)\\
&+\DF{2}{\sigma_n}\sum_{i=2}^n\sum_{j=1}^{i-1}\sum_{k=k_{ \min}}^{k_{\max}} {\bf V}^{\top}_k({\bf x}_i){\bf V}_k({\bf x}_j)m_n({\bf x}_j)[{\bf x}_i^\top({\bm\beta}-\widehat{\bm \beta})]\\
=&\DF{1}{\sigma_n}\sum_{k=k_{ \min}}^{k_{\max}}\left\|\sum_{i=1}^n {\bf V}^{\top}_k({\bf x}_i) m_n({\bf x}_i)\right\|^2 -\DF{1}{\sigma_n}\sum_{k=k_{\min}}^{k_{\max}}\sum_{i=1}^n \|V^{\top}_k({\bf x}_i) m_n({\bf x}_i)\|^2\\
&+\sqrt{n}({\bm\beta}-\widehat{\bm \beta})^\top \left[\DF{1}{n\sigma_n}\sum_{k=k_{ \min}}^{k_{\max}}\left(\sum_{i=1}^n {\bf x}_i {\bf V}^{\top}_k({\bf x}_i)\right) \left(\sum_{i=1}^n {\bf x}_i {\bf V}^{\top}_k({\bf x}_i)\right)^\top\right]\sqrt{n}({\bm\beta}-\widehat{\bm \beta})\\
&-\sqrt{n}({\bm\beta}-\widehat{\bm \beta})^\top \left[\DF{1}{n\sigma_n}\sum_{k=k_{ \min}}^{k_{\max}}\sum_{i=1}^n {\bf x}_i {\bf V}^{\top}_k({\bf x}_i) {\bf V}_k({\bf x}_i){\bf x}_i^\top\right]\sqrt{n}({\bm\beta}-\widehat{\bm \beta})\\
&+\DF{2}{\sigma_n}\sum_{i=2}^n\sum_{j=1}^{i-1}\sum_{k=k_{ \min}}^{k_{\max}} {\bf V}^{\top}_k({\bf x}_i){\bf V}_k({\bf x}_j)[{\bf x}_j^\top({\bm\beta}-\widehat{\bm \beta})]m_n({\bf x}_i)\\
&+\DF{2}{\sigma_n}\sum_{i=2}^n\sum_{j=1}^{i-1}\sum_{k=k_{ \min}}^{k_{\max}} {\bf V}^{\top}_k({\bf x}_i){\bf V}_k({\bf x}_j)m_n({\bf x}_j)[{\bf x}_i^\top({\bm\beta}-\widehat{\bm \beta})]\\
=&J_1+\cdots+J_6, \ \ \text{say}.
\end{align*}

Observe that
\begin{align*}
J_1=&\DF{a_n^2}{\sigma_n}\sum_{k=k_{ \min}}^{k_{\max}}\left\|\sum_{i=1}^n {\bf V}^{\top}_k({\bf x}_i) \Delta({\bf x}_i)\right\|^2= \DF{a_n^2n}{k_{\max}^{3/2}}\sum_{k=k_{ \min}}^{k_{\max}}\left\|\DF{1}{n}\sum_{i=1}^n {\bf V}^{\top}_k({\bf x}_i)\Delta({\bf x}_i)\right\|^2\\
=&\DF{a_n^2n}{k_{\max}^{3/2}}\sum_{k=k_{ \min}}^{k_{\max}}\left\|\e {\bf V}^{\top}_k({\bf x}_1)\Delta({\bf x}_1)+\DF{1}{n}\sum_{i=1}^n [V^{\top}_k({\bf x}_i)\Delta({\bf x}_i)-\e {\bf V}^{\top}_k({\bf x}_i)\Delta({\bf x}_i)]\right\|^2\\
=&\DF{a_n^2n}{k_{\max}^{3/2}}\sum_{k=k_{ \min}}^{k_{\max}}\|\e V^{\top}_k({\bf x}_1)\Delta({\bf x}_1)\|^2+
\DF{a_n^2n}{k_{\max}^{3/2}}\sum_{k=k_{ \min}}^{k_{\max}}\left\|\DF{1}{n}\sum_{i=1}^n [{\bf V}^{\top}_k({\bf x}_i)\Delta({\bf x}_i)-\e {\bf V}^{\top}_k({\bf x}_i)\Delta({\bf x}_i)]\right\|^2\\
&+2\DF{a_n^2n}{k_{\max}^{3/2}}\sum_{k=k_{ \min}}^{k_{\max}}[\e{\bf V}_k({\bf x}_1) \Delta({\bf x}_1)]^\top \left(\DF{1}{n}\sum_{i=1}^n [{\bf V}_k({\bf x}_i)\Delta({\bf x}_i)-\e{\bf V}_k({\bf x}_i)\Delta({\bf x}_i)]\right)\\
=&\DF{a_n^2n}{k_{\max}^{3/2}}\sum_{k=k_{ \min}}^{k_{\max}}\|\widetilde{\bm \gamma}\|^2+
\DF{a_n^2n}{k_{\max}^{3/2}}\sum_{k=k_{ \min}}^{k_{\max}}\left\|\DF{1}{n}\sum_{i=1}^n [{\bf V}^{\top}_k({\bf x}_i)\Delta({\bf x}_i)-\e {\bf V}^{\top}_k({\bf x}_i)\Delta({\bf x}_i)]\right\|^2\\
&+2\DF{a_n^2n}{k_{\max}^{3/2}}\sum_{k=k_{ \min}}^{k_{\max}}[\e{\bf V}_k({\bf x}_1) \Delta({\bf x}_1)]^\top \left(\DF{1}{n}\sum_{i=1}^n [{\bf V}_k({\bf x}_i)\Delta({\bf x}_i)-\e{\bf V}_k({\bf x}_i) m({\bf x}_i)]\right)\\
\ge&\DF{a_n^2n}{k_{\max}^{3/2}}\sum_{k=k_{ \min}}^{k_{\max}}[\e \Delta^2({\bf x})-\sum_{j=k}^\infty \widetilde{\gamma}_j^2]
+\DF{a_n^2n}{k_{\max}^{3/2}} \cdot O_P\left(\sum_{k=k_{ \min}}^{k_{\max}}\DF{k}{n}\right)\\
&-2\DF{a_n^2n}{k_{\max}^{3/2}}\sum_{k=k_{ \min}}^{k_{\max}}\|\e{\bf V}_k({\bf x}_1) \Delta({\bf x}_1)\| \left\|\DF{1}{n}\sum_{i=1}^n [{\bf V}_k({\bf x}_i)\Delta({\bf x}_i)-\e{\bf V}_k({\bf x}_i)\Delta({\bf x}_i)]\right\|\\
\ge&\|\Delta( x)\|^2 \DF{a_n^2n}{k_{\max}^{1/2}}-\DF{a_n^2n}{k_{\max}^{1/2}}\sum_{j=k_{\min}}^\infty \widetilde{\gamma}_j^2+O_P(a_n^2k_{\max}^{1/2})-2\|\Delta( x)\| \DF{a_n^2n}{k_{\max}^{3/2}} \sum_{k=k_{ \min}}^{k_{\max}}O_P\left(\DF{\sqrt{k}}{\sqrt{n}}\right)\\
=&\|\Delta( x)\|^2 \DF{a_n^2n}{k_{\max}^{1/2}}-O_P\left(\DF{a_n^2n}{k_{\max}^{1/2}}\sum_{j=k_{\min}}^\infty \widetilde{\gamma}_j^2\right)+O_P(a_n^2k_{\max}^{1/2})-O_P(a_n^2\sqrt{n})\\
=&\DF{a_n^2n}{k_{\max}^{1/2}}(1+o_P(1))
\end{align*}
where $\widetilde{\gamma}_j$ are coefficients in the expansion of $\Delta({\bf x})$, $\widetilde{\bm \gamma}=(\widetilde\gamma_1, \cdots, \widetilde\gamma_k)^\top$ and by Parseval equality $\e[\Delta^2({\bf x})]=\sum_{j=1}^\infty \widetilde{\gamma}_j^2$.

In addition, the order of $\sum_k\left\|\DF{1}{n}\sum_{i=1}^n [{\bf V}^{\top}_k({\bf x}_i) m({\bf x}_i)-\e {\bf V}^{\top}_k({\bf x}_i)\Delta({\bf x}_i)]\right\|^2=O_P(\sum_kk/n)$ is derived straightforwardly so that we omit it.

Next, we analyze $J_2, \cdots, J_6$ term by term. For $J_2$,
\begin{align*}
 J_2=&\DF{1}{\sigma_n}\sum_{k=k_{\min}}^{k_{\max}}\sum_{i=1}^n \|{\bf V}^{\top}_k({\bf x}_i) m_n({\bf x}_i)\|^2
 =\DF{a_n^2}{k_{\max}^{3/2}}\sum_{k=k_{\min}}^{k_{\max}}\DF{1}{n}\sum_{i=1}^n \|{\bf V}^{\top}_k({\bf x}_i)\Delta({\bf x}_i)\|^2\\
 =&\DF{a_n^2}{k_{\max}^{3/2}}\sum_{k=k_{\min}}^{k_{\max}}\e\|{\bf V}^{\top}_k({\bf x}) \Delta({\bf x})\|^2 +\DF{a_n^2}{k_{\max}^{3/2}}\sum_{k=k_{\min}}^{k_{\max}}\left(\DF{1}{n}\sum_{i=1}^n \|{\bf V}^{\top}_k({\bf x}_i) \Delta({\bf x}_i)\|^2-\e\|{\bf V}^{\top}_k({\bf x}) \Delta({\bf x})\|^2\right)\\
 =&O(a_n^2k_{\max}^{1/2})+\DF{a_n^2}{k_{\max}^{3/2}}O_P\left(\sum_{k=k_{\min}}^{k_{\max}}  k/\sqrt{n}\right)=O(a_n^2k_{\max}^{1/2})+o_P(1).
\end{align*}

For $J_3$, consider
\begin{align*}
&\DF{1}{n\sigma_n}\sum_{k=k_{ \min}}^{k_{\max}}\left(\sum_{i=1}^n {\bf x}_i {\bf V}^{\top}_k({\bf x}_i)\right) \left(\sum_{i=1}^n {\bf x}_i {\bf V}^{\top}_k({\bf x}_i)\right)^\top\\
=&\DF{1}{k_{\max}^{3/2}}\sum_{k=k_{ \min}}^{k_{\max}}\left(\DF{1}{n}\sum_{i=1}^n {\bf x}_i {\bf V}^{\top}_k({\bf x}_i)\right) \left(\DF{1}{n}\sum_{i=1}^n {\bf x}_i {\bf V}^{\top}_k({\bf x}_i)\right)^\top\\
=&\DF{1}{k_{\max}^{3/2}}\sum_{k=k_{ \min}}^{k_{\max}}(\e {\bf x} {\bf V}^{\top}_k({\bf x})) (\e {\bf x} {\bf V}^{\top}_k({\bf x}))^\top\\
&+\DF{1}{k_{\max}^{3/2}}\sum_{k=k_{ \min}}^{k_{\max}}\left(\DF{1}{n}\sum_{i=1}^n {\bf x}_i {\bf V}^{\top}_k({\bf x}_i)-\e {\bf x} {\bf V}^{\top}_k({\bf x})\right) \left(\DF{1}{n}\sum_{i=1}^n {\bf x}_i {\bf V}^{\top}_k({\bf x}_i)-\e {\bf x} {\bf V}^{\top}_k({\bf x})\right)^\top\\
&+\DF{1}{k_{\max}^{3/2}}\sum_{k=k_{ \min}}^{k_{\max}}\e {\bf x} {\bf V}^{\top}_k({\bf x}) \left(\DF{1}{n}\sum_{i=1}^n {\bf x}_i {\bf V}^{\top}_k({\bf x}_i)-\e {\bf x} {\bf V}^{\top}_k({\bf x})\right)^\top\\
&+\DF{1}{k_{\max}^{3/2}}\sum_{k=k_{ \min}}^{k_{\max}}\left(\DF{1}{n}\sum_{i=1}^n {\bf x}_i {\bf V}^{\top}_k({\bf x}_i)-\e {\bf x} {\bf V}^{\top}_k({\bf x})\right) (\e {\bf x} {\bf V}^{\top}_k({\bf x}))^\top\\
\le&\e\|{\bf x}\|^2\DF{1}{k_{\max}^{1/2}}+\DF{1}{k_{\max}^{3/2}}\sum_{k=k_{ \min}}^{k_{\max}}O_P(k/n)
+\DF{2}{k_{\max}^{3/2}}\sum_{k=k_{ \min}}^{k_{\max}}O_P(\sqrt{k/n})=o_P(1),
\end{align*}
where $\|(\e {\bf x} {\bf V}^{\top}_k({\bf x})) (\e {\bf x} {\bf V}^{\top}_k({\bf x}))^\top\|\le \e\|{\bf x}\|^2$ since $\e {\bf x} {\bf V}^{\top}_k({\bf x})$ are coefficients of projecting ${\bf x}$ onto ${\bf V}^{\top}_k({\bf x})$. Thus, $J_3=o_P(1)$. This implies $J_4=o_P(1)$ as well.

For $J_5$, observe
\begin{align*}
&\DF{2}{\sqrt{n}\sigma_n}\sum_{i=2}^n\sum_{j=1}^{i-1}\sum_{k=k_{ \min}}^{k_{\max}} {\bf x}_j {\bf V}^{\top}_k({\bf x}_j){\bf V}_k({\bf x}_i)m_n({\bf x}_i)\\
=&\DF{a_n\sqrt{n}}{k_{\max}^{3/2}}\sum_{k=k_{ \min}}^{k_{\max}}\left(\DF{1}{n}\sum_{j=1}^{n} {\bf x}_j {\bf V}^{\top}_k({\bf x}_j)\right)\left( \DF{1}{n}\sum_{i=1}^n {\bf V}_k({\bf x}_i)\Delta({\bf x}_i)\right)\\
=&\DF{a_n\sqrt{n}}{k_{\max}^{3/2}}\sum_{k=k_{ \min}}^{k_{\max}} (\e {\bf x} {\bf V}^{\top}_k({\bf x}))(\e{\bf V}_k({\bf x})\Delta({\bf x}))\\
&+\DF{a_n\sqrt{n}}{k_{\max}^{3/2}}\sum_{k=k_{ \min}}^{k_{\max}}(\e {\bf x} {\bf V}^{\top}_k({\bf x}))\left( \DF{1}{n}\sum_{i=1}^nV_k({\bf x}_i)\Delta({\bf x}_i)-\e{\bf V}_k({\bf x})\Delta({\bf x})\right)\\
&+\DF{a_n\sqrt{n}}{k_{\max}^{3/2}}\sum_{k=k_{ \min}}^{k_{\max}}\left(\DF{1}{n}\sum_{j=1}^{n} {\bf x}_j {\bf V}^{\top}_k({\bf x}_j)-\e {\bf x} {\bf V}^{\top}_k({\bf x})\right)\e{\bf V}_k({\bf x})\Delta({\bf x})\\
&+\DF{a_n\sqrt{n}}{k_{\max}^{3/2}}\sum_{k=k_{ \min}}^{k_{\max}}\left(\DF{1}{n}\sum_{j=1}^{n} {\bf x}_j {\bf V}^{\top}_k({\bf x}_j)-\e {\bf x} {\bf V}^{\top}_k({\bf x})\right)\left( \DF{1}{n}\sum_{i=1}^n\left[\mathbf{V}_k({\bf x}_i)\Delta({\bf x}_i)-\e{\bf V}_k({\bf x})\Delta({\bf x})\right]\right)\\
=&O\left(\DF{a_n\sqrt{n}}{k_{\max}^{1/2}}\right)+2\DF{a_n\sqrt{n}}{k_{\max}^{3/2}}\sum_{k=k_{ \min}}^{k_{\max}} O_P(\sqrt{k/n})+\DF{a_n\sqrt{n}}{k_{\max}^{3/2}}\sum_{k=k_{ \min}}^{k_{\max}} O_P(k/n) = O_P(a_n\sqrt{n/k_{\max}}),
\end{align*}
where once again $\|\e {\bf x} {\bf V}^{\top}_k({\bf x})\|\le \sqrt{\e\|{\bf x}\|^2}$ and $\|\e \Delta({\bf x}) {\bf V}^{\top}_k({\bf x})\|\le \sqrt{\e[\Delta^2({\bf x})]}$. 

Thus, $J_5=O_P(a_n\sqrt{n/k_{\max}})$, and $J_6=O_P(a_n\sqrt{n/k_{\max}})$ due to their similarity.

To conclude, $\DF{1}{\sigma_n}T_{4n}\ge a_n^2n/k_{\max}^{1/2}(1+o_P(1))\to_P\infty $, and hence, the theorem is fulfilled.

\end{proof}

\begin{proof}[Proof of Theorem 4.1]

(i)-(ii). Before proceeding, we justify some basic facts. Let $\pmb{\Omega}_k =\begin{pmatrix}
    \mathbf{A}_{11} & \mathbf{A}_{12} \\
    \mathbf{A}_{12}^\top & \mathbf{I}_k \\
\end{pmatrix}$, where $\mathbf{A}_{11}=\e[\mathbf{x}\mathbf{x}^\top]$, $\mathbf{A}_{12}\coloneqq (\pmb{\eta}_1,\cdots , \pmb{\eta}_{k})$, and $\pmb{\eta}_\ell=\e[\mathbf{x}\phi_\ell(\mathbf{x})]$. To show the invertibility of $\pmb{\Omega}_k$, we note that $ \pmb{\Omega}_k =\begin{pmatrix}
    \mathbf{A}_{11} &\mathbf{0} \\
    \mathbf{0} & \mathbf{I}_k \\
\end{pmatrix}+\begin{pmatrix}
    \mathbf{0} & \mathbf{A}_{12} \\
    \mathbf{A}_{12}^\top & \mathbf{0} \\
\end{pmatrix}$, so by Weyl's inequality,
{\footnotesize
\bea
\rho_{\min}(\pmb{\Omega}_k) & \ge & \rho_{\min}(\mathbf{A}_{11})\wedge 1 +\rho_{\min}^{1/2}\left(\begin{pmatrix}
    \mathbf{0} & \mathbf{A}_{12} \\
    \mathbf{A}_{12}^\top & \mathbf{0} \\
\end{pmatrix}\begin{pmatrix}
    \mathbf{0} & \mathbf{A}_{12} \\
    \mathbf{A}_{12}^\top & \mathbf{0} \\
\end{pmatrix}\right)\notag \\
&=&\rho_{\min}(\mathbf{A}_{11})\wedge 1 +\rho_{\min}^{1/2}\left(\begin{pmatrix}
    \mathbf{A}_{12}\mathbf{A}_{12}^\top & \mathbf{0}  \\
    \mathbf{0}  & \mathbf{A}_{12}^\top\mathbf{A}_{12}  \\
\end{pmatrix} \right) >0,
\nonumber
\eea}

\noindent where $\rho_{\min}(\cdot)$ stands for the minimum eigenvalue value. 

Therefore, for $\forall(\boldsymbol{\beta}, \boldsymbol{\xi}_k)\in \mathbb{A}$ and $\VEC(\boldsymbol{\beta}, \boldsymbol{\xi}_k)\ne \mathbf{0}$, we have
$\VEC(\boldsymbol{\beta}, \boldsymbol{\xi}_k)^\top \pmb{\Omega}_k\VEC(\boldsymbol{\beta}, \boldsymbol{\xi}_k)\ge \alpha \|\VEC(\boldsymbol{\beta}, \boldsymbol{\xi}_k)\|^2$
uniformly in $k$, where $\mathbb{A}  \coloneqq \{ \VEC(\boldsymbol{\beta}, \boldsymbol{\xi}_k)\mid \|\VEC(\boldsymbol{\beta}, \mathbf{S}_{\bar{\mathscr{C}}}  \boldsymbol{\xi}_k ) \|_1\le 3 \|\VEC(\boldsymbol{\beta}, \mathbf{S}_{\mathscr{C}}  \boldsymbol{\xi}_k )\|_1\}$, and $\alpha>0$ is a fixed positive constant, in which $\mathbf{S}_{\mathscr{C}}$ and $\mathbf{S}_{\bar{\mathscr{C}}} $ are respectively $k_0\times k$ and $(k-k_0)\times k$ selection matrices selecting the elements indexed by $\mathscr{C}$ and $\bar{\mathscr{C}}$. 

In this case the so-called restricted eigenvalue condition is automatically fulfilled, which has been fully discussed in the literature, such as \cite[pp. 1709-1710]{BRT2009} and \cite[p. 2245]{raskutti10a}. 

To avoid introducing any new notation, in a similar way to the notation introduced in Section 3, we consider a vector model of the form:  $\mathbf{y} = \mathbf{X}\boldsymbol{\beta}_0 + \mathbf{V}\boldsymbol{\xi} +\boldsymbol{\delta} +\mathbf{e}$.

We then define the following function: $Q(\boldsymbol{\beta}, \boldsymbol{\xi}_k) =\frac{1}{n}\|\mathbf{y}- \mathbf{X}\boldsymbol{\beta} - \mathbf{V}\boldsymbol{\xi}_k\|^2 + \lambda \|\boldsymbol{\xi}_k\|_1$ to conduct the LASSO regression, where $\|\cdot \|_1$ stands for $L_1$ norm,  $\boldsymbol{\beta}$ and $\boldsymbol{\xi}_k$ are generic vectors with dimensions being the same as $\boldsymbol{\beta}_0$ and $\boldsymbol{\xi}$ respectively, in which $\boldsymbol{\xi}$ collects all $\xi_\ell$'s stated in Section 4. The LASSO estimates of $\boldsymbol{\beta}_0$ and $\boldsymbol{\xi}$ are given by $(\widehat{\boldsymbol{\beta}}^*, \widehat{\boldsymbol{\xi}}^*)  =\argmin Q(\boldsymbol{\beta}, \boldsymbol{\xi}_k)$.

By definition, we have $Q(\widehat{\boldsymbol{\beta}}^*, \widehat{\boldsymbol{\xi}}^*)\le  Q(\boldsymbol{\beta}_0, \boldsymbol{\xi}) = \frac{1}{n}\|\mathbf{e} + \boldsymbol{\delta}\| + \lambda \| \boldsymbol{\xi} \|_1$. Furthermore, we obtain
{\small
\begin{eqnarray*}
\frac{1}{n}\left| (\mathbf{e} + \boldsymbol{\delta})^\top [\mathbf{X}, \mathbf{V}] \VEC(\widehat{\boldsymbol{\beta}}^*- \boldsymbol{\beta}_0, \widehat{\boldsymbol{\xi}}^* - \boldsymbol{\xi} )\right| \le \frac{1}{n} \| (\mathbf{e} + \boldsymbol{\delta})^\top [\mathbf{X}, \mathbf{V}] \|_\infty \|\VEC(\widehat{\boldsymbol{\beta}}^* - \boldsymbol{\beta}_0,\widehat{\boldsymbol{\xi}}^* - \boldsymbol{\xi}) \|_1,
\end{eqnarray*}}

\noindent where $\|\cdot \|_\infty$ stands for the maximum norm of a vector.  We focus on $\frac{1}{n} \| (\mathbf{e} + \boldsymbol{\delta})^\top [\mathbf{X}, \mathbf{V}] \|_\infty$ in what follows. Note that $\boldsymbol{\phi}_\ell \coloneqq (\phi_\ell(\mathbf{x}_1),\ldots, \phi_\ell(\mathbf{x}_n))^\top$ is the $\ell^{th}$ column of $\mathbf{V}$ and then derive

\begin{eqnarray*}
    \max_{\ell}\frac{1}{n}|\boldsymbol{\delta}^\top \boldsymbol{\phi}_\ell|\le \max_{\ell}\frac{1}{n}\sum_{i=1}^n |\delta_k(\mathbf{x}_i) \phi_\ell(\mathbf{x}_i)| \le O_P(1)\left\{\frac{1}{n}\mathbb{E}|\delta_k(\mathbf{x}) |^2\right\}^{1/2} ,
\end{eqnarray*}
where the last step has been proved in Theorem 3.1 by noting that $\mathbb{E}[\delta_k(\mathbf{x}) \phi_\ell(\mathbf{x})]=0$. Note also that the proof of Theorem 3.1 has also shown that $ \frac{1}{n}\|\boldsymbol{\delta}^\top \mathbf{X}\| =O_P(\sum_{j=k+1}^\infty\pmb{\lambda}_j\gamma_j)$. We thus conclude 

\begin{eqnarray}\label{rate.delta}
    \frac{1}{n} \|\boldsymbol{\delta}^\top [\mathbf{X}, \mathbf{V}]\|_\infty= O_P\left(\left\{\frac{1}{n}\mathbb{E}|\delta_k(\mathbf{x}) |^2\right\}^{1/2}\right) + O_P\left(\left\|\sum_{j=k+1}^\infty\pmb{\lambda}_j\gamma_j )\right\|\right).
\end{eqnarray}

We consider $\frac{1}{n} \mathbf{e}^\top \boldsymbol{\phi}_\ell  =\frac{1}{n}\sum_{i=1}^n \phi_\ell (\mathbf{x}_i) \, e_i$. In order to do so, we let $\epsilon = c_0\sqrt{n\log k}$ and write
\begin{eqnarray*}
&&{\rm P}\left(\sum_{i=1}^n \phi_\ell (\mathbf{x}_i)e_i\ge \epsilon\right) \le c_1 \, \sum_{i=1}^n\|e_i\|_J^J \ \epsilon^{-J} + 2\exp\left(- c_2\epsilon^2 \ \left(\sum_{i=1}^n\| e_i\|_2^2\right)^{-1}\right) \notag \\
&=&  c_1 \frac{n}{\epsilon^J} + 2\exp\left(  \frac{- c_2\epsilon^2}{n}\right) =  \frac{c_1}{n^{J/2-1}(\log k)^{J/2}}+2\, e^{- c_2\log k},
\end{eqnarray*}
where the first inequality follows from Corollary 1.8 in \cite{nagaev1979large} by letting $J=4$ in Assumption 4.1. Additionally, we require $c_2>1$ which is achievable by varying the value of $c_0$. 

Similarly, we show that ${\Pr}\left(\sum_{i=1}^n x_{i,j} e_i\ge \epsilon\right) = \frac{c_1}{n^{J/2-1}(\log k)^{J/2}}+2\exp\left( - c_2\log k\right)$, where $x_{i,j}$ stands for the $j^{th}$ element of $\mathbf{x}_i$. Thus, under the condition $\frac{k}{n^{J/2-1}(\log k)^{J/2}}\to 0$, we have
\begin{eqnarray}\label{rate.e}
    \frac{1}{n} \|  \mathbf{e}^\top [\mathbf{X}, \mathbf{V}] \|_\infty = O_P\left(\frac{\sqrt{\log k}}{\sqrt{n}}\right).
\end{eqnarray}

By \eqref{rate.delta} and \eqref{rate.e}, and using Assumption 3.1(iii), we obtain
\bea
 &&   \frac{1}{n} \| (\mathbf{e} + \boldsymbol{\delta})^\top [\mathbf{X}, \mathbf{V}] \|_\infty=O_P\left(\frac{\sqrt{\log k}}{\sqrt{n}} \vee \left(\left\{\frac{1}{n}\mathbb{E}|\delta_k(\mathbf{x}) |^2\right\}^{1/2} + \left\|\sum_{j=k+1}^\infty\pmb{\lambda}_j\gamma_j )\right\|\right)\right)
 \nonumber\\
 && =O_P\left(\frac{\sqrt{\log k}}{\sqrt{n}} \right).
\label{rate.ed}
\eea

Then with probability approaching 1, we have $\frac{1}{n} \| (\mathbf{e} + \boldsymbol{\delta})^\top [\mathbf{X}, \mathbf{V}] \|_\infty\le \frac{\lambda}{4}$, where $\lambda = c^* \frac{\sqrt{\log k}}{\sqrt{n}} $ with $c^*$ being a sufficiently large number. 

Then, we have $\frac{1}{n}\left| (\mathbf{e} + \boldsymbol{\delta})^\top [\mathbf{X}, \mathbf{V}] \VEC(\boldsymbol{\beta}_0 - \widehat{\boldsymbol{\beta}}^*, \boldsymbol{\xi} -\widehat{\boldsymbol{\xi}}^* )\right| \le \frac{\lambda}{4} \|\VEC(\boldsymbol{\beta}_0 - \widehat{\boldsymbol{\beta}}^*, \boldsymbol{\xi} -\widehat{\boldsymbol{\xi}}^* )\|_1$. Therefore, we have
{\small
\begin{eqnarray}\label{Ineq.2}
\frac{1}{n} \| [\mathbf{X}, \mathbf{V}_k] \VEC(\boldsymbol{\beta}_0 - \widehat{\boldsymbol{\beta}}^*, \boldsymbol{\xi} -\widehat{\boldsymbol{\xi}}^* ) \|^2 + \lambda\|\widehat{\boldsymbol{\xi}}^* \|_1 \le \frac{\lambda}{2} \|\VEC(\boldsymbol{\beta}_0 - \widehat{\boldsymbol{\beta}}^*, \boldsymbol{\xi} -\widehat{\boldsymbol{\xi}}^* )\|_1 + \lambda \| \widehat{\boldsymbol{\xi}}^* \|_1.
\end{eqnarray}}

Recall the definitions of $\mathbf{S}_{\mathscr{C}}$ and $\mathbf{S}_{\bar{\mathscr{C}}}$ in Assumption 4.1. Simple algebra shows that 
{\small
\begin{eqnarray*}
    &&\|\VEC(\boldsymbol{\beta}_0 - \widehat{\boldsymbol{\beta}}^*, \boldsymbol{\xi} -\widehat{\boldsymbol{\xi}}^* )\|_1 +\|\VEC( \boldsymbol{\beta}_0,  \boldsymbol{\xi} )\|_1 -\|\VEC( \widehat{\boldsymbol{\beta}}^*, \widehat{\boldsymbol{\xi}}^* )\|_1 \le 2\|\VEC(\boldsymbol{\beta}_0 - \widehat{\boldsymbol{\beta}}^*, \mathbf{S}_{\mathscr{C}} ( \boldsymbol{\xi} -\widehat{\boldsymbol{\xi}}^*) )\|_1,
\end{eqnarray*}}

\noindent so we can rewrite \eqref{Ineq.2} as $\frac{1}{n} \| [\mathbf{X}, \mathbf{V}] \VEC(\boldsymbol{\beta}_0 - \widehat{\boldsymbol{\beta}}^*, \boldsymbol{\xi} -\widehat{\boldsymbol{\xi}}^* ) \|^2 +\frac{\lambda}{2} \|\VEC(\boldsymbol{\beta}_0 - \widehat{\boldsymbol{\beta}}^*, \boldsymbol{\xi} -\widehat{\boldsymbol{\xi}}^* )\|_1 \leq 2\lambda\|\VEC(\boldsymbol{\beta}_0 - \widehat{\boldsymbol{\beta}}^*, \mathbf{S}_{\mathscr{C}} ( \boldsymbol{\xi} -\widehat{\boldsymbol{\xi}}^*) )\|_1$. It further yields 
\begin{eqnarray}\label{eq.lasso2_1}
&& \frac{1}{n} \| [\mathbf{X}, \mathbf{V}] \VEC(\boldsymbol{\beta}_0 - \widehat{\boldsymbol{\beta}}^*, \boldsymbol{\xi} -\widehat{\boldsymbol{\xi}}^* ) \|^2 +\frac{\lambda}{2} \|\VEC(\boldsymbol{\beta}_0 - \widehat{\boldsymbol{\beta}}^*, \mathbf{S}_{\bar{\mathscr{C}}} ( \boldsymbol{\xi} -\widehat{\boldsymbol{\xi}}^*) )\|_1
\nonumber\\
&& \le \frac{3}{2}\lambda\|\VEC(\boldsymbol{\beta}_0 - \widehat{\boldsymbol{\beta}}^*, \mathbf{S}_{\mathscr{C}} ( \boldsymbol{\xi} -\widehat{\boldsymbol{\xi}}^*) )\|_1.
\end{eqnarray}

Finally, we focus on $\frac{1}{n} \| [\mathbf{X}, \mathbf{V}] \VEC(\boldsymbol{\beta}_0 - \widehat{\boldsymbol{\beta}}^*, \boldsymbol{\xi} -\widehat{\boldsymbol{\xi}}^* ) \|^2$ of \eqref{eq.lasso2_1}. As explained in the beginning of the proof, we have for $\forall (\boldsymbol{\beta}, \boldsymbol{\xi}_k)\in \mathbb{A}$
\bea
 &&\|\VEC(\boldsymbol{\beta}, \boldsymbol{\xi}_k ) \|_1 = \|\VEC(\boldsymbol{\beta}, \mathbf{S}_{\mathscr{C}} \boldsymbol{\xi}_k)  \|_1 +\|\VEC(\boldsymbol{\beta}, \mathbf{S}_{\bar{\mathscr{C}}}\boldsymbol{\xi}_k)  \|_1\le 4 \|\VEC(\boldsymbol{\beta}, \mathbf{S}_{\mathscr{C}} \boldsymbol{\xi}_k)  \|_1
    \nonumber\\
    && \le 4\sqrt{k_0}\| \VEC(\boldsymbol{\beta}, \mathbf{S}_{\mathscr{C}} \boldsymbol{\xi}_k) \|.
    \label{eq.lasso3}
    \eea

Thus, with probability approaching 1,
{\footnotesize
\bea
    &&\frac{1}{n} \| [\mathbf{X}, \mathbf{V}] \VEC(\boldsymbol{\beta}_0 - \widehat{\boldsymbol{\beta}}^*, \boldsymbol{\xi} -\widehat{\boldsymbol{\xi}}^* ) \|^2  \ge \VEC(\boldsymbol{\beta}_0 - \widehat{\boldsymbol{\beta}}^*, \boldsymbol{\xi} -\widehat{\boldsymbol{\xi}}^* )^\top  \pmb{\Omega}_k \VEC(\boldsymbol{\beta}_0 - \widehat{\boldsymbol{\beta}}^*, \boldsymbol{\xi} -\widehat{\boldsymbol{\xi}}^* ) \nonumber \\
    &&- \frac{1}{n} \left\| [\mathbf{X}, \mathbf{V} ]^\top [\mathbf{X}, \mathbf{V} ]  - \pmb{\Omega}_k\right\|_{\max} \|\VEC(\boldsymbol{\beta}_0 - \widehat{\boldsymbol{\beta}}^*, \boldsymbol{\xi} -\widehat{\boldsymbol{\xi}}^* )\|_1^2 
    \nonumber\\
    && \geq \VEC(\boldsymbol{\beta}_0 - \widehat{\boldsymbol{\beta}}^*, \boldsymbol{\xi} -\widehat{\boldsymbol{\xi}}^* )^\top  \pmb{\Omega}_k \VEC(\boldsymbol{\beta}_0 - \widehat{\boldsymbol{\beta}}^*, \boldsymbol{\xi} -\widehat{\boldsymbol{\xi}}^* ) 
    \nonumber\\
    &&-  4 k_0\frac{1}{n}\left\| [\mathbf{X}, \mathbf{V}]^\top [\mathbf{X}, \mathbf{V}]  - \pmb{\Omega}_k\right\|_{\max}\cdot\|\VEC(\boldsymbol{\beta}_0 - \widehat{\boldsymbol{\beta}}^*, \boldsymbol{\xi} -\widehat{\boldsymbol{\xi}}^* )\|^2\
    \nonumber \\
    &&\ge \frac{1}{2}\VEC(\boldsymbol{\beta}_0 - \widehat{\boldsymbol{\beta}}^*, \boldsymbol{\xi} -\widehat{\boldsymbol{\xi}}^* )^\top  \pmb{\Omega}_k \VEC(\boldsymbol{\beta}_0 - \widehat{\boldsymbol{\beta}}^*, \boldsymbol{\xi} -\widehat{\boldsymbol{\xi}}^* ),
\label{def.inv.rate}
\eea}

\noindent where $\|\cdot \|_{\max}$ stands for the max norm for a matrix,  the second inequality follows from \eqref{eq.lasso3}, and the third inequality follows from a development similar for \eqref{rate.e} and $k_0 \frac{\sqrt{\log k}}{\sqrt{n}}\to 0$. 

Again as explained in the beginning of this proof, we obtain 
\begin{eqnarray*} 
  &&  \frac{1}{2}\alpha \|\VEC(\boldsymbol{\beta}_0 - \widehat{\boldsymbol{\beta}}^*, \boldsymbol{\xi} -\widehat{\boldsymbol{\xi}}^* ) \|^2 \le  \frac{3}{2}\lambda\|\VEC(\boldsymbol{\beta}_0 - \widehat{\boldsymbol{\beta}}^*, \mathbf{S}_{\mathscr{C}} ( \boldsymbol{\xi} -\widehat{\boldsymbol{\xi}}^*) )\|_1\notag \\
    &&\le  \frac{3}{2}\lambda  \sqrt{k_0}\|\VEC(\boldsymbol{\beta}_0 - \widehat{\boldsymbol{\beta}}^*, \mathbf{S}_{\mathscr{C}} ( \boldsymbol{\xi} -\widehat{\boldsymbol{\xi}}^*) ) \| \le  \frac{3}{2}\lambda  \sqrt{k_0}\|\VEC(\boldsymbol{\beta}_0 - \widehat{\boldsymbol{\beta}}^*,  \boldsymbol{\xi} -\widehat{\boldsymbol{\xi}}^*) \| 
\end{eqnarray*}
where the first inequality follows from \eqref{eq.lasso2_1}. We can conclude that with probability approaching 1, $\|\VEC(\boldsymbol{\beta}_0 - \widehat{\boldsymbol{\beta}}^*, \boldsymbol{\xi} -\widehat{\boldsymbol{\xi}}^* ) \| \le \frac{3\lambda \sqrt{k_0}}{\alpha}.$

Additionally, we can obtain {\footnotesize $\|\VEC(\boldsymbol{\beta}_0 - \widehat{\boldsymbol{\beta}}^*, \boldsymbol{\xi} -\widehat{\boldsymbol{\xi}}^* )\|_1\le 4 \sqrt{k_0}\|\VEC(\boldsymbol{\beta}_0 - \widehat{\boldsymbol{\beta}}^*, \mathbf{S}_{\mathscr{C}} (\boldsymbol{\xi} -\widehat{\boldsymbol{\xi}}^*))\|\le  \frac{12\lambda k_0}{\alpha}$} with probability approaching 1, where the first inequality follows from  \eqref{eq.lasso3}, and the second inequality follows from $\|\VEC(\boldsymbol{\beta}_0 - \widehat{\boldsymbol{\beta}}^*, \boldsymbol{\xi}-\widehat{\boldsymbol{\xi}}^* ) \| \le \frac{3\lambda \sqrt{k_0}}{\alpha}$. The proof for (i)-(ii) is completed.

\smallskip

(iii). To proceed, we define a new objective function: $Q(\boldsymbol{\beta}, \boldsymbol{\xi}_k) =\frac{1}{n}\|\mathbf{y}- \mathbf{X}\boldsymbol{\beta} - \mathbf{V}\boldsymbol{\xi}_k\|^2 + \lambda \|\boldsymbol{\xi}_k\circ\pmb{\zeta}\|_1$, where $\pmb{\zeta} =(\zeta_1,\ldots, \zeta_k)^\top$. Accordingly, Step 2 can be written as $(\widehat{\pmb{\beta}}^\dag, \widehat{\pmb{\xi}}_k^\dag) =\argmin Q(\boldsymbol{\beta}, \boldsymbol{\xi}_k)$, where $\widehat{\pmb{\xi}}_k^\dag =(\widehat{\xi}_{k,1}^\dag,\ldots, \widehat{\xi}_{k,k}^\dag)^\top$.

Due to the property of convex optimization, $(\widehat{\pmb{\beta}}^\dag, \widehat{\pmb{\xi}}_k^\dag)$ is a solution if and only if there exits a subgradient

\begin{eqnarray}\label{def.subg}
    \mathbf{g} &\in &\{\VEC(\mathbf{0}_{d\times 1},\mathbf{z}) \mid \mathbf{z}\in \mathbb{R}^{k},\, \notag \\
    &&z_\ell = \sgn(\widehat{\xi}_{k,\ell}^\dag) \zeta_\ell \text{ for }\widehat{\xi}_{k,\ell}^\dag\ne 0, \text{ and } |z_\ell|\le \zeta_\ell \text{ elsewhere, for }\ell\in [k] \}
\end{eqnarray}
such that

\begin{eqnarray}\label{def.FOC}
    \mathbf{0} = \frac{1}{n} (\mathbf{X}, \mathbf{V})^\top (\mathbf{X}, \mathbf{V})\VEC(\widehat{\pmb{\beta}}^\dag, \widehat{\pmb{\xi}}_k^\dag) - \frac{1}{n} (\mathbf{X}, \mathbf{V})^\top \mathbf{y}+ \frac{\lambda}{2} \mathbf{g}.
\end{eqnarray}

We now partition $\mathbf{V}$, $\mathbf{g}$, $\pmb{\xi}$ and $\widehat{\pmb{\xi}}_k^\dag$ based on the sets $\mathscr{C}$ and $\bar{\mathscr{C}}$. Specifically, let $\mathbf{V}_{\mathscr{C}}$, $\mathbf{g}_{\mathscr{C}}$, $\widehat{\pmb{\xi}}_{k,\mathscr{C}}^\dag$ and $\pmb{\xi}_{\mathscr{C}}$ all include the columns/elements corresponding to the set $\mathscr{C}$, and let $\mathbf{V}_{\bar{\mathscr{C}}}$ and $\mathbf{g}_{\bar{\mathscr{C}}}$ include the columns/elements corresponding to the set $\bar{\mathscr{C}}$. Here we let $\mathbf{g}_{\mathscr{C}}$ include the first $d$ 0's in $\mathbf{g}$.

We now proceed. By \eqref{def.FOC}, $\sgn(\widehat{\pmb{\xi}}_k^\dag)=\sgn(\pmb{\xi})$ holds if and only if the following Karush-Kuhn-Tucker conditions hold:

\begin{eqnarray}
    \mathbf{0}&=& \frac{1}{n}(\mathbf{X}, \mathbf{V}_{\mathscr{C}})^\top (\mathbf{X}, \mathbf{V}_{\mathscr{C}}) \VEC(\widehat{\pmb{\beta}}^\dag -\pmb{\beta}_0, \widehat{\pmb{\xi}}_{k,\mathscr{C}}^\dag-\pmb{\xi}_{\mathscr{C}})\notag\\
    &&-\frac{1}{n} (\mathbf{X}, \mathbf{V}_{\mathscr{C}})^\top (\pmb{\delta}+\mathbf{e}) +\frac{\lambda}{2}\mathbf{g}_{\mathscr{C}},\label{def.L.rate1} \\
    \mathbf{0}&=& \frac{1}{n}\mathbf{V}_{\bar{\mathscr{C}}}^\top (\mathbf{X}, \mathbf{V}_{\mathscr{C}}) \VEC(\widehat{\pmb{\beta}}^\dag -\pmb{\beta}_0, \widehat{\pmb{\xi}}_{k,\mathscr{C}}^\dag-\pmb{\xi}_{\mathscr{C}}) -\frac{1}{n} \mathbf{V}_{\bar{\mathscr{C}}}^\top (\pmb{\delta}+\mathbf{e}) +\frac{\lambda}{2}\mathbf{g}_{\bar{\mathscr{C}}},\label{def.L.rate2}
\end{eqnarray}
and 
\begin{eqnarray}\label{def.L.rate3} 
    \sgn(\xi_{\ell})(\widehat{\xi}_{k, \ell}^\dag -\xi_{\ell})>-|\xi_{\ell}|\quad\text{for}\quad\ell \in \mathscr{C}.
\end{eqnarray}

We justify \eqref{def.L.rate1} and \eqref{def.L.rate3} first. Define $\pmb{\Omega}_{\mathscr{C}} =\begin{pmatrix}
        \mathbf{A}_{11} & \mathbf{A}_{12} \\
        \mathbf{A}_{12}^\top & \mathbf{I}_k \\
    \end{pmatrix}$ and $\mathbf{A}_{12}\coloneqq (\pmb{\eta}_1,\cdots , \pmb{\eta}_{k_0})$, where $\mathbf{A}_{11}=E[\mathbf{x}\mathbf{x}^\top]$ and $\pmb{\eta}_\ell=E[\mathbf{x}\psi_\ell(\mathbf{x})]$. By the development of $\pmb{\Omega}_k$, it is easy to know that $\pmb{\Omega}_{\mathscr{C}}$ is invertible and $\rho_{\min}(\pmb{\Omega}_{\mathscr{C}})>0$. 

Note that  we have shown that 

\begin{eqnarray*}
    &&\left|\rho_{\min}(\frac{1}{n}(\mathbf{X}, \mathbf{V}_{\mathscr{C}})^\top (\mathbf{X}, \mathbf{V}_{\mathscr{C}})) -\rho_{\min}(\pmb{\Omega}_{\mathscr{C}} )\right| \notag \\
    &\le & (d+k_0) \left\| \frac{1}{n}(\mathbf{X}, \mathbf{V}_{\mathscr{C}})^\top (\mathbf{X}, \mathbf{V}_{\mathscr{C}})- \pmb{\Omega}_{\mathscr{C}}  \right\|_{\max}\lesssim k_0\frac{\sqrt{\log k}}{\sqrt{n}} \asymp k_0\lambda,
\end{eqnarray*}
where $\rho_{\min}(\cdot)$ stands for the minimum eigenvalue, the second step follow from the development of \eqref{def.inv.rate}, and the last step follows from $\lambda\asymp \frac{\sqrt{\log k}}{\sqrt{n}} $. Thus, we can write

\begin{eqnarray}\label{def.rate.xv}
    &&\left\| \left(\frac{1}{n}(\mathbf{X}, \mathbf{V}_{\mathscr{C}})^\top (\mathbf{X}, \mathbf{V}_{\mathscr{C}})\right)^{-1} \right\|_\infty\le \sqrt{k_0+d}\left\|  \left(\frac{1}{n}(\mathbf{X}, \mathbf{V}_{\mathscr{C}})^\top (\mathbf{X}, \mathbf{V}_{\mathscr{C}})\right)^{-1}\right\|_2\notag \\
    &\le & \sqrt{k_0+d}/\rho_{\min}(\frac{1}{n}(\mathbf{X}, \mathbf{V}_{\mathscr{C}})^\top (\mathbf{X}, \mathbf{V}_{\mathscr{C}}))  \lesssim \sqrt{k_0}.
\end{eqnarray}

From \eqref{def.L.rate1}, we obtain the following expression

\begin{eqnarray}\label{def.rate.b}
     &&\VEC(\widehat{\pmb{\beta}}^\dag -\pmb{\beta}_0, \widehat{\pmb{\xi}}_{k,\mathscr{C}}^\dag-\pmb{\xi}_{\mathscr{C}})\notag \\
     &=&\left(\frac{1}{n}(\mathbf{X}, \mathbf{V}_{\mathscr{C}})^\top (\mathbf{X}, \mathbf{V}_{\mathscr{C}})\right)^{-1}\left(\frac{1}{n} (\mathbf{X}, \mathbf{V}_{\mathscr{C}})^\top (\pmb{\delta}+\mathbf{e}) -\frac{\lambda}{2}\mathbf{g}_{\mathscr{C}}\right).
\end{eqnarray}

The right hand side of \eqref{def.rate.b} can be further written as follows:
\begin{eqnarray}
    &&\left\|\left(\frac{1}{n}(\mathbf{X}, \mathbf{V}_{\mathscr{C}})^\top (\mathbf{X}, \mathbf{V}_{\mathscr{C}})\right)^{-1}\left(\frac{1}{n} (\mathbf{X}, \mathbf{V}_{\mathscr{C}})^\top (\pmb{\delta}+\mathbf{e}) -\frac{\lambda}{2}\mathbf{g}_{\mathscr{C}}\right)\right\|_\infty\notag \\
    &\lesssim &\left\| \left(\frac{1}{n}(\mathbf{X}, \mathbf{V}_{\mathscr{C}})^\top (\mathbf{X}, \mathbf{V}_{\mathscr{C}})\right)^{-1} \right\|_\infty\left(\left\|\frac{1}{n} (\mathbf{X}, \mathbf{V}_{\mathscr{C}})^\top (\pmb{\delta}+\mathbf{e}) \right\|_\infty + \frac{\lambda}{2}\max_{\ell\in \mathscr{C}} \zeta_{\ell}\right)\notag \\
    &\lesssim & \sqrt{k_0}\left(\frac{\sqrt{\log k}}{\sqrt{n}} + \frac{\lambda}{2}\max_{\ell\in \mathscr{C}} \zeta_{\ell}\right) \lesssim \lambda\sqrt{k_0} (1+  \max_{\ell\in \mathscr{C}} \zeta_{\ell} ),
\end{eqnarray}
where the second step follows from \eqref{def.rate.xv} and \eqref{rate.ed}. Thus, \eqref{def.L.rate3} is satisfied provided $\min_{\ell \in \mathscr{C}} |\xi_{\ell}| \gg \lambda\sqrt{k_0} (1+  \max_{\ell\in \mathscr{C}} \zeta_{\ell} ).$

We now turn to \eqref{def.L.rate2}. Combining \eqref{def.L.rate2} and \eqref{def.rate.b}, we shall show that for $\forall\ell \in \bar{\mathscr{C}}$

\begin{eqnarray}\label{def.rate.cbar}
    \frac{\lambda}{2}\zeta_{\ell} &\ge &\left|\frac{1}{n}\mathbf{V}_{\ell}^\top (\mathbf{X}, \mathbf{V}_{\mathscr{C}}) \left(\frac{1}{n}(\mathbf{X}, \mathbf{V}_{\mathscr{C}})^\top (\mathbf{X}, \mathbf{V}_{\mathscr{C}})\right)^{-1}\left(\frac{1}{n} (\mathbf{X}, \mathbf{V}_{\mathscr{C}})^\top (\pmb{\delta}+\mathbf{e}) -\frac{\lambda}{2}\mathbf{g}_{\mathscr{C}}\right) \right| \notag \\
    &&+\left| \frac{1}{n} \mathbf{V}_{\ell}^\top (\pmb{\delta}+\mathbf{e})\right|,
\end{eqnarray}
where $\mathbf{V}_{\ell}$ stands for the $\ell^{th}$ column of $\mathbf{V}$. Note that for $\forall\ell \in \bar{\mathscr{C}}$
\begin{eqnarray*}
    &&\left|\frac{1}{n}\mathbf{V}_{ \ell}^\top (\mathbf{X}, \mathbf{V}_{\mathscr{C}}) \left(\frac{1}{n}(\mathbf{X}, \mathbf{V}_{\mathscr{C}})^\top (\mathbf{X}, \mathbf{V}_{\mathscr{C}})\right)^{-1}\left(\frac{1}{n} (\mathbf{X}, \mathbf{V}_{\mathscr{C}})^\top (\pmb{\delta}+\mathbf{e}) -\frac{\lambda}{2}\mathbf{g}_{\mathscr{C}}\right) \right| \notag \\
    &\le &\sqrt{k_0}\left\| \frac{1}{n}(\mathbf{X}, \mathbf{V}_{\mathscr{C}})^\top (\mathbf{X}, \mathbf{V}_{\mathscr{C}}) \right\|_{\max} \notag \\
    &&\cdot \left\| \left(\frac{1}{n}(\mathbf{X}, \mathbf{V}_{\mathscr{C}})^\top (\mathbf{X}, \mathbf{V}_{\mathscr{C}})\right)^{-1}\left(\frac{1}{n} (\mathbf{X}, \mathbf{V}_{\mathscr{C}})^\top (\pmb{\delta}+\mathbf{e}) -\frac{\lambda}{2}\mathbf{g}_{\mathscr{C}}\right) \right\|_{\infty}\notag \\
    &\lesssim & \sqrt{k_0}\|\diag\{\e[\mathbf{x}_i\mathbf{x}_i^\top], \mathbf{I}_{k_0} \}\|_{\max}\lambda\sqrt{k_0} (1+  \max_{\ell\in \mathscr{C}} \zeta_{\ell} )\lesssim \lambda k_0 (1+  \max_{\ell\in \mathscr{C}} \zeta_{\ell} )
\end{eqnarray*}
and $\left| \frac{1}{n} \mathbf{V}_{\ell}^\top (\pmb{\delta}+\mathbf{e})\right|\lesssim \frac{\sqrt{\log(k)}}{\sqrt{n}}\asymp \lambda$. Thus, to show \eqref{def.rate.cbar}, we require $\min_{\ell \in \bar{\mathscr{C}}} \zeta_{\ell}\gg k_0 \left(1+  \max_{\ell\in \mathscr{C}} \zeta_{\ell}\right)$. 

Putting everything together, we have shown that $\Pr(\sgn(\widehat{\pmb{\xi}}_k^{\dag}) =\sgn(\pmb{\xi}))\to 1$, which implies $\Pr(\widehat{\mathcal{S}}_{(k)}=\mathcal{S}_{(k)})\rightarrow 1$ as $(n, k)\rightarrow (\infty, \infty)$. The proof of (iii) is therefore completed.
\smallskip

(iv). Note that \eqref{def.L.rate1} requires 
\begin{eqnarray*}
    \mathbf{0}&=& \frac{1}{n}(\mathbf{X}, \mathbf{V}_{\mathscr{C}})^\top (\mathbf{X}, \mathbf{V}_{\mathscr{C}}) \VEC(\widehat{\pmb{\beta}}^\dag -\pmb{\beta}_0, \widehat{\pmb{\xi}}_{k,\mathscr{C}}^\dag-\pmb{\xi}_{\mathscr{C}})-\frac{1}{n} (\mathbf{X}, \mathbf{V}_{\mathscr{C}})^\top (\pmb{\delta}+\mathbf{e}) +\frac{\lambda}{2}\mathbf{g}_{\mathscr{C}},
\end{eqnarray*}
which can also be written as

\begin{eqnarray*}
    \mathbf{0} =  (\mathbf{X}, \mathbf{V}_{\mathscr{C}})^\top (\mathbf{Y}-(\mathbf{X}, \mathbf{V}_{\mathscr{C}}) \VEC(\widehat{\pmb{\beta}}^\dag , \widehat{\pmb{\xi}}_{k,\mathscr{C}}^\dag )) -\frac{n\lambda}{2}\mathbf{g}_{\mathscr{C}}.
\end{eqnarray*}

Noting that the first $d$ rows of $\mathbf{g}_{\mathscr{C}}$ are 0's and solving the above equation carefully yields that
\begin{eqnarray*}
    \widehat{\pmb{\beta}}^\dag &=&\left(\frac{1}{n}\mathbf{X}^\top\mathbf{M}_{\mathbf{V}_{\mathscr{C}}}\mathbf{X}\right)^{-1} \, \left(\frac{1}{n}\mathbf{X}^\top\mathbf{M}_{\mathbf{V}_{\mathscr{C}}}\mathbf{Y} - \frac{\lambda}{2n}\mathbf{X}^\top \mathbf{V}_{\mathscr{C}} \left(\frac{1}{n}\mathbf{V}_{\mathscr{C}}^\top \mathbf{V}_{\mathscr{C}}\right)^{-1}\mathbf{g}_{\mathscr{C}}^\dag\right),
\end{eqnarray*}
where $\mathbf{g}_{\mathscr{C}}^\dag$ includes from the $(d+1)^{th}$ element of $\mathbf{g}_{\mathscr{C}}$ to the $(d+k_0)^{th}$ element of $\mathbf{g}_{\mathscr{C}}$. Furthermore,
\begin{eqnarray*}
    \widehat{\pmb{\beta}}^\dag &=&\left(\frac{1}{n}\mathbf{X}^\top\mathbf{M}_{\mathbf{V}_{\mathscr{C}}}\mathbf{X}\right)^{-1} \, \left(\frac{1}{n}\mathbf{X}^\top\mathbf{M}_{\mathbf{V}_{\mathscr{C}}}\mathbf{Y} - \frac{\lambda}{2n}\mathbf{X}^\top \mathbf{V}_{\mathscr{C}} \mathbf{g}_{\mathscr{C}}^\dag\right)(1+o_p(1)).
\end{eqnarray*}
By taking $\zeta_j =1/(\widehat{\xi}_j^*\cdot \log k)$ for example, we can further obtain that 

\begin{eqnarray*}
    \widehat{\pmb{\beta}}^\dag &=&\left(\frac{1}{n}\mathbf{X}^\top\mathbf{M}_{\mathbf{V}_{\mathscr{C}}}\mathbf{X}\right)^{-1}\, \left(\frac{1}{n}\mathbf{X}^\top\mathbf{M}_{\mathbf{V}_{\mathscr{C}}}\mathbf{Y}\right)+o_P(1/\sqrt{n}).
\end{eqnarray*}

From there, the proof then follows from that of Theorem 3.1.

\end{proof}

\end{appendix}

}

\end{document}